\documentclass[11pt]{article}

\usepackage{arxiv}
\usepackage[utf8]{inputenc} 
\usepackage[T1]{fontenc}    
\usepackage[breaklinks, colorlinks=true, citecolor=black, urlcolor=black, linkcolor=black]{hyperref}
\usepackage{url}            
\usepackage{booktabs}       
\usepackage{amsfonts}       
\usepackage{nicefrac}       
\usepackage[final,nopatch=footnote]{microtype}
\usepackage{graphicx}
\usepackage{natbib}
\usepackage{doi}
\usepackage{amsmath,amssymb,amsthm,mathrsfs,bbm,amsfonts,bm}
\usepackage{color,graphicx,cases,caption,enumerate}
\usepackage{float,multicol}
\usepackage[dvipsnames]{xcolor}
\usepackage{colortbl}
\usepackage{booktabs}
\usepackage{multirow}
\usepackage{setspace}
\usepackage{hyperref}
\usepackage{capt-of}
\usepackage{array}
\usepackage[title]{appendix}
\newcolumntype{C}[1]{>{\centering\let\newline\\\arraybackslash\hspace{0pt}}m{#1}}
\usepackage{natbib}
\usepackage{placeins} 
\usepackage[breaklinks, colorlinks=true, citecolor=black, urlcolor=black, linkcolor=black]{hyperref}
\usepackage{subcaption}
\usepackage{tikz}
\usetikzlibrary{fit,positioning, arrows.meta}
\usepackage{algorithm}
\usepackage{algpseudocode}
\usepackage{amsthm}
\newtheorem{assumption}{Assumption}
\newtheorem{definition}{Definition}[section]
\newtheorem{theorem}[definition]{Theorem}
\newtheorem{corollary}[definition]{Corollary}
\newtheorem{lemma}[definition]{Lemma} 
\newtheorem{proposition}[definition]{Proposition}

\newcommand{\Ree}{\ensuremath{\mathbb{R}}}

\newcommand{\p}{\ensuremath{\mathbb{P}}}

\newcommand{\e}{\ensuremath{\mathbb{E}}}

\title{Identification and Inference for Causal Effects in Extremes under General Conditions}

\author{
  Lisa Leimenstoll\textsuperscript{a,*} \qquad
  Melanie Schienle\textsuperscript{a,b} \\[0.5em]
  \textsuperscript{a}Karlsruhe Institute of Technology (KIT), Institute of Statistics (STAT),\\
  Blücherstr.~17, 76185 Karlsruhe, Germany \\[0.5em]
  \textsuperscript{b} Heidelberg Institute for Theoretical Studies (HITS), Computational Statistics (CST)\\
  Schloss-Wolfsbrunnenweg ~35, 69118 Heidelberg, Germany \\[0.5em]
  \textsuperscript{*}Corresponding author: \texttt{lisa.leimenstoll@kit.edu}
}

\renewcommand{\headeright}{}
\renewcommand{\undertitle}{}
\renewcommand{\headerleft}{Identification and Inference for Causal Effects in Extremes under General Conditions}
\renewcommand{\shorttitle}{}

\usepackage{xcolor}

\begin{document}
\maketitle

\begin{abstract}
\noindent
Understanding the propagation of extreme events is important in many economic and environmental applications, yet most econometric methods for causal inference focus on average effects rather than tail behavior. This paper studies the identification of causal relations in extremes and derives resulting estimators and their asymptotic inference. As measure of causal dependence between extreme realizations of variables, we analyze the asymptotic behavior of the Causal Tail Coefficient (CTC) within a linear structural causal model with heavy-tailed regularly varying innovations. In contrast to the existing literature, we allow the variables in the system to exhibit heterogeneous tail indices and consider the presence of potentially heavy-tailed confounders. We derive theoretical results assessing the limiting behavior of the CTC under these conditions and show how differences in tail behavior can help to reach identification of the causal structure. Light-tailed confounders are asymptotically negligible, but sufficiently heavy-tailed confounders can induce extremal dependence patterns that are observationally indistinguishable from direct causal effects. When suitable proxy information is available, identification can be recovered using an adjusted Causal Tail Coefficient. Based on these results, we develop estimation and inference procedures for causal relations in extremes under general conditions. We establish asymptotic properties of the proposed estimators and derive tests for the causal direction and heavy-tailed confounding. Simulation evidence examines their finite-sample performance and provides guidance on their implementation. Applications to climate and financial extremes illustrate how the proposed methods can uncover causal relations that may remain undetected by approaches targeting average dependence.
\end{abstract}


\section{Introduction}
Identification of causal effects from observational data is one of the central objectives of econometrics. Over the past decades, substantial progress has been made in understanding when causal effects are identifiable from structural assumptions and how they can be estimated consistently. Most existing identification theory, however, is developed for average effects or relies on information from the bulk of the distribution. In many economic applications, however, the phenomena of interest occur only during rare and extreme events, where the assumptions underlying classical methods may no longer be informative.
Examples include financial crises, systemic risk, climate disasters, infrastructure failures, and supply-chain disruptions. In these settings, the relevant causal mechanisms are often activated only by unusually large shocks and impacts particularly matter in the extremes. Consequently, understanding how extreme events propagate through an economic system requires identification methods that explicitly exploit information contained in the tails of the distribution rather than in average behavior.

This paper studies identification of causal effects in the tails of the distribution within linear structural causal models defined on directed acyclic graphs. We consider observational settings in which variables follow heavy-tailed distributions, i.e., extreme values occur with higher probability than under a normal distribution, as is typically the case for financial returns. However, the variables may have potentially different tail indices and unobserved confounding may be present. 
While graphical causal models have become an important framework for formalizing identification assumptions in econometrics, existing results for causal inference in extremes are scarce and so far rely on restrictive assumptions, in particular all variables should exhibit comparable tail behavior. Such assumptions are often violated in empirical applications, where different economic variables can exhibit substantially different probabilities of extreme realizations.

Our first contribution is to characterize identification of causal effects under heterogeneous tail behavior. We derive the limiting behavior of the Causal Tail Coefficient (CTC) \citep{gnecco2021causal} when innovations are regularly varying with different tail indices. Heterogeneity in the tails is prevalent in practice (see e.g. \autoref{fig:histograms}) and the resulting theory demonstrates that heterogeneity in tail behavior is not merely a technical complication but constitutes an additional source of identifying information. Depending on the ordering of the tail indices, heterogeneous tails can either strengthen or weaken identification of causal directions, yielding identification results that are fundamentally different from the equal-tail setting considered in the existing literature.
Our second contribution concerns confounding. We distinguish between light-tailed and heavy-tailed unobserved confounders and show that they affect identification in fundamentally different ways. Whereas light-tailed confounders become asymptotically negligible for tail identification, sufficiently heavy-tailed confounders can generate extremal dependence patterns that are observationally indistinguishable from direct causal effects. We further show that identification can be recovered whenever suitable proxy information for the confounder is available through an adjusted version of the Causal Tail Coefficient.
Our third contribution is on inference. Building on the identification theory, we develop estimation and inference procedures for causal effects in extremes under general conditions. We derive asymptotic properties of the proposed estimators, establish hypothesis tests for causal direction and heavy-tailed confounding, and investigate their finite-sample performance through extensive simulation experiments.

\begin{figure}[ht]
    \centering
    \begin{subfigure}{0.33\textwidth}
        \centering
        \includegraphics[width=\textwidth]{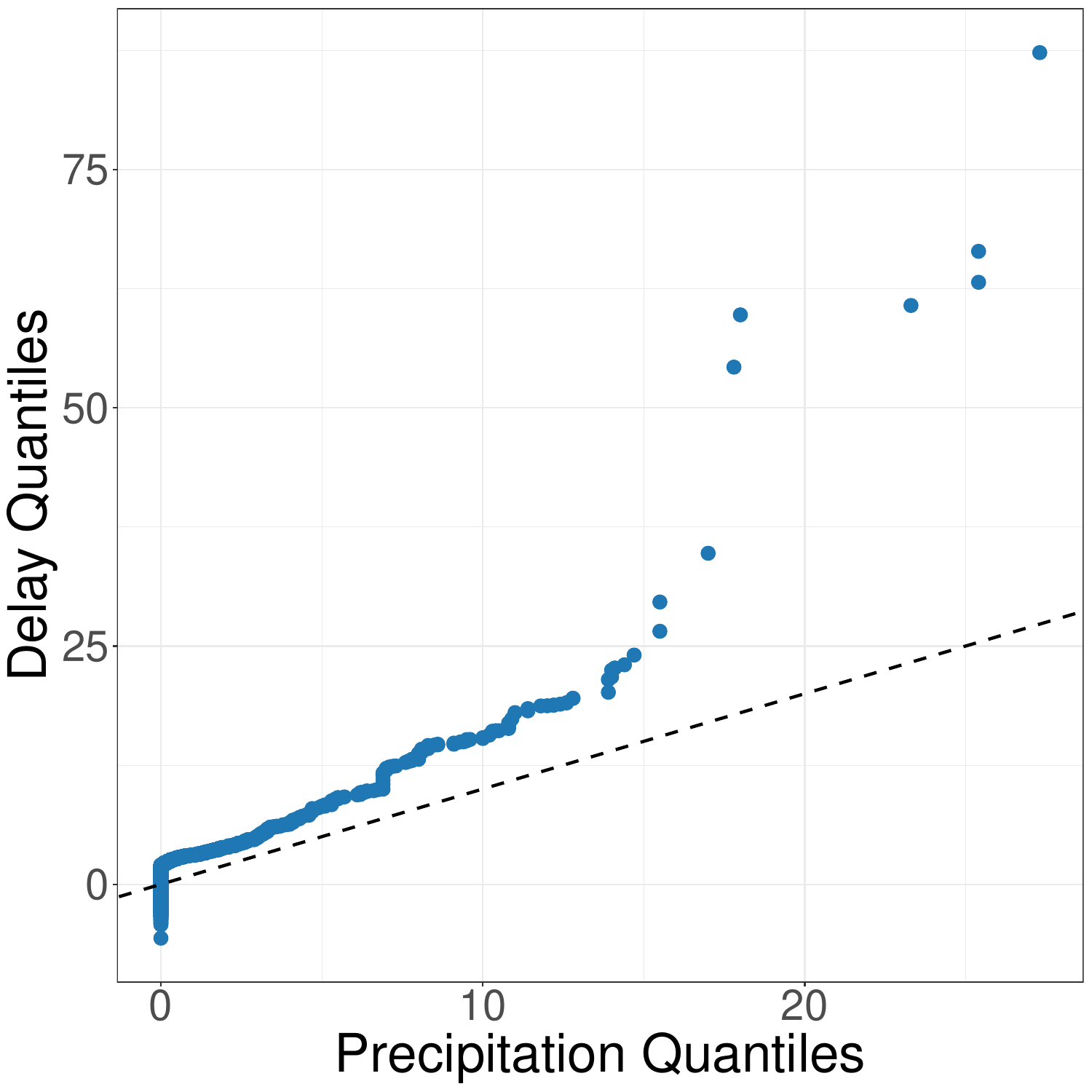}
        \caption{Precipitation and Train Delays}
    \end{subfigure}
    \hfill
    \begin{subfigure}{0.33\textwidth}
        \centering
        \includegraphics[width=\textwidth]{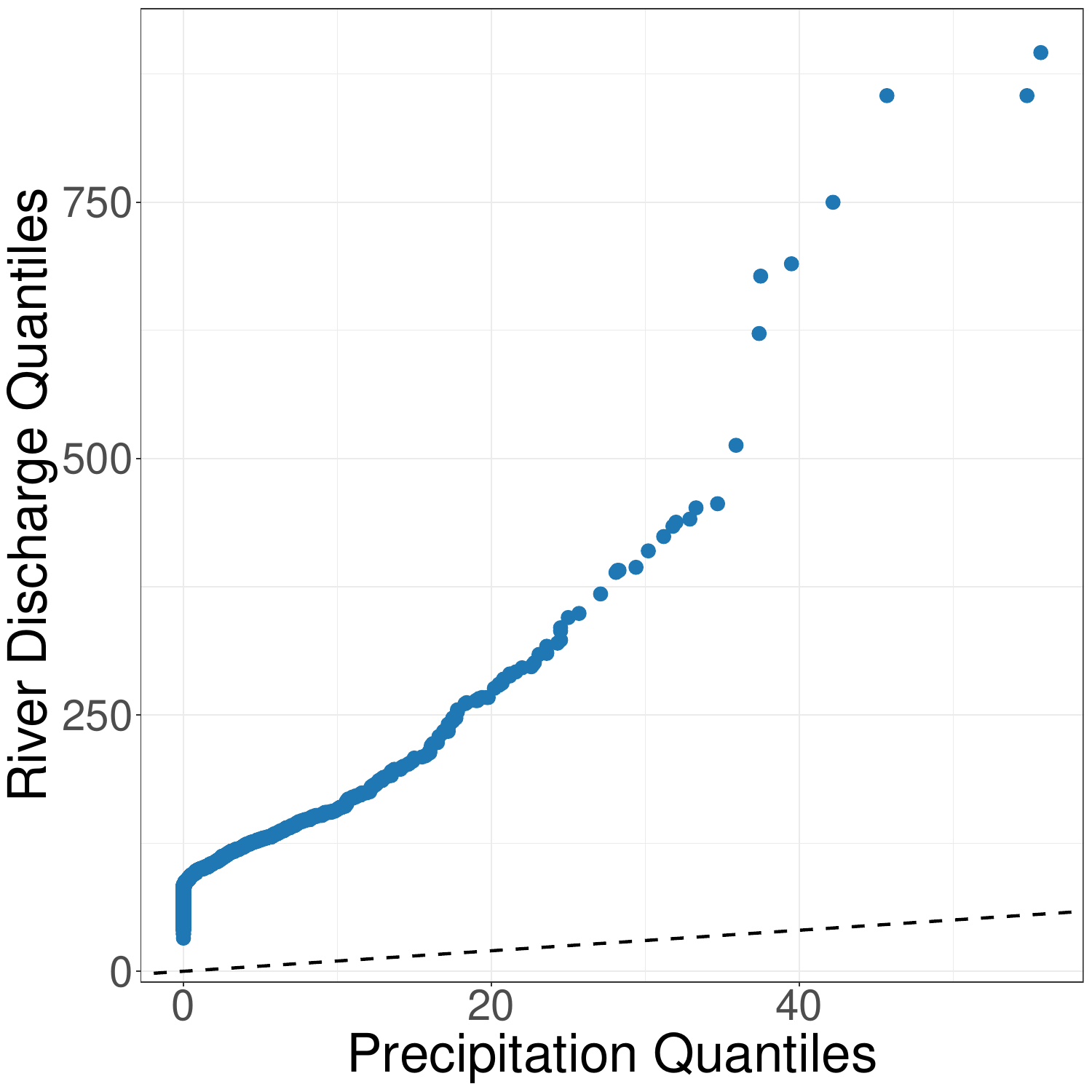}
        \caption{Precipitation and River Flows}
    \end{subfigure}
    \hfill
    \begin{subfigure}{0.33\textwidth}
        \centering
        \includegraphics[width=\textwidth]{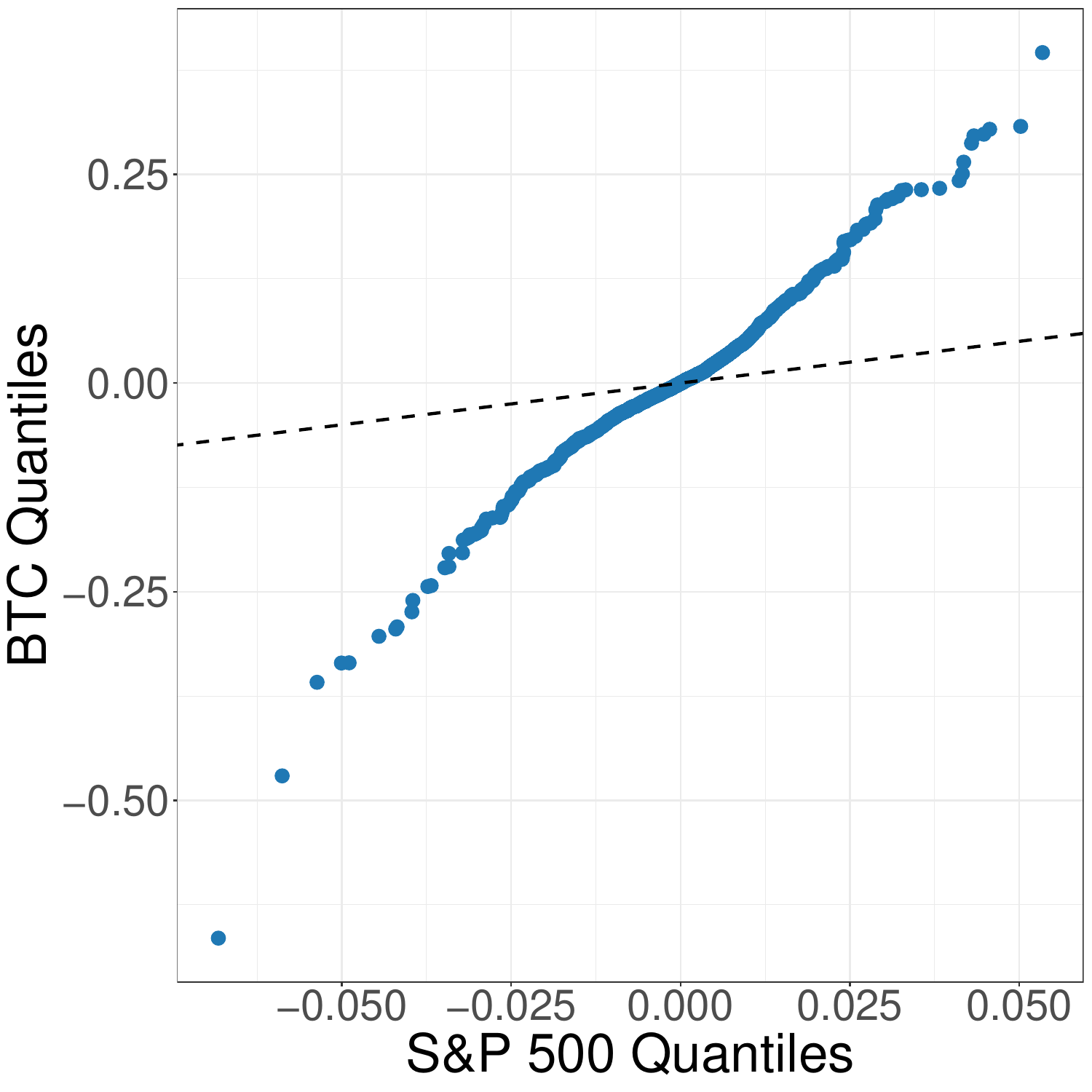}
        \caption{S\&P 500 and Bitcoin}
    \end{subfigure}
    \caption{Quantile-Quantile plots for the three empirical applications in climate and finance. For detailed descriptions of the data we refer to Section~\ref{sec:application}.}
    \label{fig:histograms}
\end{figure}

Finally, we demonstrate the practical relevance of the proposed methodology using applications from climate economics and financial markets. These applications illustrate that economically meaningful causal relations may become identifiable through information contained in extreme observations even when conventional methods focusing on average dependence fail to detect them. In particular, we find a causal impact of heavy precipitation on train delays and also find evidence of a causal impact of extreme precipitation on river discharge and investigate the temporal delay in flood responses. In studying the relation of extremes of the S\&P 500 and Bitcoin returns, we detect causal directions in the positive and negative tails that might be of interest for investors and regulators.

Our work contributes to three strands of the literature. First, we contribute to the literature on causal inference using structural and graphical models. In a structural causal framework, identifying causal relations corresponds to identifying the directed acyclic graph (DAG) underlying the data-generating process. Graphical causal models have become increasingly influential in econometrics as tools for formalizing identification assumptions and reasoning about confounding and causal propagation \citep[see, e.g.,][]{CorSan2024, HuntKlei2022, imb2020}. Our approach adopts this structural perspective and studies how causal relations can be identified when attention is focused on extreme realizations of the variables in the general case of heterogeneous tails.
Second, our work contributes to the growing literature on extreme value econometrics and tail dependence modeling. A large body of work studies dependence between extremes using copulas, multivariate extreme value models, or conditional tail dependence measures \citep[see, e.g.,][]{patton2006modelling, bormann2020detecting, oh2017modeling, jiang2026efficient, bucher2015nonparametric}. For example, \cite{nikoloulopoulos2012vine} assess asymmetric tail dependence in financial returns using vine copula models, while \cite{mhalla2019exceedance} develop a generalized additive modeling framework for estimating covariate-dependent tail dependence in multivariate extremes and demonstrate its effectiveness in environmental applications. Extensions to non-stationary environments include \cite{murphy2024modelling}, who model evolving joint tail risks in environmental data. These approaches provide flexible tools for modeling extremal dependence but typically focus on association rather than causal direction, and often require strong assumptions about the joint tail structure.
Third, our paper contributes to the emerging literature on causal inference for extreme events. The Causal Tail Coefficient (CTC) introduced by \cite{gnecco2021causal} provides a framework for measuring causal propagation between extreme realizations in linear structural causal models. Recent work has extended this framework in several directions. For example, \cite{pasche2023causal} study the role of observable heavy-tailed confounders, while \cite{bodik2024causality} and \cite{bodik2024granger} develop extensions to time-series settings and extreme-event Granger causality. Alternative approaches to causal inference in extremes have also been proposed. For instance, \cite{tran2024estimating} estimate directed tree structures for extremes using max-linear causal graphical models, and \cite{mhalla2020causal} develop a causal discovery method based on the Kolmogorov complexity of extreme conditional quantiles. Despite this growing literature, existing approaches typically rely on restrictive assumptions on tail behavior, such as equal tail indices across variables or specific joint tail structures. In many empirical applications, however, variables exhibit significantly different tail behavior, and heavy-tailed confounders may distort the observed extremal dependence. Modeling the joint tail structure based on a small subset of extreme observations is therefore challenging and may lead to misleading causal conclusions.

The remainder of the paper proceeds as follows. Section~\ref{sec:theory} introduces the structural framework and Section~\ref{sec:CTC} derives the asymptotic properties of the CTC when tail indices differ. Section~\ref{sec:device} develops the proposed testing strategy and Section~\ref{sec:simulation} investigates its finite-sample properties. Section~\ref{sec:application} presents three empirical applications in climate and financial markets. Section~\ref{sec:conclusion} concludes. All proofs are contained in the appendix. All replication material for the simulation study and the empirical applications is available at \url{https://github.com/KITmetricslab/Tail-Causality}.

\section{Structural Model and Identification Problem}
\label{sec:theory}

This section introduces the structural framework used to analyze causal relations between extreme realizations of variables. We consider linear structural causal models associated with directed acyclic graphs and study the behavior of the Causal Tail Coefficient when the involved variables exhibit heavy-tailed distributions with possibly different tail indices.

\subsection{Linear Structural Causal Model for Extremes}

As we are interested in causal relations for extremes, we assume that variables have heavy-tailed distributions where the probability of occurrence of extremes is not too small. For such variables, it is natural to classify their tail behavior via {regular variation}, which provides a flexible framework for modeling extremal dependence both marginally and jointly. A random variable $X$ is regularly varying with index $\alpha > 0$ ($X\in RV(\alpha))$ if
\begin{equation}
\label{eq:reg_var}
    \mathbb{P}(X>x)\sim \ell(x)x^{-\alpha} \textit{ for } x\rightarrow \infty, 
\end{equation}
where $\ell$ is called slowly varying \citep{resnick2007heavy} if for all $t > 0$, it is positive measurable and $\ell(tx)/\ell(x) \to 1$ as $x \to \infty$. We write $f \sim g$ if $f(x)/g(x) \to 1$ as $x \to \infty$. When working with (several) heavy-tailed random variables, it is important to categorize them into classes of regular variation that capture the respective tail behavior. We say that two independent random variables $\epsilon_1$ and $\epsilon_2$ have {comparable upper tails}, if there exist $c_1$ and $c_2$, $\alpha>0$, so that $\mathbb{P}(\epsilon_j>x)\sim c_j\ell(x)x^{-\alpha}$ for $x\rightarrow \infty$, $j=1,2$. We refer to this index $\alpha$ as the {tail index}. Examples of regularly varying variables include Pareto distributed variables, as well as Student's $t$, where the degree of freedom corresponds to the regular variation index, and Cauchy distributed variables. 

If $X$ is regularly varying with tail index $\alpha>0$, then $X$ lies in the so-called Fréchet max-domain of attraction. Consequently, there exists a positive normalizing function $\tau_u$ such that the conditional distribution of exceedances above a high threshold satisfies 
\begin{equation}
\label{eq:GPDas}
    \lim_{u\to\infty}
    \p\left(
        \frac{X-u}{\tau_u}\leq x
        \,\Big|\, X>u
    \right)
    =
    G(x; (1,\xi))
    \qquad x\geq 0,
\end{equation}
where $G(\cdot;(1,\xi))$ denotes a generalized Pareto distribution (GPD) with unit scale and shape parameter $\xi$. More generally, the GPD with scale parameter $\tau>0$ and shape parameter $\xi\in\Ree$ is given by
\begin{equation}
\label{eq:GPD}
    G(x;(\tau,\xi))
    =
    1-\left(1+\frac{\xi x}{\tau}\right)_+^{-1/\xi},
    \qquad x\geq0,
\end{equation}
where we use the compact notation $z_+=\max(0,z)$ for any $z$. For regularly varying $X$ we have $\xi = 1/\alpha > 0$. The shape parameter $\xi$ governs the tail behavior of the distribution \citep{beirlant2005estimation}. In particular, $\xi>0$ corresponds to a heavy tail marked by a comparatively slow polynomial tail decay in its distribution. When estimating the distribution in \eqref{eq:GPD}, a GPD is often fitted to the exceedances of a sufficiently high threshold $u$, where the choice of the threshold matters. This corresponds to the peaks-over-threshold approach (POT) (\cite{pickands1975statistical}, \cite{balkema1974residual}). See also Section~\ref{subsec:tail_index_sim} for details.

\bigskip
For modeling causal relations in extremes, the setting of regularly varying functions and the corresponding multivariate Pareto limits of exceedances is the natural key assumption. Moreover, it has also been shown statistically that a notion of conditional independence tailored to multivariate Pareto laws in fact is a sensible basis for a corresponding graphical modeling framework for multivariate extremes \citep{engelke2020graphical, engelke2022structure}. This creates a direct link between graphical models and extreme value statistics and shows that familiar graph-based dependence concepts in the mean of the distribution under Gaussian distributional assumptions can be carried over to tail events under regular variation and respective GPD limits.
Thus, let $G=(V,E)$ denote a directed acyclic graph (DAG) where the set of nodes 
$V=\{1,\ldots ,p \}$ corresponds to random variables $X_1,\ldots,X_p$ and directed edges $E \subset V \times V$ 
represent causal directions between variables. The data-generating process is described by a linear structural causal model (LSCM)
\begin{equation}\label{eq:LSCM}
            X_j:=\sum_{k\in pa(j,G)} \beta_{kj}X_k+\epsilon_j,\, j\in V
\end{equation}
where $pa(j,G)\subseteq V=\{1,\ldots,p\}$ denotes the parents of node $j$ in the graph marking all $X_k$ with a directed edge from $k$ to $j$. The innovation terms $\epsilon_j$ are assumed to be mutually independent in $RV(\alpha_j)$, and represent exogenous shocks. The parameters $\beta_{kj} \in \Ree \setminus\{0\}$ are structural coefficients in the extremal causal model.   
For an acyclic $G$, recursive substitution of the structural equations \eqref{eq:LSCM} yields
\begin{equation}
    X_j:=\sum_{k\in An(j,G)} \beta_{k\rightarrow j}\epsilon_k,\, j\in V,
\end{equation}
where $An(j,G)$ denotes the set of ancestors of node $j$ (including $j$ itself) that contains all $X_k$ with some directed path from $k$ to $j$ and $\beta_{k \to j}$ denotes the total effect of the noise term $\epsilon_k$ on node $j$. We further define $A_{ij}:=An(i,G)\cap An(j,G)$.

In this framework, directed paths in the graph of \eqref{eq:LSCM} correspond to causal propagation mechanisms \citep{gnecco2021causal} even in the presence of some confounder. In particular, if there exists a directed path from $X_k$ to $X_j$, extreme realizations of $X_k$ may propagate through the system and generate extreme outcomes in $X_j$. For the remainder of this work we assume $\beta_{ij}\geq 0$, but all results extend also to the case of negative weights.

\subsection{The Causal Tail Coefficient } 
To quantify causal propagation between extremes, we build on the Causal Tail Coefficient (CTC) introduced by \citet{gnecco2021causal}.
For any pair of regularly varying variables $X_i$ and $X_j$ in the LSCM \eqref{eq:LSCM} the CTC is defined as 
\begin{equation}
    \label{eq:ctc}
    \Gamma_{X_i\rightarrow X_j} := \underset{u \rightarrow 1-}{\lim} \mathbb{E}(F_j(X_j)|F_i(X_i)>u).        
\end{equation}
where $F_j$ denotes the distribution function of $X_j$ for $j=1,2$. Intuitively, the coefficient measures a normalized version of the expected rank of $X_j$ conditional on $X_i$ taking an extreme upper-tail value. For extreme events, the Causal Tail Coefficient (CTC) reveals causal links between any two nodes in \eqref{eq:LSCM}, analogous to how the precision matrix captures conditional dependencies in standard Gaussian graphical models.

$\Gamma_{X_i\rightarrow X_j}$ ranges between zero and one and remains unchanged under any monotone increasing marginal transformation,  as it depends on the rescaled margins $F_j(X_j)$. For equal tails of the innovations in \eqref{eq:LSCM}, if extreme realizations of $X_i$ systematically lead to extreme realizations of $X_j$, the coefficient approaches one \citep{gnecco2021causal}. In contrast, if no causal relation exists between the variables, the coefficient approaches one half. The CTC therefore provides a measure of directional dependence between extreme realizations of variables in a linear structural causal model.

\bigskip 

For ease of exposition, in this paper we focus on the causal relationship between two variables and a potential confounder in the extreme. Our approach, however, is theoretically scalable to an arbitrary number of fixed dimensions $p$ by systematically examining all pairwise interactions between nodes. For specific structural scenarios, it might also be possible to obtain tractable results for the full graph in dimension $p$. 

Thus in the following, without loss of generality we focus on identifying causal relations between two variables $X_1$ and $X_2$ under general conditions. Three basic causal structures can arise:
\begin{itemize}
	\item Independence: no directed path between the variables, and no common ancestor
	\item Direct causal relation: $X_1\rightarrow X_2$ or $X_2\rightarrow X_1$
	\item Confounding: both variables share a common ancestor H
\end{itemize} 
The goal is to distinguish these cases using information contained in extreme observations. In contrast to the existing theoretical literature on causal tail inference, we allow the variables in the system and a potential confounder to exhibit heterogeneous tail indices, i.e., the tail decay of different variables may differ substantially which is the prevalent case in practice. We also show that this heterogeneity can in fact be helpful for identifying the causal direction.

\section{Identification under Heterogeneous Tails}\label{sec:CTC}
Existing identification results for the CTC as defined in equation \eqref{eq:ctc} typically assume that all involved variables have identical tail indices. While this assumption simplifies the limit behavior of the coefficient in the extreme tail, it is restrictive in practice, as empirical variables often exhibit substantially different tail behavior.
In this subsection we show that heterogeneous tail indices can themselves provide identifying information about the direction of causal propagation. To address this, we require the following assumption that admits different tails in the innovations and forms the basis of our proposed identification strategy.
\begin{assumption}\label{assumption1}
Let $\epsilon_1, \ldots, \epsilon_p$ be independent random variables such that $\epsilon_i \in RV(\alpha_i)$ for $i = 1, \ldots, p$, and denote $\bar{F}_{\epsilon_r}(x) = \mathbb{P}(\epsilon_r > x)$, where $\epsilon_r$ is chosen as the reference innovation with the dominant right tail. Assume that the following conditions hold: 
\begin{enumerate}[(i)]
    \item For some non-negative constants $c_i^{+}, c_i^{-}$ and $c_i$,
    \[
        \lim_{x \to \infty} \frac{\mathbb{P}(\epsilon_i > x)}{\bar{F}_{\epsilon_r}(x)} = c_i^{+}, 
        \qquad 
        \lim_{x \to \infty} \frac{\mathbb{P}(\epsilon_i \leq -x)}{\bar{F}_{\epsilon_r}(x)} = c_i^{-},
        \qquad 
        \lim_{x \to \infty} \frac{\mathbb{P}(\epsilon_i \leq -x)}{\mathbb{P}(\epsilon_i > x)} = c_i, 
        \quad i = 1, \ldots, p.
    \]
    
    \item For all $i \neq j$,
    \[
        \lim_{x \to \infty} 
        \frac{\mathbb{P}(\epsilon_i > x, \epsilon_j > x)}{\bar{F}_{\epsilon_r}(x)} 
        = 
        \lim_{x \to \infty} 
        \frac{\mathbb{P}(\epsilon_i \leq -x, \epsilon_j > x)}{\bar{F}_{\epsilon_r}(x)} 
        = 
        \lim_{x \to \infty} 
        \frac{\mathbb{P}(\epsilon_i \leq -x, \epsilon_j \leq -x)}{\bar{F}_{\epsilon_r}(x)} 
        = 0.
    \]
    \item Whenever $\alpha_i=\alpha_j$, the corresponding right
    tails are asymptotically comparable:
    \[
        \underset{x\to \infty}{\lim}\frac{\p(\epsilon_i>x)}
             {\p(\epsilon_j>x)}
        \to c_{ij}\in(0,\infty).
    \]
\end{enumerate}
For any $s\in\{1,\ldots,p\}$ such that $\alpha_s>\alpha_r$ $c_s^+=c_s^-=0$. If $\alpha_s=\alpha_r$, $c_s^+$ is positive, while $c_s^-$ may be zero or positive.
\end{assumption}

Assumption \ref{assumption1} holds, for example, for Pareto or Student's $t$ distributed variables, if $\epsilon_r$ is chosen so that $\alpha_r= \underset{i=1,\ldots,p}{\min}\alpha_i$.  $c_s^+=0$ for unequal tail indices follows directly by the definition of regularly varying variables. For ease of exposition, we set $c_{ij}=1$ in Assumption \ref{assumption1} (iii) in the following, abstracting from differences in the slowly varying part of the tails \citep[compare, e.g., ][]{gnecco2021causal}. As we focus on the case of $X_1$ and $X_2$, we have $p=2$ without a confounder, or $p=3$ for the case with confounder.

Assumption \ref{assumption1} is key in order to show that the tail distribution of the sum of regularly varying variables $S_p:=\epsilon_1+\epsilon_2+\ldots+\epsilon_p$ can be approximated by the sum of the tail distributions of the individual variables (Lemma~\ref{lem:twotails}). This is the main technical ingredient for deriving the tail behavior of the CTC when tails in the innovations of the LSCM differ.

\subsection{The Case without Confounder}
In the following, we state all results and estimators for a potential causal relation of the form $X_1\rightarrow X_2$. With interchange of variables analogous results apply for $X_2\rightarrow X_1$.

We can use the different limit behavior of $\Gamma_{X_2\rightarrow X_1}$ and $\Gamma_{X_1\rightarrow X_2}$ in the extreme tails in order to identify the underlying causal direction in the tails. The following proposition studies different cases for different tail indices $\alpha_1$ and $\alpha_2$.

\begin{proposition} \label{proposition}
For the extreme LSCM under Assumption \ref{assumption1} it holds that 
\begin{enumerate}
    \item no connection // $X_1$ and $X_2$ are independent:\\  $\Gamma_{X_1\rightarrow X_2}=\Gamma_{X_2\rightarrow X_1}=0.5$ for all $\alpha_1, \alpha_2$.
    \item causal relation $X_1\rightarrow X_2$: \\ 
    $\Gamma_{X_1\rightarrow X_2}=1$ and $\Gamma_{X_2\rightarrow X_1}=\left\{\begin{tabular}{ll} 0.5 & for $\alpha_1>\alpha_2$ \\ $c \in (0.5, 1)$ & for $\alpha_1=\alpha_2$ with $c$ depending on $\beta_{1\rightarrow 2},\alpha_1,\alpha_2$ \\  1 & for $\alpha_2>\alpha_1$ \end{tabular} \right.$ 
\end{enumerate}
\end{proposition}
Consequences of the above proposition are the following three cases for identification: Case 1: If $\alpha_1>\alpha_2$, we identify that $X_1$ tail causes $X_2$ if $\Gamma_{X_1\rightarrow X_2}=1$, $\Gamma_{X_2\rightarrow X_1}=0.5$ and $\Delta_{X_1 \to X_2}:=\Gamma_{X_1\rightarrow X_2}-\Gamma_{X_2\rightarrow X_1}=0.5$. Case 2: If $\alpha_1<\alpha_2$, the direction of causality cannot be discerned since $\Gamma_{X_1\rightarrow X_2}=1$ and  $\Gamma_{X_2\rightarrow X_1}=1$.
Case 3: For $\alpha_1=\alpha_2$ we can identify that $X_1$ tail causes $X_2$ if $\Gamma_{X_1\rightarrow X_2}=1$, $\Gamma_{X_2\rightarrow X_1}=c$ and $\Delta_{X_1\to X_2}=1-c>0$, where $c\in (0.5,1)$ depending on the causal effects $\beta_{1\rightarrow 2}$ and the tail index. Our contributions are the cases 1 and 2 where the tail differences allow for identification. Case 3 has been extensively studied in \cite{gnecco2021causal} and is contained here for completeness.

\subsection{The Case with Confounder}

In this subsection we consider the extension of the bivariate setting to an additional confounder $H$. In econometric settings this setting is of particular practical relevance. In graph terms, a hidden confounder is an unobserved common ancestor of the two variables where we assume that confounding effects in the LSCM are also linear as illustrated in Figure \ref{fig:cases}. Here we distinguish essentially between two additional cases depending on the degree of heavy-tailedness of $\epsilon_H\in RV(\alpha_H)$ relative to $\epsilon_1$ and $\epsilon_2$ as measured by the respective tail indices. In particular, if $\alpha_H>\max(\alpha_1,\alpha_2)$ the above results continue to hold without any changes as the proposition below indicates. This case corresponds to a confounder whose tail is lighter than the tails of both innovations. If, however, $\alpha_H\leq\max(\alpha_1,\alpha_2)$ identification of the causal direction via the standard CTC can fail. As a remedy, we propose to use instead the adjusted CTC \citep{pasche2023causal}, which conditions on an observed (proxy of the) confounding vector $H$ 
\begin{equation}\label{eq:ctc_confounder}
    \Gamma_{X_2\rightarrow X_1|H}:=\underset{u\rightarrow 1-}{\lim} \e_{(X_1,X_2,H)}\{F_{1}(X_1|H)|F_{2}(X_2|H)>u\}.
\end{equation}
We assume that conditioning on $H$ accounts for all relevant confounding effects requiring a rather strong proxy of the true confounder $H$. Under this assumption, the following results hold in the presence of a confounder $H$.

\begin{proposition} \label{proposition_confounder}
For the extreme LSCM under Assumption \ref{assumption1} and in the presence of a confounder with $\epsilon_H\in RV(\alpha_H)$ it holds that 
\begin{enumerate}
    \item no causal connection between $X_1$ and $X_2$:\\  
\[
    \begin{array}{rclll}
    \Gamma_{X_1\rightarrow X_2}&=&&\Gamma_{X_2\rightarrow X_1}=0.5, & \quad \mbox{ for all } \alpha_H>\max(\alpha_1,\alpha_2).\\
    \Gamma_{X_1\rightarrow X_2}&=&&\Gamma_{X_2\rightarrow X_1}=1, & \quad \mbox{ for all } \alpha_H <\min(\alpha_1,\alpha_2).\\
     \Gamma_{X_1\rightarrow X_2}&=&c_{H1},\, &\Gamma_{X_2\rightarrow X_1}=c_{H2}, & \quad \mbox{ for all } \min(\alpha_1,\alpha_2)\leq \alpha_H\leq\max(\alpha_1,\alpha_2), 
    \end{array}
\]
    with $c_{H1},c_{H2}\in[0.5,1]$. For all $\alpha_H$, including the case $\alpha_H\leq\max(\alpha_1,\alpha_2)$ it holds that\\ 
    \begin{align*}
    \Gamma_{X_1\rightarrow X_2|H}=\Gamma_{X_2\rightarrow X_1|H}=0.5, & \quad \mbox{ for all } \alpha_1,\alpha_2. \\
    \end{align*}

    \item causal relation $X_1\rightarrow X_2$: $\quad$ $\Gamma_{X_1\rightarrow X_2} = 1$ and \\ 
 \[
\begin{array}{rcll}
\Gamma_{X_2\rightarrow X_1}
&=&
\begin{cases}
0.5, & \text{for } \alpha_1>\alpha_2,\\
c_H \in (0.5,1), &  \text{for } \alpha_1=\alpha_2,\\
1, & \text{for } \alpha_2>\alpha_1,
\end{cases}
&\quad  \text{for all } \alpha_H>\max(\alpha_1,\alpha_2), \\[1.2em]

\Gamma_{X_2\rightarrow X_1} &=& \quad \, 1,
& \quad  \text{for all } \alpha_H<\min(\alpha_1,\alpha_2), \\[0.4em]
\Gamma_{X_2\rightarrow X_1} &=&\quad \, c_H' \in [0.5,1],
& \quad  \text{for all } \min(\alpha_1,\alpha_2)\le \alpha_H\le \max(\alpha_1,\alpha_2).
\end{array}
\]

with $c_H, c_H'$  depending on  $\beta_{12},\beta_{H1},\beta_{H2},\alpha_1,\alpha_2,\alpha_H$.\\
For all $\alpha_H$ in particular also in the case $\alpha_H\leq\max(\alpha_1,\alpha_2)$ it holds that\\ 

    $\Gamma_{X_1\rightarrow X_2|H}=1$ and $\Gamma_{X_2\rightarrow X_1|H}=\left\{\begin{tabular}{ll} 0.5 & for $\alpha_1>\alpha_2$, \\ 
    $c_H' \in (0.5, 1)$ & for $\alpha_1=\alpha_2$ with $c_H'$ depending on $\beta_{1\rightarrow 2},\alpha_1,\alpha_2$, \\  
    1 & for $\alpha_2>\alpha_1$. \end{tabular} \right.$
\end{enumerate}
\end{proposition}
From Proposition \ref{proposition_confounder} it follows that for a lighter-tailed confounder $H$ relative to $X_1, X_2$ , the three cases for identification below Proposition \ref{proposition}  apply as before. Thus in particular, we can still get identification of $X_1\rightarrow X_2$ if $\alpha_H>\alpha_1\geq \alpha_2$. As confounding effects might be asymmetric, i.e. $\beta_{H  1}\neq \beta_{H  2}$ in Figure \ref{fig:cases},  they influence non-identification via CTC for all other cases of $\alpha_H$. For the adjusted CTC, however, we get identification for any $\alpha_H$ in particular also for $\alpha_H\leq\max(\alpha_1,\alpha_2)$ as in case 1 above under $\alpha_1>\alpha_2$. 

\begin{figure}[ht]
\centering

\begin{minipage}[t]{0.49\textwidth}
\centering
\caption*{\textbf{(A)} No connection}
\vspace{0.6em}
\begin{tikzpicture}[>=stealth, scale=0.75, transform shape]
    \node[draw, circle] (x1) {$X_1$};
    \node[draw, circle, right=2cm of x1] (x2) {$X_2$};
\end{tikzpicture}
\end{minipage}
\hfill
\begin{minipage}[t]{0.49\textwidth}
\centering
\caption*{\textbf{(B)} Causal connection}
\vspace{0.6em}
\begin{tikzpicture}[>=stealth, scale=0.75, transform shape]
    \node[draw, circle] (x1) {$X_1$};
    \node[draw, circle, right=2cm of x1] (x2) {$X_2$};
    \draw[->] (x1) -- node[midway, above] {$\beta_{12}$} (x2);
\end{tikzpicture}
\end{minipage}

\vspace{0.2em}

\begin{minipage}[t]{0.49\textwidth}
\centering
\scriptsize
\renewcommand{\arraystretch}{1.2}
\begin{tabular}{c|cc}
Index Case & $\Gamma_{X_1\to X_2}$ & $\Gamma_{X_2\to X_1}$ \\
\hline
$\alpha_1 > \alpha_2$ & 0.5 & \fbox{0.5} \\
$\alpha_1 = \alpha_2$ & 0.5 & 0.5 \\
$\alpha_1 < \alpha_2$ & 0.5 & 0.5
\end{tabular}
\end{minipage}
\hfill
\begin{minipage}[t]{0.49\textwidth}
\centering
\scriptsize
\renewcommand{\arraystretch}{1.2}
\begin{tabular}{c|cc}
Index Case & $\Gamma_{X_1\to X_2}$ & $\Gamma_{X_2\to X_1}$ \\
\hline
$\alpha_1 > \alpha_2$ & \cellcolor{gray!20}1 & \cellcolor{gray!20}\fbox{0.5} \\
$\alpha_1 = \alpha_2$ & \cellcolor{gray!20}1 & \cellcolor{gray!20}(0.5,1) \\
$\alpha_1 < \alpha_2$ & 1 & 1
\end{tabular}
\end{minipage}

\vspace{1.8em}

\begin{minipage}[t]{0.49\textwidth}
\centering
\caption*{\textbf{(C)} No connection with light-tailed confounder $\tilde{H}$}
\begin{tikzpicture}[>=stealth, scale=0.75, transform shape]
    \node[draw, circle] (h) {$\tilde{H}$};
    \node[draw, circle, below left=1cm and 1.2cm of h] (x1) {$X_1$};
    \node[draw, circle, below right=1cm and 1.2cm of h] (x2) {$X_2$};

    \draw[->] (h) -- node[midway, above left] {$\beta_{\tilde{H}1}$} (x1);
    \draw[->] (h) -- node[midway, above right] {$\beta_{\tilde{H}2}$} (x2);
\end{tikzpicture}
\end{minipage}
\hfill
\begin{minipage}[t]{0.49\textwidth}
\centering
\caption*{\textbf{(D)} Causal connection with light-tailed confounder $\tilde{H}$}
\begin{tikzpicture}[>=stealth, scale=0.75, transform shape]
    \node[draw, circle] (h) {$\tilde{H}$};
    \node[draw, circle, below left=1cm and 1.2cm of h] (x1) {$X_1$};
    \node[draw, circle, below right=1cm and 1.2cm of h] (x2) {$X_2$};

    \draw[->] (h) -- node[midway, above left] {$\beta_{\tilde{H}1}$} (x1);
    \draw[->] (h) -- node[midway, above right] {$\beta_{\tilde{H}2}$} (x2);
    \draw[->] (x1) -- node[midway, above] {$\beta_{12}$} (x2);
\end{tikzpicture}
\end{minipage}

\vspace{0.2em}

\begin{minipage}[t]{0.49\textwidth}
\centering
\scriptsize
\renewcommand{\arraystretch}{1.2}
\begin{tabular}{c|cc}
Index Case & $\Gamma_{X_1\to X_2}$ & $\Gamma_{X_2\to X_1}$ \\
\hline
$\alpha_1 > \alpha_2$ & 0.5 & \fbox{0.5}  \\
$\alpha_1 = \alpha_2$ & 0.5 & 0.5 \\
$\alpha_1 < \alpha_2$ & 0.5 & 0.5
\end{tabular}
\end{minipage}
\hfill
\begin{minipage}[t]{0.49\textwidth}
\centering
\scriptsize
\renewcommand{\arraystretch}{1.2}
\begin{tabular}{c|cc}
Index Case & $\Gamma_{X_1\to X_2}$ & $\Gamma_{X_2\to X_1}$ \\
\hline
$\alpha_1 > \alpha_2$ &\cellcolor{gray!20} 1 & \cellcolor{gray!20}\fbox{0.5}  \\
$\alpha_1 = \alpha_2$ & \cellcolor{gray!20}1 & \cellcolor{gray!20}(0.5,1) \\
$\alpha_1 < \alpha_2$ & 1 & 1
\end{tabular}
\end{minipage}

\vspace{1.8em}

\begin{minipage}[t]{0.49\textwidth}
\centering
\caption*{\textbf{(E)} No connection with heavy-tailed confounder $H$}
\begin{tikzpicture}[>=stealth, scale=0.75, transform shape]
    \node[draw, circle] (h) {$H$};
    \node[draw, circle, below left=1cm and 1.2cm of h] (x1) {$X_1$};
    \node[draw, circle, below right=1cm and 1.2cm of h] (x2) {$X_2$};

    \draw[->] (h) -- node[midway, above left] {$\beta_{H1}$} (x1);
    \draw[->] (h) -- node[midway, above right] {$\beta_{H2}$} (x2);
\end{tikzpicture}
\end{minipage}
\hfill
\begin{minipage}[t]{0.49\textwidth}
\centering
\caption*{\textbf{(F)} Causal connection with heavy-tailed confounder $H$}
\begin{tikzpicture}[>=stealth, scale=0.75, transform shape]
    \node[draw, circle] (h) {$H$};
    \node[draw, circle, below left=1cm and 1.2cm of h] (x1) {$X_1$};
    \node[draw, circle, below right=1cm and 1.2cm of h] (x2) {$X_2$};

    \draw[->] (h) -- node[midway, above left] {$\beta_{H1}$} (x1);
    \draw[->] (h) -- node[midway, above right] {$\beta_{H2}$} (x2);
    \draw[->] (x1) -- node[midway, above] {$\beta_{12}$} (x2);
\end{tikzpicture}
\end{minipage}

\vspace{0.2em}

\begin{minipage}[t]{0.49\textwidth}
\centering
\scriptsize
\renewcommand{\arraystretch}{1.2}
\begin{tabular}{c|cc|cc}
Index Case & $\Gamma_{X_1\to X_2}$ & $\Gamma_{X_2\to X_1}$ & $\Gamma_{X_1\to X_2\mid H}$ & $\Gamma_{X_2\to X_1\mid H}$ \\
\hline
$\alpha_1 > \alpha_2$ & 1 & \fbox{1} & 0.5 & 0.5 \\
$\alpha_1 = \alpha_2$ & 1 & 1 & 0.5 & 0.5 \\
$\alpha_1 < \alpha_2$ & 1 & 1 & 0.5 & 0.5
\end{tabular}
\end{minipage}
\hfill
\begin{minipage}[t]{0.49\textwidth}
\centering
\scriptsize
\renewcommand{\arraystretch}{1.2}
\begin{tabular}{c|cc|cc}
Index Case & $\Gamma_{X_1\to X_2}$ & $\Gamma_{X_2\to X_1}$ & $\Gamma_{X_1\to X_2\mid H}$ & $\Gamma_{X_2\to X_1\mid H}$ \\
\hline
$\alpha_1 > \alpha_2$ & 1 & \fbox{1} & \cellcolor{gray!20}1 & \cellcolor{gray!20}0.5 \\
$\alpha_1 = \alpha_2$ & 1 & 1 & \cellcolor{gray!20}1 & \cellcolor{gray!20}(0.5,1) \\
$\alpha_1 < \alpha_2$ & 1 & 1 & 1 & 1
\end{tabular}
\end{minipage}

\caption{Possible causal configurations between $X_1$ and $X_2$ with and without a confounder $H$. Panels (A) and (B) show the cases without confounding, panels (C) and (D) the cases with a light-tailed confounder $(\alpha_H>\max(\alpha_1,\alpha_2))$, and panels (E) and (F) the cases with a heavy-tailed confounder $(\alpha_H<\min(\alpha_1,\alpha_2))$. The corresponding implications from Proposition~\ref{proposition} and~\ref{proposition_confounder} for the CTC are shown below each graph. Grey shading highlights cases in which the causal direction can be distinguished, while boxed entries indicate settings in which the confounder can be detected using the Confounder-Test in Section~\ref{sec:device}. In cases where a heavy-tailed confounder $H$ is present, we can apply the adjusted CTC \eqref{eq:ctc_confounder} to recover the values from the cases without confounding.}
\label{fig:cases}
\end{figure}

Figure \ref{fig:cases} illustrates the implications of Propositions~\ref{proposition} and~\ref{proposition_confounder} for the considered scenarios, including both a heavy-tailed confounder ($\alpha_H<\min(\alpha_1,\alpha_2)$) and a light-tailed confounder ($\alpha_H>\max(\alpha_1,\alpha_2)$). Additional cases are provided in Appendix~\ref{app:examples}. It also shows that in the case $\alpha_1> \alpha_2$ a positive difference $|\Delta|:=|\Gamma_{X_1\rightarrow X_2}-\Gamma_{X_2\rightarrow X_1}| >0$ could be either due to a true causal relation or due to the existence of a confounder. Thus identification can be blurred by the existence of a heavy-tailed confounder. When accounting for a heavy-tailed confounder through $\Gamma_{X_2\rightarrow X_1|H}$, the same effect might occur if the considered graph is incomplete in the sense that there are additional heavy-tailed confounding effects that $H$ does not account for.

\section{Estimation and Inference} \label{sec:device}
From the identification results of the previous section we can obtain respective estimators and derive inference results that facilitate the empirical detection of causal relations between extremes from data. In particular, we use sample analogue estimators for the CTC and propose a testing framework for identifying causal directions from observational data. In practice the relation of tail indices is unknown and must be estimated from data. For details on this we refer to Section~\ref{subsec:tail_index_sim}.

In the following we assume that $\left((X_{i,1},X_{i,2})\right)_{i=1}^n$ are independent observations of the random vector $(X_1,X_2)$. A non-parametric estimator for $\Gamma_{X_1\rightarrow X_2}$ is given by
\begin{equation}
\label{eq:estimator}
    \hat{\Gamma}_{X_1\rightarrow X_2}=\frac{1}{k}\sum_{i=1}^n \hat{F}_2(X_{i,2})\textbf{1}\left(X_{i,1}>X_{(n-k),1}\right),
\end{equation}
where $X_{(\ell),1}$ denotes the $\ell$-th order statistic of $X_1$, $k\in \{1,\ldots,n-1\}$ marks the rank threshold and $\hat{F}_2$ denotes the empirical cumulative distribution function of $X_2$. 

\begin{assumption}\label{ass2}
$X_1$ and $X_2$ have continuous marginal distribution functions $F_1$ and $F_2$  with existing density functions $f_1$ and $f_2$. Moreover the von Mises Condition 
\begin{equation*}
    \underset{x \to \infty}{\lim}\frac{x f_j(x)}{1-F_j(x)}=\frac{1}{\gamma_j},\quad\textrm{ for } \gamma_j>0
\end{equation*}
holds for $j\in \{1,2\}$. Analogously, all innovations $\epsilon_i$ appearing in the structural equations considered below, with distribution functions $F_{\epsilon_i}$ and densities $f_{\epsilon_i}$, satisfy the corresponding von Mises condition.
\end{assumption}

Assumption~\ref{ass2} is standard in the context of heavy-tailed models (see e.g. \cite{haan2006extreme}, Theorem 1.1.11). In view of \eqref{eq:reg_var}, it is typically satisfied for most variables that fulfill Assumption~\ref{assumption1}, with $\gamma_j=1/\alpha_{X_j}$ where $\alpha_{X_j}$ denotes the tail index of $X_j$. Assumption~\ref{ass2} is key for consistency of the estimator $\hat{\Gamma}_{X_1\rightarrow X_2}$ of $\Gamma_{X_1\rightarrow X_2}$, i.e. $|\hat{\Gamma}_{X_1\rightarrow X_2}- \Gamma_{X_1\rightarrow X_2}|= o_P(1)$ for $k,n\rightarrow\infty$ and $k/n\rightarrow 0$ \citep{gnecco2021causal}. 

In the presence of a light-tailed confounder the case of no causal link between $X_1$ and $X_2$ can be tested directly using the following theorem.
\begin{theorem} \label{theorem:independence}
Let Assumptions \ref{assumption1} and \ref{ass2} hold. Under $H_0$ of no directed path between $X_1$ and $X_2$ and no confounder, we obtain, for $k/n \to 0$ as $k,n \to \infty$,
\begin{equation}\sqrt{k}\left(\hat{\Gamma}_{X_1\rightarrow X_2}-0.5\right)\xrightarrow{d} N\left(0,\frac{1}{12}\right) \textrm{ and } \sqrt{k}\left(\hat{\Gamma}_{X_2\rightarrow X_1}-0.5\right)\xrightarrow{d} N\left(0,\frac{1}{12}\right).
\label{eq:test_id}
\end{equation}

If $X_1$ and $X_2$ share a common light-tailed confounder $H$ with $\alpha_H>\max\{\alpha_1,\alpha_2\}$, define
$    \rho_j:=\min\{1,\alpha_H-\alpha_j\}$ for $j\in\{1,2\}$.
Then the above convergence result \eqref{eq:test_id} still holds provided that $k=\lfloor n^\nu \rfloor$ with 
\begin{equation*}
    \nu < \min_{j\in\{1,2\}}
    \frac{2\rho_j}{\alpha_j+2\rho_j}.
\end{equation*}
\end{theorem}

Note that the above theorem shows that under no confounding, the null of independence can be tested for any magnitude and relation of the tail indices $\alpha_1$ and $\alpha_2$. If no confounder is present, there is no additional rate restriction required, since under the null $X_1$ and $X_2$ are fully independent and not only asymptotically independent in the tails.

The distributional convergence holds even in the case of a confounder $H$ if its tail is lighter than that of $X_1$ and $X_2$ for a dedicated choice of $k=\lfloor n^\nu\rfloor$ depending on the involved tail indices that drives the overall rate of convergence. Note that $\nu$ is always positive for light-tailed confounders $\alpha_H>\max\{\alpha_1,\alpha_2\}$ taking values in $(0,1)$. Generally, the growth rate of $k$ decreases with the size of $\alpha_1$ and $\alpha_2$, so that larger admissible values of $k$ correspond to heavier tails in $\epsilon_1$ and $\epsilon_2$ if the confounder has much lighter tails than any of the $\epsilon_j$ ($\rho_j=1$), $j\in \{1,2\}$. If instead $\alpha_H-\alpha_j<1$, so that $\rho_j=\alpha_H-\alpha_j$, the admissible growth rate additionally decreases as the tail indices of the confounder and the innovations become closer. Thus, $k$ must grow more slowly when the innovations are lighter-tailed or when the confounder is only slightly lighter-tailed than the innovations. 

The results of Theorem \ref{theorem:independence} can be easily extended if additional ancestors of $X_1$ or $X_2$ are present. Then the corresponding rate is determined by the heaviest-tailed relevant ancestral innovation, that is, by the smallest tail index among the relevant ancestors, rather than only by $\alpha_1$ and $\alpha_2$.

In light of Proposition \ref{proposition_confounder}, our simulations in Section~\ref{sec:simulation} suggest that our tests can even be extended to the case of heavy-tailed confounders by using the adjusted CTC $\Gamma_{X_1\rightarrow X_2|H}$ from \eqref{eq:ctc_confounder} as a basis for the test statistic that applies for all $\alpha_H$ in particular also in the case $\alpha_H\leq\max(\alpha_1,\alpha_2)$. 
For an estimator of the adjusted CTC  a generalized Pareto distribution is fitted to the data of $X_1$ and $X_2$, where the scale parameter $\tau$ accounts for the confounder $H$  \citep{pasche2023causal}. This parametric pre-step using the asymptotic GPD approximation allows us to control the overall statistical properties of the estimator pragmatically
\begin{equation*}
    \widehat{\Gamma}_{X_1\to X_2\mid H}
= \frac{1}{k_l}\sum_{i=1}^{n}
\widehat{F}_2\!\left\{X_{i,2};\,\widehat{\tau}_2(i),\,\widehat{\xi}_2\right\}
\,\mathbf{1}\!\left[
\widehat{F}_1\!\left\{X_{i,1};\,\widehat{\tau}_1(i),\,\widehat{\xi}_1\right\}
> 1-\frac{k}{n}
\right],
\end{equation*}
where $k_l := \left|\left\{\, i \in \{1,\ldots,n\} : \widehat{F}_1\left(X_{i,1};\,\widehat{\tau}_1(i),\,\widehat{\xi}_1\right)> 1-\frac{k}{n}\,\right\}\right|$. The scale parameter is modeled as $\tau_j(i) = \tau_j^0 + \tau_j^{1 \top} H_i,\, i = 1,\ldots,n,\, j = 1,2$ and estimates of $\tau_j^0, \tau_j^1, \xi_j$ are obtained by maximum likelihood. This estimator was proposed under the condition that the confounder shares the same tail indices as $\alpha_1$ and $\alpha_2$. However, our simulations show that it also adequately adjusts for heavier-tailed confounding, specifically when $\alpha_H< \min(\alpha_1,\alpha_2)$, i.e. our simulations suggest that Theorem~\ref{theorem:independence} continues to hold for an adjusted CTC basis.

\bigskip

For determining the causal direction, we test whether the asymmetry in the direct causal relationship within our data is statistically significant with $\Delta_{X_1\to X_2}=\Gamma_{X_1 \to X_2} -\Gamma_{X_2 \to X_1}$. Building on the identified cases for different $\alpha_1,\alpha_2$ in Proposition  \ref{proposition} and the consistency of the empirical estimator in \eqref{eq:estimator} we derive as $k,n\to \infty$ and $k/n \to 0$ under standard conditions that
\begin{equation*}
    \hat{\Delta}_{X_1 \to X_2}=\hat{\Gamma}_{X_1\rightarrow X_2}-\hat{\Gamma}_{X_2\rightarrow X_1}=\Delta_{X_1 \to X_2}+o_P(1).
\end{equation*}
where for $\alpha_1\geq\alpha_2$ and sufficiently large sample size, $\hat{\Delta}_{X_1\to X_2}>0$ marks a causal link $X_1\to X_2$.

For deriving the asymptotic distribution of the corresponding test statistic in Theorem~\ref{theorem:causal}, we first establish the correspondence between the CTC and the survival copula in a lemma. This result is of independent interest. In addition, existing convergence results for upper tail copulas also simplify the proof of the main theorem. 
For $U_1 := F_1(X_1),\, U_2 := F_2(X_2)$ we can define
$\Gamma_{X_1\to X_2}(t) = \e(U_2\mid U_1>1-t)$ for $t\in(0,1)$. This corresponds to the CTC in the limiting case $t\to 0$. We denote by $\bar C$ the survival copula of $(U_1,U_2)$, i.e.
\[
\bar C(u_1,u_2)
=
\p(U_1>1-u_1,\;U_2>1-u_2),
\quad u=(u_1,u_2)\in[0,1]^2.
\]
Further, let $\hat{\bar C}_n$ denote the empirical survival copula,
\[
\hat{\bar C}_n(u_1,u_2)
=
\frac{1}{n}\sum_{m=1}^n
\mathbf{1}\left\{
\hat F_{1}(X_{m,1})>1-u_1,\;
\hat F_{2}(X_{m,2})>1-u_2
\right\},
\quad u\in[0,1]^2.
\]
\begin{lemma}\label{lemma:empcopula}

Under Assumptions \ref{assumption1} and \ref{ass2} the following representations hold:
\begin{enumerate}
    \item For every $t\in(0,1)$,
    \[
    \Gamma_{X_1\to X_2}(t)
    =
    \frac{1}{t}\int_0^1 \bar C(t,r)\,dr.
    \]
    Consequently, if the limit exists,
    \[
    \Gamma_{X_1\to X_2}
    =
    \lim_{t\downarrow 0}\Gamma_{X_1\to X_2}(t)
    =
    \lim_{t\downarrow 0}\frac{1}{t}\int_0^1 \bar C(t,r)\,dr.
    \]

    \item The CTC estimator $\hat{\Gamma}_{X_1\to X_2}$ from \eqref{eq:estimator} admits the representation
    \[
    \hat{\Gamma}_{X_1\to X_2}
    =
    \frac{1}{k}\sum_{s=1}^n
    \hat{\bar C}_n\left(\frac{k}{n},\frac{s}{n}\right).
    \]
\end{enumerate}
\end{lemma}

For the asymptotic result below, we require that the survival copula $\bar C$ of $(U_1,U_2):=\bigl(F_1(X_1),F_2(X_2)\bigr)$
satisfies the regularity conditions required for weak convergence of the
empirical copula process. 
\begin{assumption}\label{ass3} For each $j\in\{1,2\}$, the $j$th first-order partial derivative $\dot{\bar C}_j$ of ${\bar C}$ exists and is continuous on $V_{2,j}:=\{u\in[0,1]^2 : 0 < u_j < 1\}$.
\end{assumption}
See, e.g. \citep{segers2012asymptotics}, who shows that Assumption \ref{ass3} together with Assumption \ref{ass2} implies weak convergence of the empirical survival copula process. Note that while Assumption \ref{ass3} is standard and not restrictive in the interior of the support $t\in (0,1)$, the continuity of  $\dot{\bar C}_j$ might, however, be problematic in the limit case of $t\to 0$ under linear relationships between $X_1$ and $X_2$. We refer to \citet{bucher2013multiplier} for details on this point. For this paper, we therefore restrict attention to $t>0$ in the theorem below. This is sufficient in practice for finite samples.

\begin{theorem}
\label{theorem:causal}

We set $k=\lfloor tn\rfloor \in \mathbb{N}$ for $t\in (0,1)$. Then under Assumptions \ref{assumption1}-\ref{ass3} there exist $\sigma_{\Gamma_1}^2,\sigma_{\Gamma_2}^2, \sigma_\Delta^2\ge 0$ such that for $n\to\infty$
\[
\sqrt{n}\Bigl(
\hat{\Gamma}_{X_1\to X_2}-\Gamma_{X_1\to X_2}(t)
\Bigr)
\xrightarrow{d}
\mathcal N(0,\sigma_{\Gamma_1}^2),
\quad \mbox{and}\quad 
\sqrt{n}\Bigl(
\hat{\Gamma}_{X_2\to X_1}-\Gamma_{X_2\to X_1}(t)
\Bigr)
\xrightarrow{d}
\mathcal N(0,\sigma_{\Gamma_2}^2),
\]
This implies that
\[
\sqrt{n}\Bigl(
\hat{\Delta}_{X_1\to X_2}-\Delta_{X_1\to X_2}(t)
\Bigr)
\xrightarrow{d}
\mathcal N(0,\sigma_\Delta^2).
\]

\end{theorem}

Note that combining Proposition~\ref{proposition_confounder} and Theorem~\ref{theorem:causal} implies that for $\alpha_1>\alpha_2$ and $\alpha_H> max(\alpha_1,\alpha_2)$ the causal direction $X_1\to X_2$ can be empirically tested by assessing $\hat\Delta$ as a test statistic where the asymptotic distribution under the null of no causal link is provided by Theorem~\ref{theorem:causal}. Thus, we formally check 
\begin{equation*}
    H_{0}: \beta_{1  2}=0 \textrm{ against } H_{1}:\beta_{1 2}> 0.
\end{equation*}
Empirically, this causal identification scheme also works for specific small causal effects $\beta_{12}$ even in the case $\alpha_1\leq\alpha_2$ (see simulations in Section~\ref{sec:simulation}) complementing the general theoretical results for the case $\alpha_1<\alpha_2$ of Proposition~\ref{proposition_confounder}.

Since the true variances in Theorem~\ref{theorem:causal} are unknown in practice, we adopt a bootstrap permutation test, following the approach of \citet{pasche2023causal}.

We randomly permute the indices $j=1,2$ for each pair $(\hat{F}_1(X_{i,1}),\hat{F}_2(X_{i,2}))$, $i=1,\ldots,n$, apply bootstrap resampling, calculate each time $\hat{\Delta}_{X_1 \to X_2}$ and compute a bootstrapped p-value to determine the test outcome. With the result for Theorem~\ref{theorem:causal} we refine the test procedure of \cite{pasche2023causal} for settings where the direction or sign of the effect is unknown and use instead a two-sided alternative  $H_1:\beta_{1 2}\neq 0$. To test the presence of a causal effect in the tails for a significance level of $\alpha$ we denote the empirical $(1-\alpha/2)$-quantile of the bootstrap samples by $\Delta_{\alpha}$ and reject $H_0$ if $|\hat{\Delta}_{X_1 \to X_2}|>\Delta_{\alpha}$. If the null hypothesis cannot be rejected, the observed asymmetry in our data is insufficiently pronounced, or the data may be influenced by an unobserved heavy-tailed confounder. 

\bigskip

As Propositions~\ref{proposition} and~\ref{proposition_confounder} and Figure \ref{fig:cases} show, identification of causality can be blurred by a heavy-tailed confounder. We therefore propose the following test to check if such a confounder is present. We want to assess
\begin{equation*}
    H_0: \Gamma_{X_2\to X_1}=0.5 \textrm{ against } H_1:\Gamma_{X_2\to X_1}> 0.5.
\end{equation*}
with the test statistic
\begin{equation}\label{eq:test_statistic}
    T_{C}:=\sqrt{k} \left( \hat{\Gamma}_{X_2\to X_1}-0.5\right).
\end{equation}
The theorem below provides the asymptotic distribution of the test statistic $T_C$.
\begin{theorem}\label{theorem}
    Let Assumptions \ref{assumption1} and \ref{ass2} hold. Moreover assume that $X_1\to X_2$ with $\alpha_2<\alpha_1$. Define $\rho:=\min\{1,\alpha_1-\alpha_2\}$ and let $k=\lfloor n^\nu\rfloor$ with $0<\nu<\frac{2\rho}{\alpha_2+2\rho}$.
    Then the estimator $\hat{\Gamma}_{X_2\to X_1}$  satisfies 
     \begin{equation*}
     \sqrt{k}\left(\hat{\Gamma}_{X_2\to X_1}-0.5\right) \xrightarrow{d} N\left(0 ,\frac{1}{12}\right),  
     \end{equation*}
under the $H_0$ of no confounder.
\end{theorem}
Similar to Theorem~\ref{theorem:independence}, the admissible growth rate of $k=\lfloor n^\nu \rfloor$ depends on the involved tail indices. For $\alpha_1-\alpha_2\geq1$, such that $\rho=1$, the upper bound on $\nu$ increases as $\alpha_2$ decreases. Hence, heavier tailed $\epsilon_2$ allow for a faster growth of $k$, reflecting that the relative contribution of $\epsilon_1$ to extreme realizations of $X_2$ becomes negligible more quickly. If instead $\alpha_1-\alpha_2<1$, then the admissible growth rate additionally decreases as the two tail indices become closer. Intuitively, when the two tail indices are close, the contribution of $X_1$ to extreme realizations of $X_2$ vanishes only far in the tail, requiring a more extreme threshold and hence a slower growth of $k$.

Based on Theorem~\ref{theorem}, under $H_0$, the test statistic $T_C$ converges in distribution to $\mathcal{N}\left(0,1/12\right)$. Hence, at significance level $\alpha$, we reject $H_0$ whenever $T_C>z_{\alpha}$, where $z_{\alpha}$ denotes the $(1-\alpha)$-quantile of $\mathcal N(0,1/12)$.

In practice, we suggest using the Confounder-Test together with the Causality-Test above. If all confounding effects are observable, identification of causal effects and causal directions is possible through joint assessment using both tests. Thus if $H_0$ is not rejected, we find no evidence for heavy-tailed confounding and infer the causal direction $X_1\to X_2$. Otherwise, we account for confounding by using the adjusted estimator $\hat{\Gamma}_{X_1\to X_2\mid H}$ from \eqref{eq:ctc_confounder} and return to the Causality-Test. The Confounder-Test is particularly relevant for assessing whether the adjusted CTC-type tests work as they require that the used proxy for the confounder captures all relevant confounding effects.

\bigskip

Since the proposed test for confounding is only applicable in scenarios where $\alpha_1>\alpha_2$, we propose tests in Section~\ref{sec:simulation} to assess this condition in practice.

In the estimation of $\Gamma_{X_i\rightarrow X_j}$ and all proposed tests the optimal choice of the tuning parameter $k$ is crucial. 
In Section~\ref{sec:simulation} below, we conduct a separate simulation study in order to derive practical guidance on this. In particular we suggest a $k$ that is minimizing the percentage of incorrectly inferred causal connections.

\section{Simulation Evidence}
\label{sec:simulation}
In this section, we first evaluate the finite-sample behavior of the CTC estimators and of the corresponding asymmetry measure $\hat{\Delta}_{X_1\to X_2}$ across different tail, model and sample size scenarios. In a pre-step, we comprehensively address the choice of tuning parameter $k$ in the CTC. As a benchmark we use methods that focus on causality in the center of the distribution and that fail if the tail causality differs from the one in the center.
In the second subsection, we evaluate the proposed tests from Section~\ref{sec:device}. For the crucial conditions on the tail indices, we provide a pretest and also assess its finite-sample performance.
 
We study two types of models. The first model is
\begin{equation}
 H=\epsilon_H,\quad X_1=\beta_{H1}H+\epsilon_1, \quad X_2=\beta_{1 2} X_1+\beta_{H2}H+\epsilon_2.\label{eq:model1}
\end{equation}

where $\epsilon_1$, $\epsilon_2$ and $\epsilon_H$ follow the same distribution family with possibly different tails and $\beta_{1 2}, \beta_{H1},\beta_{H2}\in [0,1)$. For $\beta_{1 2}=\beta_{H1}=\beta_{H2}=0$ this equals configuration (A) and for $\beta_{1 2} >0,\beta_{H1}=\beta_{H2}=0$ configuration (B) in Figure \ref{fig:cases}. For the configurations (C)-(F) the same model is applicable with $\beta_{H1},\beta_{H2}>0$. To study the performance of the methods in the tail of the data we further examine a model, where the causal connection is limited to the tail of the variables
\begin{equation}
 X_1=\epsilon_1, \quad X_2=\textbf{1}(X_1>q_{0.95})\beta_{1 2} X_1+\epsilon_2,\label{eq:model2}
\end{equation}
where $q_{0.95}$ denotes the $95\%$-quantile of the distribution of $X_1$. In this case, the causal relationship between $X_1$ and $X_2$ can still be described with an LSCM in the tails, whereas the overall relationship is non-linear. 

As a benchmark we use the \textit{Linear, Non-Gaussian, Acyclic Model (LiNGAM)} causal discovery algorithm \citep{shimizu2006linear}. This technique focuses not only on potential causality in the tail of the data, but instead on global causal dependencies between the variables. It utilizes the findings of independent component analysis (ICA) \citep{comon1994independent} to estimate the directed acyclic graph for an LSCM under general non-Gaussian errors. If this classical general LiNGAM fails to detect a causal relation, we use the pairwise LiNGAM (\cite{hyvarinen2013pairwise}) which infers the causal direction via the likelihood ratio of the two competing models but does not indicate whether a causal relationship actually exists \citep{gnecco2021causal}. While the CTC is robust against hidden confounders, as long as their tail is not heavier than the tails of the observed variables, LiNGAM lacks the ability to deal with hidden common causes.

For the complete replication code in R we refer to \url{https://github.com/KITmetricslab/Tail-Causality}. With our software code, we build on the  \textit{CausalXtreme} package \cite{gnecco2021causal} for the CTC but also for the pairwise LiNGAM benchmark, and use the code provided in \cite{pasche2023causal} for the adjusted CTC version and permutation test. For the classical LiNGAM benchmark we rely on the R package \textit{pcalg} \citep{lingampackage}.

\subsection{Choice of the Tuning Parameter \textit{k}: Independent Pre-Study} 
\label{subsec:best_k}

We sample from \eqref{eq:model1}, where $\epsilon_1$ follows a Student's $t$ distribution with $\alpha_1$ degrees of freedom and $\epsilon_2$ follows a Student's $t$ distribution with $\alpha_2$ degrees of freedom. We consider sample sizes $n\in\{500,1000,10{,}000\}$ and causal coefficients $\beta_{12}\in(0.1,0.9)$. In order to evaluate the optimal choice of $k$ for different tail indices $\alpha_1$ and $\alpha_2$ we evaluate the proportion of cases in which $\hat{\Gamma}_{X_2\rightarrow X_1}$ is incorrectly greater than $\hat{\Gamma}_{X_1\rightarrow X_2}$. The results for different fractional exponents $\nu$ with $k=\lfloor n^\nu \rfloor$ are shown in \autoref{fig:errors_kvarying_student}. 

Additionally, in order to compare the true DAG with the estimated one, the structural intervention distance (SID) introduced by \cite{peters2015structural} is used to quantify the outcome for different fractional exponents $\nu$ and $k=n^\nu$.\footnote{As in \cite{gnecco2021causal} we use the EASE-algorithm, which is based on the CTC and tries to recover the underlying DAG for a set of variables $X_1,\ldots,X_p$ in an LSCM.} The SID measures the similarity between two DAGs based on their associated causal inference statements. This design is close to  \cite{gnecco2021causal} but extends it to unequal-tail settings. 
In our setting we simulate $80$ different DAGs for each parameter combination. The design includes sample sizes $n\in\{500,1000,10{,}000\}$, numbers of variables $p\in\{4,8,10,16,20,30, 50\}$ 
and tail indices $\alpha_1,\alpha_2\in\{2,3,4\}$. The variation of $p$ gives us additional insights into the choice of this parameter in higher dimensional settings, which generalizes our pairwise setting. For each simulated DAG, the first half of the variables have Student's $t$ distributed innovations with $\alpha_1$ degrees of freedom, while the remaining variables have Student's $t$ distributed innovations with $\alpha_2$ degrees of freedom. The resulting SIDs are shown in \autoref{fig:SID_kvarying_student}.

\begin{figure}[htbp]
    \centering
    \begin{subfigure}[b]{0.48\textwidth}
        \centering
        \includegraphics[width=\textwidth]{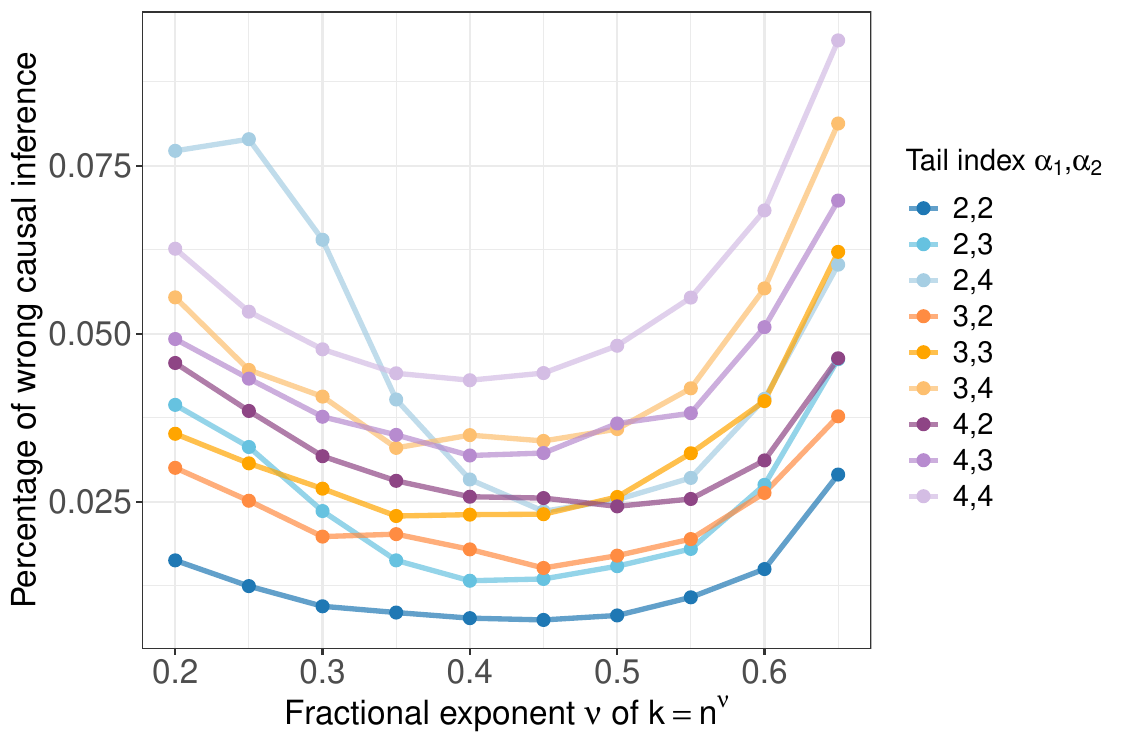}
        \caption{Percentage of incorrect causal inference between two variables.}
        \label{fig:errors_kvarying_student}
    \end{subfigure}
    \hspace{0.02\textwidth} 
    \begin{subfigure}[b]{0.48\textwidth}
        \centering
        \includegraphics[width=\textwidth]{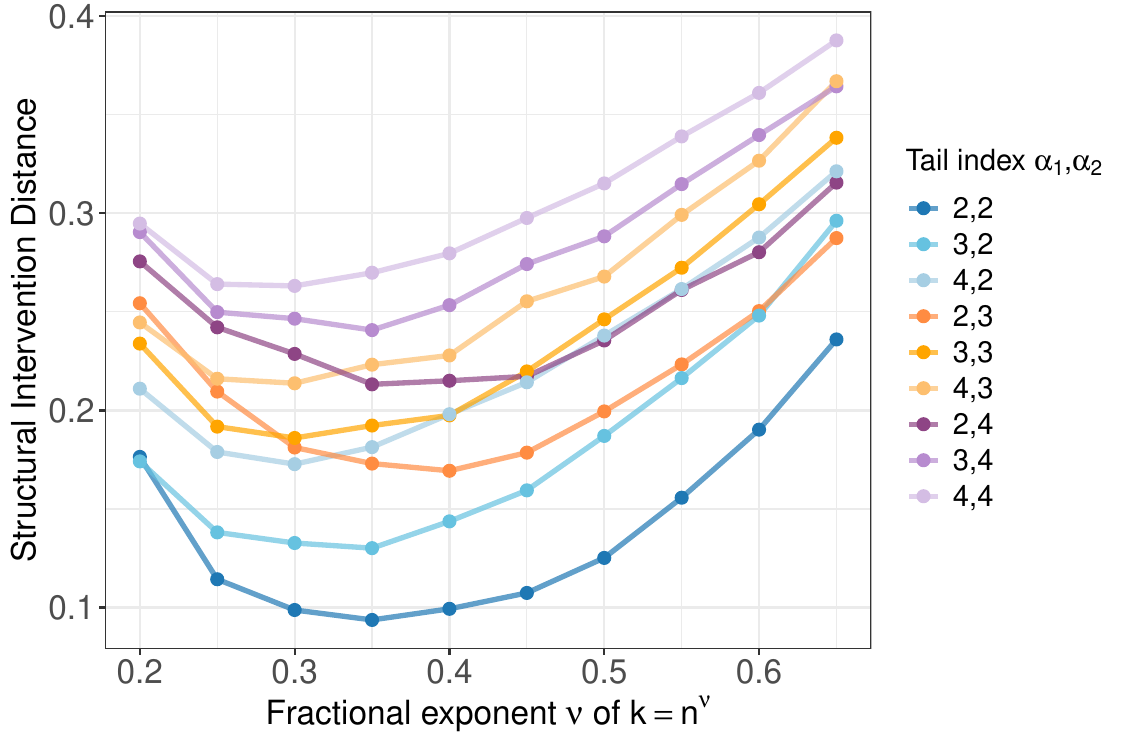}
        \caption{Mean structural intervention distance across simulated DAGs.}
        \label{fig:SID_kvarying_student}
    \end{subfigure}
    \label{fig:best_k}
    \caption{Evaluation of different fractional exponents $\nu$ in $k=\lfloor n^\nu\rfloor$ for different tail-index combinations $(\alpha_1,\alpha_2)$.}
\end{figure}

While the general tail index results in \autoref{fig:SID_kvarying_student}, similar to the equal tail index results in \cite{gnecco2021causal}, suggest selecting $\nu$ between $0.3$ and $0.4$, \autoref{fig:errors_kvarying_student} indicates that a larger value of $k$ is preferable when only two variables are considered. In such two-variable scenarios, choosing $\nu$ between $0.4$ and $0.5$ appears reasonable. Higher tail indices, which correspond to lighter tails, increase the likelihood of incorrect causal inference for both $\alpha_1$ and $\alpha_2$. Based on these observations, we set $k = \lfloor n^{0.4} \rfloor$ in the application, as it demonstrates strong performance across a variety of cases involving only two variables.

\subsection{CTC Estimates: Finite-Sample Rate of Convergence}
\label{subsec:sim_sample_rate}
In a first step, we analyze the finite-sample rate of convergence of the CTC for different tail index scenarios. While the theoretical findings indicate that the CTC exhibits a value of one in the presence of a causal relationship, even when $\alpha_1>\alpha_2$, we expect slower convergence in such cases. 
In the subsequent analysis, we examine the convergence of $\hat{\Gamma}_{X_1\to X_2}$ and $\hat{\Gamma}_{X_2\to X_1}$ under \eqref{eq:model1}. We then use \eqref{eq:model2} to compare the CTC-based approach with LiNGAM in settings where causality is present only in the tails. We study both models \eqref{eq:model1} and \eqref{eq:model2} for 
\begin{itemize}
    \item different types of regularly varying distributions, i.e. $\epsilon_1$ and $\epsilon_2$ are a) Student's $t$ or b) Pareto distributed  with respective tail indices $\alpha_1$ and $\alpha_2$ both on the grid of  $\{2, 3, 4\}$,
    \item different sample sizes on the grid $3000,4000,\ldots,9000,10000$,
    \item different values of $\beta_{12}\in\{0.1,0.3,0.5\}$.
\end{itemize}

According to the results in the previous subsection, we choose $k=25\approx 3000^{0.4}$ fixed across all settings which is optimal for the sample size 3000 and an admissible conservative choice for all others.

\begin{figure}[ht]
    \centering
    \includegraphics[width=\linewidth]{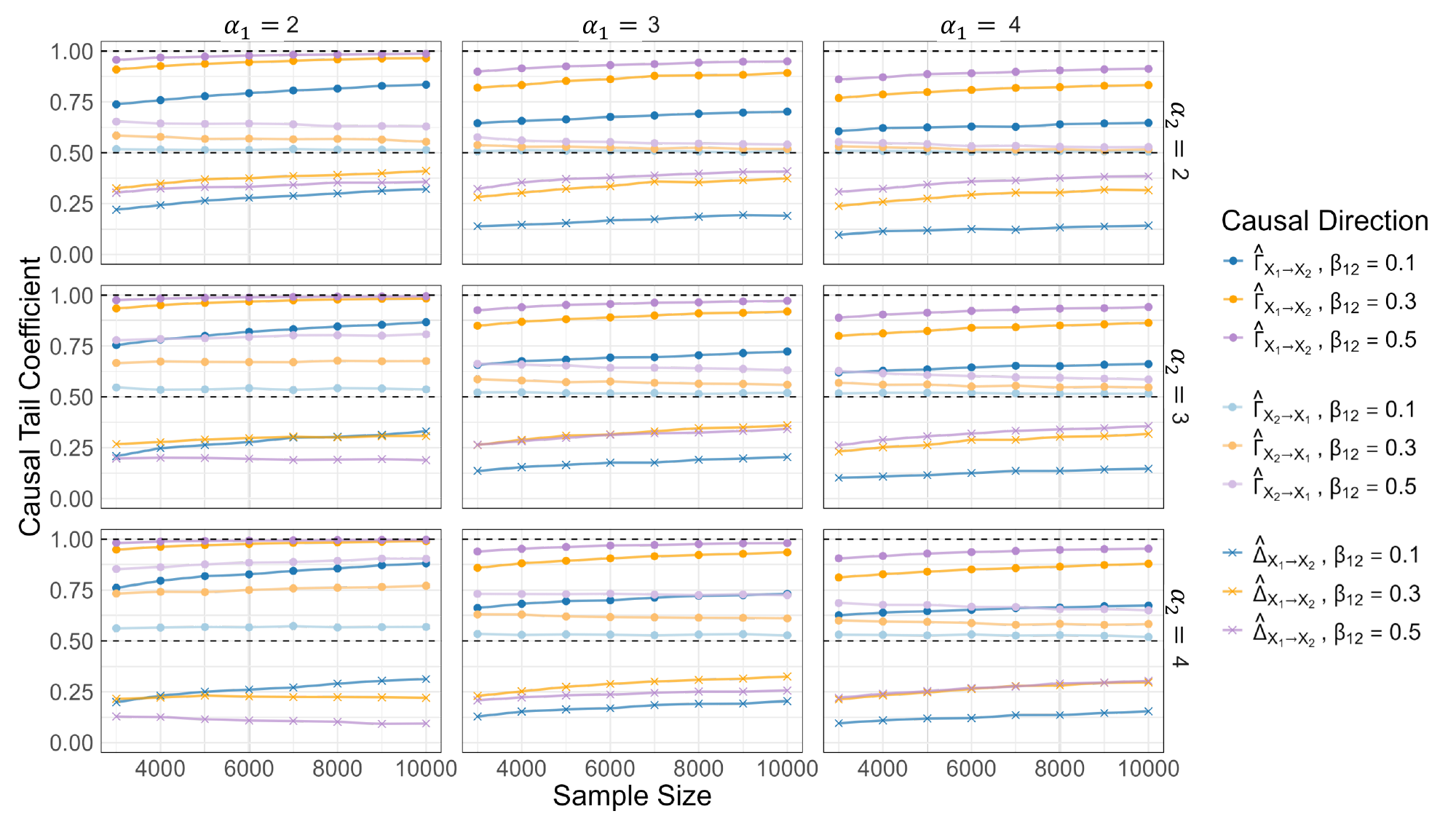}
    \caption{CTC estimates under model \eqref{eq:model1} for different sample sizes, causal coefficients $\beta_{12}$, and tail indices $\alpha_1$ and $\alpha_2$. The innovations $\epsilon_1$ and $\epsilon_2$ follow Student's $t$ distributions with $\alpha_1$ and $\alpha_2$ degrees of freedom.}
    \label{fig:simu_conv_t}
\end{figure}

The results for model~\eqref{eq:model1} with Student’s $t$ distributed noise terms and without confounding effects are shown in \autoref{fig:simu_conv_t}, while results with a confounder present are shown in Appendix~\ref{app:add_simulation}, \autoref{fig:add_simu_conv_t_with_conf_mod9}. For model~\eqref{eq:model2}, the corresponding results are given in Appendix~\ref{app:add_simulation}, \autoref{fig:add_simu_conv_mod10}, but they do not significantly differ from those in \autoref{fig:simu_conv_t}. The findings for model~\eqref{eq:model1} with Pareto distributed errors are also reported in Appendix~\ref{app:add_simulation}, \autoref{fig:add_simu_conv_pareto}. The different panels in \autoref{fig:simu_conv_t} correspond to different combinations of tail indices, while different colors indicate the structural coefficient $\beta_{12}$. Darker colors mark estimates of $\Gamma_{X_1\to X_2}$, and lighter colors correspond to estimates of $\Gamma_{X_2\to X_1}$. The crosses indicate the differences $\hat{\Delta}_{X_1\to X_2}$ between the two estimators. We recall that for $\alpha_1>\alpha_2$, the theoretical value of $\Gamma_{X_2\to X_1}$ is 0.5. This value is closely met even with small sample sizes when the causal weight $\beta_{12}$ is small (e.g. $0.1$) or when the difference in tail indices is large (e.g. $\alpha_1=4$, $\alpha_2=2$). In contrast, smaller differences in tail indices (e.g. $\alpha_1=4$, $\alpha_2=3$) and higher causal coefficients ($\beta_{12}=0.5$) require substantially larger sample sizes (over 10,000) to approach this value. However, a close look at the estimate of $\Gamma_{X_1\rightarrow X_2}$ reveals that the situation is reversed. For $\alpha_1<\alpha_2$ $\hat{\Gamma}_{X_1\rightarrow X_2}$ rapidly approaches 1, particularly for $\beta_{1 2}$ close to $0.5$. For $\alpha_1=\alpha_2$ the convergence is already slower and for $\alpha_1>\alpha_2$, the situation worsens, as the convergence becomes even slower. Compared to the results of the Pareto distribution in Appendix~\ref{app:add_simulation}, the convergence is slower for Student's $t$ distributed variables. That indicates that the convergence speed is not only based on the tail indices and the causal weight $\beta_{1 2}$ but also on the underlying distribution family. In addition to the theoretical results for the case $\alpha_1<\alpha_2$, the simulation study shows that there  still exist specific situations where it is possible to correctly infer the causal direction, in particular when the causal weight $\beta_{1 2}$ is sufficiently small, since this leads to significant differences $\hat{\Delta}_{X_1 \to X_2}$ between $\hat{\Gamma}_{X_1\rightarrow X_2}$ and $\hat{\Gamma}_{X_2\rightarrow X_1}$ for limited sample sizes. Across all considered cases $\hat{\Delta}_{X_1 \to X_2}$ is sufficiently large to uncover the causal effect.

It is therefore crucial to consider the range of values that can realistically be expected in practice. While theoretical results suggest that in population $\Gamma_{X_1\rightarrow X_2}=1$ in the case of a causal relationship between $X_1$ and $X_2$, such an outcome should not be expected when working with the estimate for limited sample sizes or large tail indices. Consequently, for our causal tail testing strategy, we focus on $\hat{\Delta}_{X_1 \to X_2}$ rather than on the concrete values of $\hat{\Gamma}_{X_1 \to X_2}$ and $\hat{\Gamma}_{X_2 \to X_1}$. To interpret  these values, we must determine which relation of the tail indices $\alpha_1$ and $\alpha_2$ holds, as $\hat{\Gamma}_{X_2 \to X_1}$ takes on different values depending on the specific tail index case.

\begin{figure}[ht]
    \centering
    \includegraphics[trim=0 0 0 0, clip,width=\linewidth]{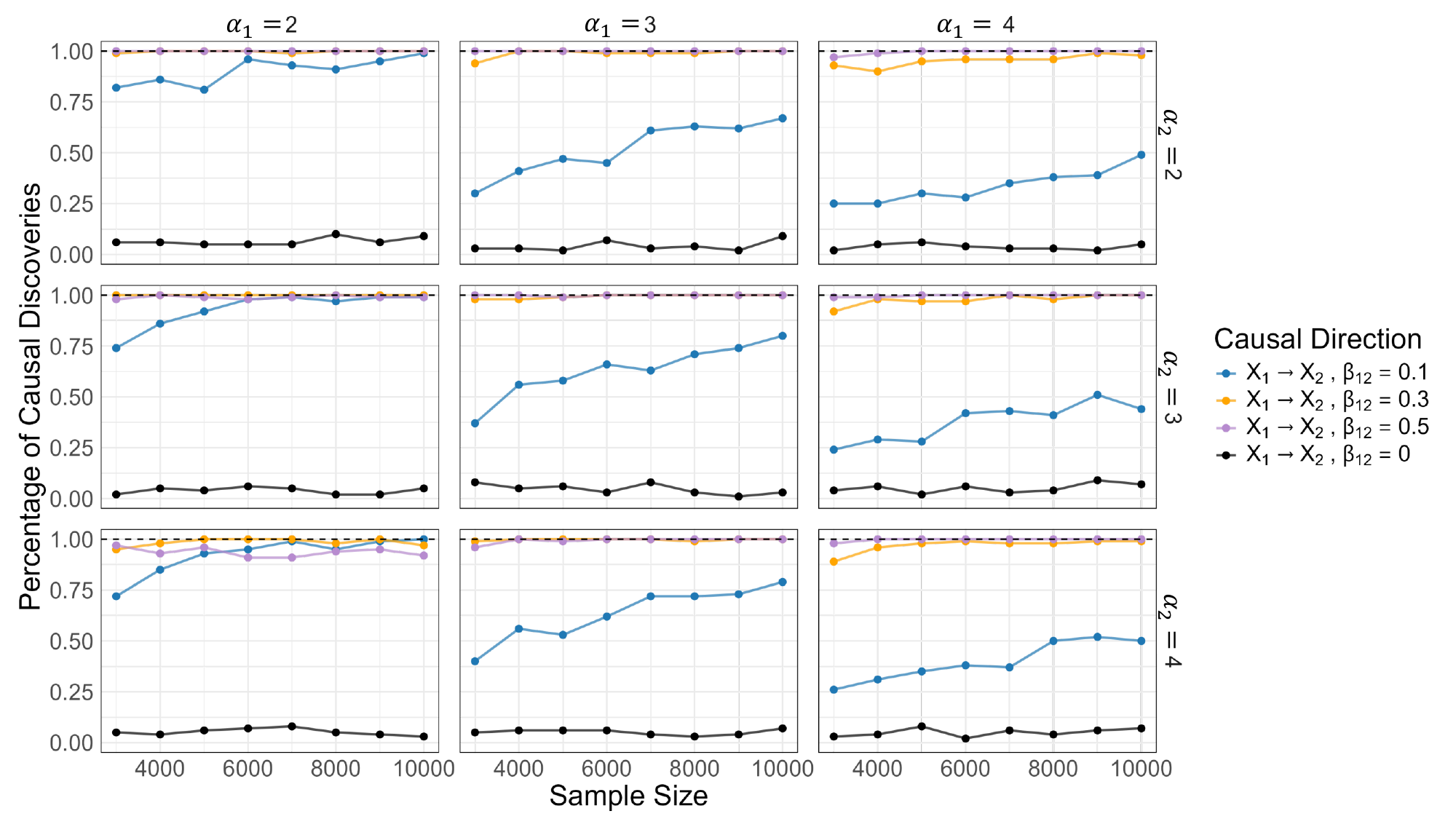}
    \caption{Percentage of causal discoveries for the Causality-Test under model \eqref{eq:model2}, based on the rejection rule $|\hat{\Delta}_{X_1\to X_2}|>\Delta_{0.05}$, for different sample sizes, causal coefficients $\beta_{12}$, and tail indices $\alpha_1$ and $\alpha_2$, where $\epsilon_1$ and $\epsilon_2$ follow Student's $t$ distributions with $\alpha_1$ and $\alpha_2$ degrees of freedom.}
    \label{fig:simulation_ctc_test}
\end{figure}

\begin{figure}[ht]
    \centering
    \includegraphics[trim=0 0 0 0, clip,width=\linewidth]{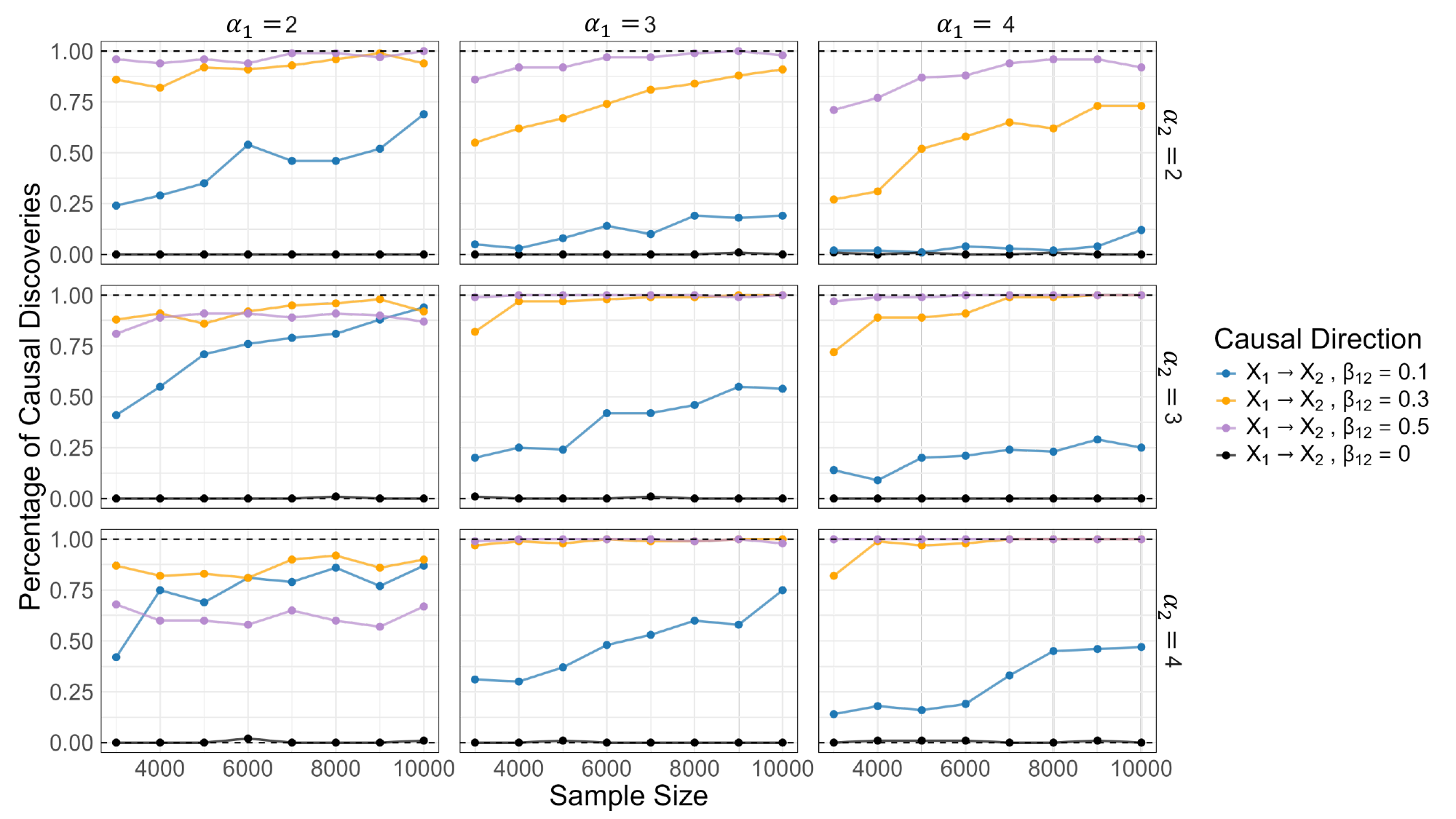}
    \caption{Percentage of causal discoveries for LiNGAM under model \eqref{eq:model2} for different sample sizes, causal coefficients $\beta_{12}$, and tail indices $\alpha_1$ and $\alpha_2$, where $\epsilon_1$ and $\epsilon_2$ follow Student's $t$ distributions with $\alpha_1$ and $\alpha_2$ degrees of freedom.}
    \label{fig:simulation_lingam}
\end{figure}

In the event of a causal relationship manifesting solely in the tails of $X_1$ and $X_2$, as in \eqref{eq:model2}, \autoref{fig:simulation_ctc_test} reports how often the Causality-Test rejects the null hypothesis at the $5\%$ significance level. Here we only sample from configurations with $\beta_{12}\in \{0, 0.1, 0.3, 0.5\}$ without confounders for a realistic assessment of the LiNGAM benchmark.
The corresponding results for LiNGAM are presented in \autoref{fig:simulation_lingam}. The results illustrate the percentage of causal discoveries. Ideally, this percentage should be close to $0$ for $\beta_{12}=0$ and close to $1$ for $\beta_{12}>0$. LiNGAM encounters difficulties when the causal relationship is confined to the tail region, especially when $X_1$ has a lighter tail than $X_2$, that is, $\alpha_1>\alpha_2$, and the causal coefficient is small, for example $\beta_{12}=0.1$. Similar difficulties arise in other heterogeneous-tail settings. In contrast, the Causality-Test results in \autoref{fig:simulation_ctc_test} show substantially better performance. Even when tail indices differ, the CTC-based test detects causal relationships in the tails, consistent with the findings in Section~\ref{sec:pretest_sim}.

The corresponding results for model~\eqref{eq:model1} are reported in Appendix~\ref{app:add_figures}, \autoref{fig:simulation_causality_test_mod9} for the Causality-Test and \autoref{fig:simulation_lingam_mod9} for LiNGAM. As expected, the performance of LiNGAM improves substantially when the causal relationship is also present in the bulk of the data, since it exploits information from the entire sample, in contrast to the proposed Causality-Test, which is restricted to the tails. The Causality-Test therefore shows similar performance for \eqref{eq:model1} and \eqref{eq:model2}.

The findings demonstrate the strong performance of the Causality-Test in detecting causal effects that occur primarily in the tails of the distribution. While LiNGAM is a useful benchmark for LSCMs when causal effects are linear and present throughout the full distribution, it frequently fails to capture causal effects that are confined to extremes, especially in heterogeneous tail settings. This highlights the benefit of using CTC-based methods when the research question concerns causal propagation in the tails rather than average or global dependence.

\subsection{Finite-Sample Performance of the Proposed Tests}
\label{sec:test_eval}
In this subsection, the proposed tail tests of Section~\ref{sec:device} are evaluated. We draw $1{,}000$ samples from each configuration (A)-(F) shown in \autoref{fig:cases} and each tail index combination of $\alpha_1,\alpha_2\in \{2,3,4\}$. The causal coefficients $\beta_{12},\beta_{H 1}, \beta_{H 2}$ are sampled from a $\mathrm{Unif}(0.1,0.9)$ distribution, allowing for asymmetric confounding effects with $\beta_{H  1}\neq \beta_{H  2}$. In light of the results in Subsection~\ref{subsec:sim_sample_rate}, we use a
sample size of $n=10{,}000$ for each LSCM configuration. The errors $\epsilon_1,\epsilon_2,\epsilon_H$ are sampled from Student’s $t$ distributions with degrees of freedom $\alpha_1,\alpha_2,\alpha_H$, respectively, where
\[
\alpha_H =
\begin{cases}
\max(\alpha_1,\alpha_2)+1, & \text{for configurations (C) and (D)},\\[0.3em]
\min(\alpha_1,\alpha_2)-1, & \text{for configurations (E) and (F)}.
\end{cases}
\]
We refrain from settings with $\alpha_H=1$ focusing on cases where the expectation of the confounder is finite and numerical instabilities are avoided. Following Subsection~\ref{subsec:best_k}, we set $k=39\approx 10{,}000^{0.4}$.

\subsubsection{Tail Pre-Tests}
\label{subsec:tail_index_sim}
As discussed in the previous section, depending on the tail indices $\alpha_1$ and $\alpha_2$ of $\epsilon_1$ and $\epsilon_2$, causal directions between $X_1$ and $X_2$ can be identified. In practice, these tail indices are unknown. Since $\epsilon_1$ and $\epsilon_2$ are unobservable, we can only estimate the tail indices of $X_1$ and $X_2$ and can draw conclusions about $\alpha_1$ and $\alpha_2$ from these. As $k,\,n \to \infty$ with $k/n\to 0$, we would not expect to observe $\hat{\alpha}_1<\hat{\alpha}_2$ in practice if $X_2=\beta_{1\to 2} X_1+\ldots+\epsilon_2$ with $\beta_{1\to 2}>0$, because the heavy tail of $X_1$ is asymptotically inherited by $X_2$. Instead, one would expect $\hat{\alpha}_1\approx \hat{\alpha}_2$. We show below, however, that in finite samples this tail transfer is not fully reflected in the estimates, so that such scenarios can still be empirically distinguished.

Since the proposed tests for causality and confounder are primarily applicable in scenarios where $\alpha_1>\alpha_2$, we provide a pre-test to check this condition in practice. To distinguish the cases $H_0:\alpha_1=\alpha_2$ and $H_1:\alpha_1 \neq \alpha_2$, we propose a test for different tail indices (Tail Index-Test) based on \cite{hoga2018detecting}, which extends the Hill estimator
\citep{hill1975simple} 
\begin{equation}\label{eq:hill_estimator}
    \hat{\gamma}(k)=\frac{1}{k}\sum_{i=1}^k \log\left(\frac{X_{(n-i+1)}}{X_{(n-k)}}\right)
\end{equation} 
that is a consistent estimator of $1/\alpha$ under standard regular variation conditions for iid data. In contrast to \eqref{eq:hill_estimator}, 
the proposed estimator is also applicable to time-series data under mild conditions on the dependence structure. Further details are given in Appendix~\ref{app:tail_test}, as well as an alternative test based on the asymptotic normality of the Hill estimator for iid samples.

We compute the test statistic according to the algorithm described in Appendix~\ref{app:tail_test} and reject $H_0$ at the 5\% significance level if the test statistic satisfies $t_I>55.44$. This critical value is simulated and taken from \autoref{tab:teststatistic_quantiles} reported in Appendix~\ref{app:tail_test} that contains the critical values across different significance levels. If $H_0$ is not rejected, the data provides no evidence against equal tail indices. In this case, the theoretical value of $\Gamma_{X_2\to X_1}$ lies in $(0.5,1)$ and is not point-identified, which precludes a direct comparison with its empirical counterpart. We therefore rely on the CTC asymmetry and conclude $X_1\to X_2$. If $H_0$ is rejected, and if the estimated ordering indicates $\hat{\alpha}_1>\hat{\alpha}_2$, then, in the absence of a heavy-tailed confounder, the theory predicts $\Gamma_{X_2\to X_1}=0.5$, which can be assessed using the Confounder-Test.

We next investigate the performance of the Tail Index-Test proposed by \citet{hoga2018detecting}, with further details provided in Appendix~\ref{app:tail_test}. Results for the iid version of the test are reported in Appendix~\ref{app:add_simulation}, \autoref{tab:alternative_index_test}, showing a similar performance to the test considered here.  \autoref{tab:index_test_hoga} reports Tail Index-Test decisions for the configurations shown in \autoref{fig:cases}. Since $\epsilon_1$ and $\epsilon_2$ are not observable in practice, the test is applied to the observed series $X_{1,1},\ldots, X_{n,1}$ and $X_{1,2},\ldots, X_{n,2}$. We set $k_{opt}=500$. 

As $k,n\to \infty$ with $k/n \to 0$, one would expect the tail index of $\epsilon_1$ to be inherited by $X_2$ in the presence of a causal link when $\epsilon_1$ is heavier-tailed than $\epsilon_2$. However, the finite-sample results reported in \autoref{tab:index_test_hoga} indicate that this transfer does not occur, as the test continues to distinguish between the tail indices of $\epsilon_1$ and $\epsilon_2$. Its performance deteriorates further in more complex settings, such as configurations (D) and (F). Moreover, the test has difficulty discriminating between tail indices that are close in value (e.g. $\alpha_1=3, \alpha_2=4$).

\begin{table}[ht]
    \centering
\begin{tabular}{cc|cc|cc|cc}
\toprule
\multicolumn{2}{c}{Tail indices} &
\multicolumn{2}{c}{no conf.} &
\multicolumn{2}{c}{light-tailed conf.} &
\multicolumn{2}{c}{heavy-tailed conf.} \\
$\alpha_1$ & $\alpha_2$ &
(A) & (B) & (C) & (D) & (E) & (F) \\
\midrule
\cellcolor{gray!20}2 & \cellcolor{gray!20}2 &  4.6 & 11.2 & 13.0 & 19.2 &   --   &   --   \\
2 & 3 & 89.2 & 92.0 & 85.9 & 94.9 &   --   &   --   \\
2 & 4 & 98.5 & 98.4 & 99.0 & 99.3 &   --   &   --   \\
\midrule
3 & 2 & 89.6 & 66.8 & 86.5 & 57.0 &   --   &   --   \\
\cellcolor{gray!20}3 & \cellcolor{gray!20}3 &  3.6 & 16.6 & 14.2 & 29.2 &  9.3 & 12.7 \\
3 & 4 & 49.0 & 74.3 & 49.2 & 74.2 & 32.2 & 26.9 \\
\midrule
4 & 2 & 98.8 & 95.7 & 99.3 & 92.9 &   --   &   --   \\
4 & 3 & 51.8 & 20.5 & 49.0 & 16.8 & 33.9 & 26.3 \\
\cellcolor{gray!20}4 & \cellcolor{gray!20}4 &  4.9 & 18.6 & 12.6 & 25.1 &  3.4 & 11.5 \\
\bottomrule
\end{tabular}
\caption{Empirical rejection rates (in \%) of the Tail Index-Test at the $5\%$ level ($t_I>55.44$) for different tail index combinations ($\alpha_1,\alpha_2$) and configurations (A)-(F), where gray cells mark the equal index cases ($H_0$ true).}
\label{tab:index_test_hoga}
\end{table}

\subsubsection{Test of Causality}
\label{sec:pretest_sim}

We study $H_0:\beta_{1 2}=0$ at the fixed $5\%$ significance level for configurations of model~\eqref{eq:model1}. The percentage of test rejections across configurations and tail index combinations are displayed in \autoref{tab:causality_results_extended} for both the CTC as in \eqref{eq:ctc} and the confounder-adjusted CTC in \eqref{eq:ctc_confounder}. Configurations (A) and (B) evaluate Causality-Test performance without confounding and show good control of error rates. In contrast, when a confounder is present, the error rate increases substantially, highlighting the necessity of detecting confounders and employing the CTC adjusted for their presence. Using the adjusted CTC, the Causality-Test error rates drop in configurations (C)-(F) once the appropriate confounder is controlled for. For light‑tailed confounders, the rejection rates are then close to those in the non-confounding cases. For heavy-tailed confounders, the effect cannot be completely removed, but the adjusted CTC still yields markedly better results than the basic CTC. In configurations (A) and (B), the tail indices do not seem to affect the test performance, indicating that the Causality-Test is applicable even when the tail indices differ.

For configurations (A) and (B), we additionally study the variation of the test performance for different $\alpha$-levels. For this we assess whether the proposed test maintains the nominal level across different values of $\alpha$. The results are shown in \autoref{tab:causality_alpha_AB}. As expected, the test maintains the nominal level and as $\alpha$ increases, the type-I error rises while the type-II error declines. Taken together, the results support the use of the Causality-Test in settings without confounding and in settings with confounding when it is properly accounted for.

\begin{table}[ht]
    \centering
\begin{tabular}{cc|cc|cc|cc|cc|cc|cc}
\toprule
\multicolumn{2}{c}{Tail indices} &
\multicolumn{2}{c}{(A)} &
\multicolumn{2}{c}{(B)} &
\multicolumn{2}{c}{(C)} &
\multicolumn{2}{c}{(D)} &
\multicolumn{2}{c}{(E)} &
\multicolumn{2}{c}{(F)} \\
$\alpha_1$ & $\alpha_2$ &
CT & CT\textsubscript{adj} &
CT & CT\textsubscript{adj} &
CT & CT\textsubscript{adj} &
CT & CT\textsubscript{adj} &
CT & CT\textsubscript{adj} &
CT & CT\textsubscript{adj}\\
\midrule
2 & 2 &  5.6 & --    &  100.0 & --    &  6.2 &  4.8 &  99.9 & 100.0 & --& -- & -- & -- \\
2 & 3 &  5.3 & --    &  99.1 & --    & 11.1 &  5.5 &  98.1 &  99.0 & -- & -- & -- & --\\
2 & 4 &  5.6 &--    & 87.0 &--    & 18.8 &  7.0 & 85.1 &  86.6 & -- & -- & -- & --\\
\midrule
3 & 2 &  5.8 & --    &  99.2 & --    & 12.4 &  5.4 &  98.4 &  98.5 & -- & -- & -- & --\\
3 & 3 &  4.4 & --    &  99.2 & --    & 10.3 &  4.9 &  97.1 &  98.5 & 42.8 &  8.9 & 60.3 & 93.6\\
3 & 4 &  4.3 & --    &  98.8 & --    & 11.0 &  5.6 &  90.6 &  95.3 & 54.4 & 15.6 & 41.5 & 74.2\\
\midrule
4 & 2 &  5.6 & --    &  97.6 & --    & 19.8 &  6.9 &  97.1 &  96.0 & -- & -- & -- & --\\
4 & 3 &  4.9 & --    &  96.6 & --    & 10.7 &  3.7 &  96.7 &  95.5 & 51.9 & 16.8 & 69.4 & 92.1\\
4 & 4 &  5.4 & --    &  97.1 &--    &  9.1 &  4.6 & 85.0 &  91.6 & 34.1 &  8.2 & 66.3 & 87.6\\

\bottomrule
\end{tabular}
\caption{Empirical rejection rates (in \%) of the standard Causality-Test (CT) and the confounder-adjusted Causality-Test (CT\textsubscript{adj}) at the $5\%$ level across configurations (A)-(F) and tail index
combinations $(\alpha_1,\alpha_2)$. For (A), (C) and (E) $H_0$ holds.}
\label{tab:causality_results_extended}
\end{table}

\begin{table}[ht]
    \centering
\begin{tabular}{cc|cccc|cccc}
\toprule
\multicolumn{2}{c}{Tail indices} &
\multicolumn{4}{c}{(A) no causal connection} &
\multicolumn{4}{c}{(B) causal connection} \\
$\alpha_1$ & $\alpha_2$ &
$1\%$ & $2.5\%$ & $5\%$ & $10\%$ &
$1\%$ & $2.5\%$ & $5\%$ & $10\%$ \\
\midrule
2 & 2 &
0.7 & 2.8 & 5.6 & 9.8 &
100.0 & 100.0 & 100.0 & 100.0 \\
2 & 3 &
1.0 & 3.3 & 5.3 & 9.8 &
97.2 & 98.4 & 99.1 & 99.4 \\
2 & 4 &
0.9 & 2.5 & 5.6 & 10.5 &
76.9 & 82.9 & 87.0 & 90.4 \\
\midrule
3 & 2 &
0.7 & 3.1 & 5.8 & 10.1 &
97.2 & 98.6 & 99.2 & 99.7 \\
3 & 3 &
0.8 & 2.1 & 4.4 & 8.7 &
97.9 & 98.6 & 99.2 & 99.5 \\
3 & 4 &
1.0 & 2.2 & 4.3 & 9.6 &
97.0 & 98.3 & 98.8 & 99.5 \\
\midrule
4 & 2 &
1.4 & 2.4 & 5.6 & 9.8 &
94.1 & 95.8 & 97.6 & 98.6 \\
4 & 3 &
1.2 & 2.7 & 4.9 & 8.6 &
93.7 & 95.5 & 96.6 & 97.7 \\
4 & 4 &
1.1 & 2.4 & 5.4 & 9.5 &
93.8 & 96.1 & 97.1 & 98.8 \\
\bottomrule
\end{tabular}
\caption{Empirical rejection rates (in \%) of the Causality-Test for
configurations (A) and (B) at nominal levels $\alpha\in\{1\%,2.5\%,5\%,10\%\}$
and different tail index combinations $(\alpha_1,\alpha_2)$.}
\label{tab:causality_alpha_AB}
\end{table}

\subsubsection{Confounder Test}
\label{sec:confounder_sim} 
The results of the Confounder-Test for scenarios with $\alpha_1>\alpha_2$ are presented in \autoref{fig:conf_test_table}. While the test performs well in detecting scenarios without a confounder and without a causal link (configuration (A)), the performance drops significantly for configurations with a causal link (configuration (B)). The choice $k=\lfloor n^{0.4}\rfloor$ is motivated by overall finite-sample performance but need not satisfy the stricter rate condition of Theorem~\ref{theorem} for every tail-index configuration. In particular, for $(\alpha_1,\alpha_2)=(4,3)$, the theorem requires $\nu<0.4$. This helps explain the size distortions observed for the Confounder-Test. We therefore run an additional simulation study for varying $k$ to assess whether the Confounder-Test also maintains the significance level in the presence of causal links. For configurations (B) and (F), we run the simulation $1{,}000$ times for each $k\in \{3,   6,  10,  15,  25,  39,  63, 100\}$ and $(\alpha_1,\alpha_2)\in \{(3,2),(4,2),(4,3)\}$, $\alpha_H=2$. The results are shown in \autoref{fig:typeI_conftest} and \autoref{fig:typeII_conftest}. As $k$ decreases, the type-I error (configuration (B)) declines and approaches the nominal $\alpha$-level more quickly when the gap between $\alpha_1$ and $\alpha_2$ is larger. Conversely, the type-II error (configuration (F)) decreases as $k$ increases. These findings indicate that the Confounder-Test controls the significance level only when applied sufficiently far into the tails, and that its performance is highly sensitive to the tail indices. Overall, these results indicate a trade-off in choosing $k$. It should be small enough to limit false detection of a heavy-tailed confounder, yet large enough to ensure adequate power to detect one when present.

\autoref{fig:confounder_test_beta_plots} displays Confounder-Test decisions in settings with and without confounding for combinations of tail indices $\alpha_1$ and $\alpha_2$ of $\epsilon_1$ and $\epsilon_2$ and for different causal coefficients $\beta_{H 1}$ and $\beta_{H2}$. In finite samples, test outcomes depend not only on whether $\alpha_1>\alpha_2$ or $\alpha_1\leq \alpha_2$, but also on the magnitude of the difference between the tail indices and on the strength of the causal weights. Larger confounding coefficients increase the likelihood that the Confounder-Test rejects $H_0$.

\begin{figure}[ht]
    \centering

    \begin{subfigure}[b]{\textwidth}
        \centering
        \begin{tabular}{cc|cc|cc|cc}
        \toprule
        \multicolumn{2}{c}{Tail indices} &
        \multicolumn{2}{c}{no conf.} &
        \multicolumn{2}{c}{light-tailed conf.} &
        \multicolumn{2}{c}{heavy-tailed conf.} \\
        $\alpha_1$ & $\alpha_2$ &
        (A) & (B) & (C) & (D) & (E) & (F) \\
        \midrule
        3 & 2 & 0.3   & 33.6  &  6.5  & 49.8 &   --   &   --   \\
        4 & 2 & 0.5   & 20.5  &  5.9  & 37.9 &   --   &   --   \\
        4 & 3 & 0.3   & 62.8  & 19.2  & 85.9 &  87.0 &  99.6 \\
        \bottomrule
        \end{tabular}
        \caption{Empirical rejection rates (in \%) of the Confounder-Test at
        the $5\%$ level for different tail index combinations
        $(\alpha_1,\alpha_2)$ and configurations (A)-(F).}
        \label{fig:conf_test_table}
    \end{subfigure}

    \vspace{0.8em}

    \begin{subfigure}[b]{\textwidth}
        \centering
        \begin{subfigure}[b]{0.48\textwidth}
            \centering
            \includegraphics[width=\textwidth]{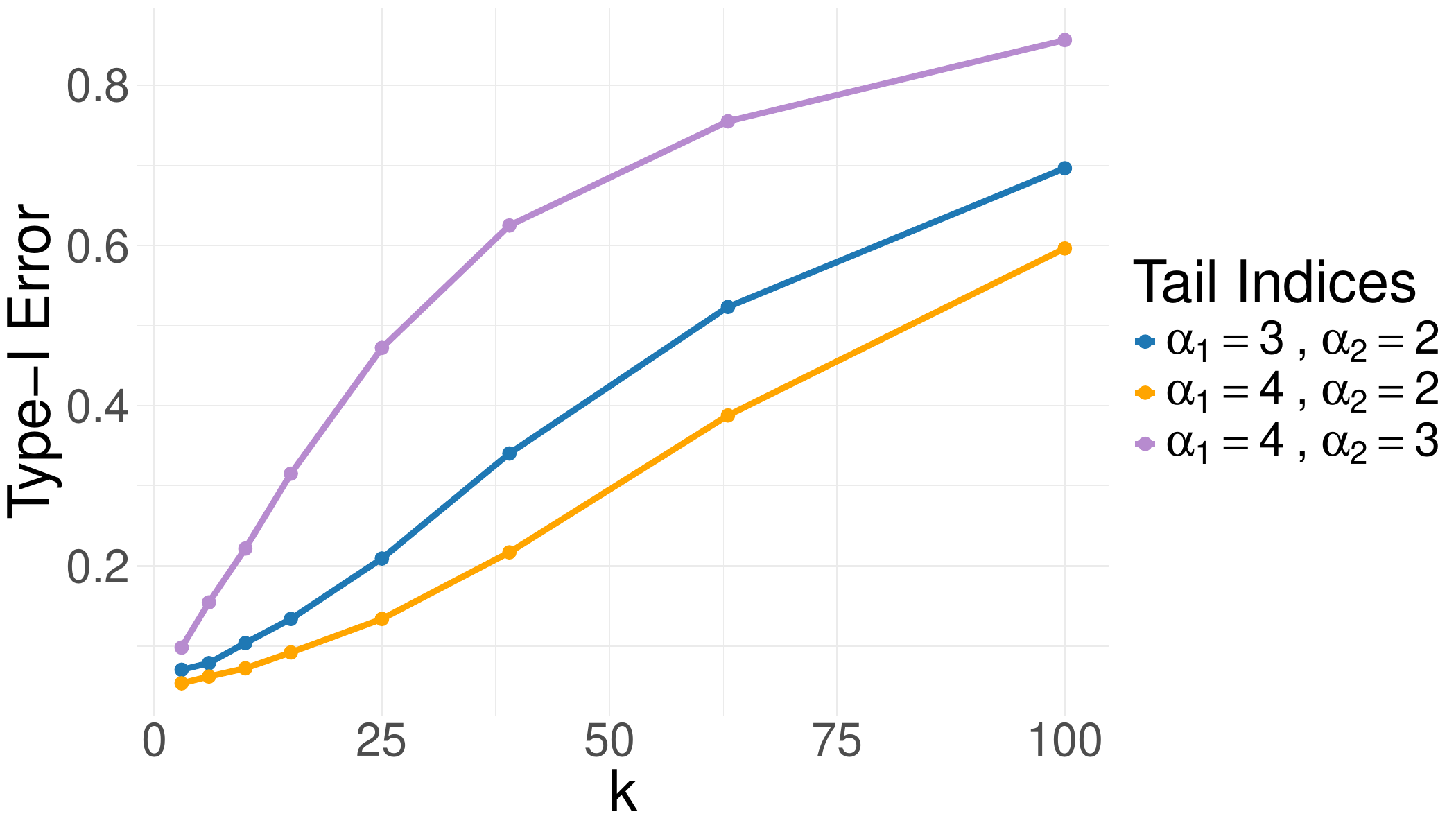}
            \caption{Type-I error based on configuration (B).}
            \label{fig:typeI_conftest}
        \end{subfigure}
        \hspace{0.02\textwidth}
        \begin{subfigure}[b]{0.48\textwidth}
            \centering
            \includegraphics[width=\textwidth]{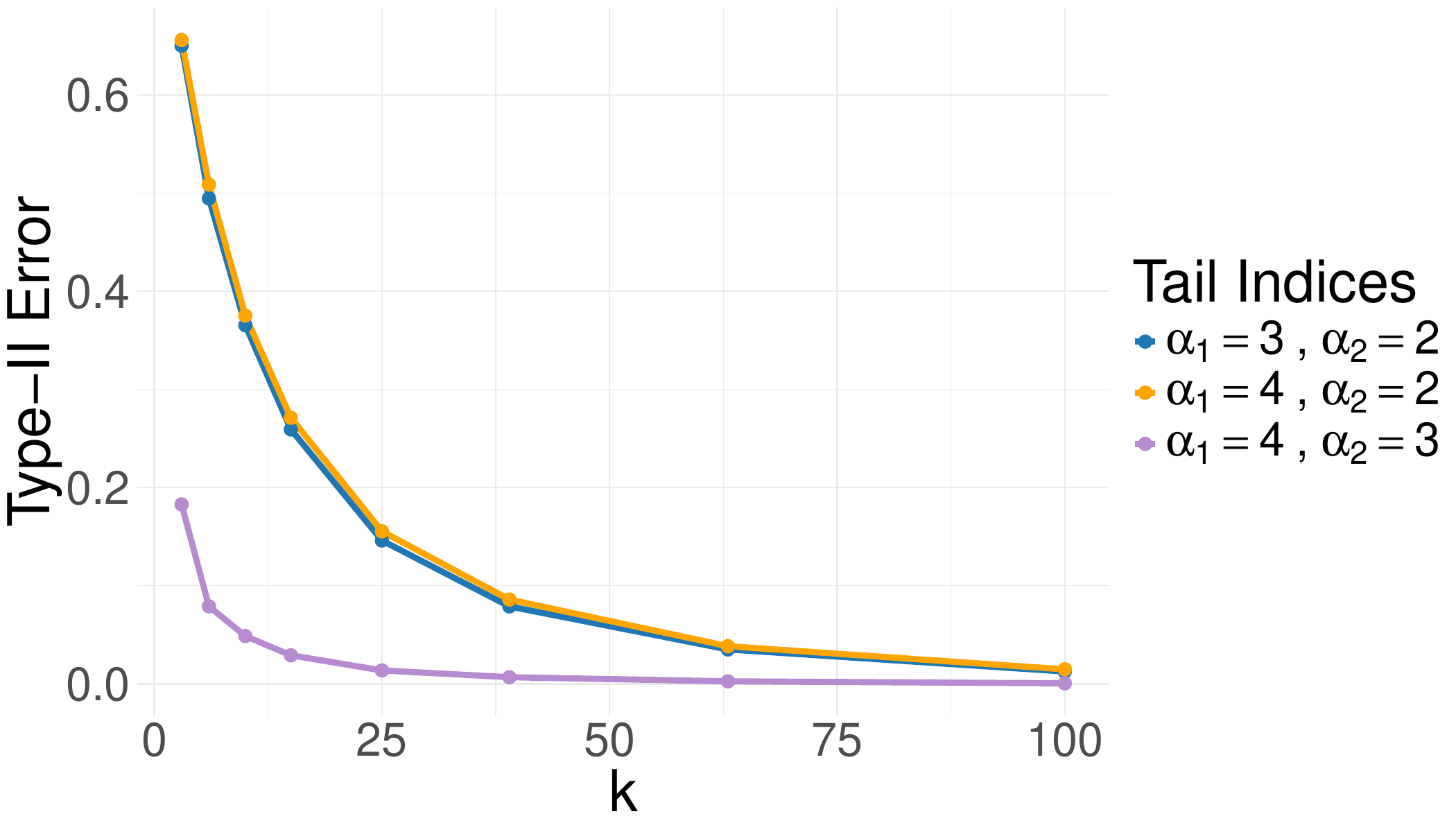}
            \caption{Type-II error based on configuration (F).}
            \label{fig:typeII_conftest}
        \end{subfigure}
    \end{subfigure}

    \caption{Confounder-Test performance: (a) empirical rejection rates at the
    $5\%$ level for different tail index combinations $(\alpha_1,\alpha_2)$ and
    configurations (A)-(F), and (b) type-I and (c) type-II errors as functions of
    $k$ for different tail index combinations $(\alpha_1,\alpha_2)$ and fixed $\alpha_H=2$.}
    \label{fig:confounder_test_overview}
\end{figure}

\begin{figure}[htbp]
    \centering
    \begin{subfigure}[b]{0.48\textwidth}
        \centering
        \includegraphics[width=\textwidth]{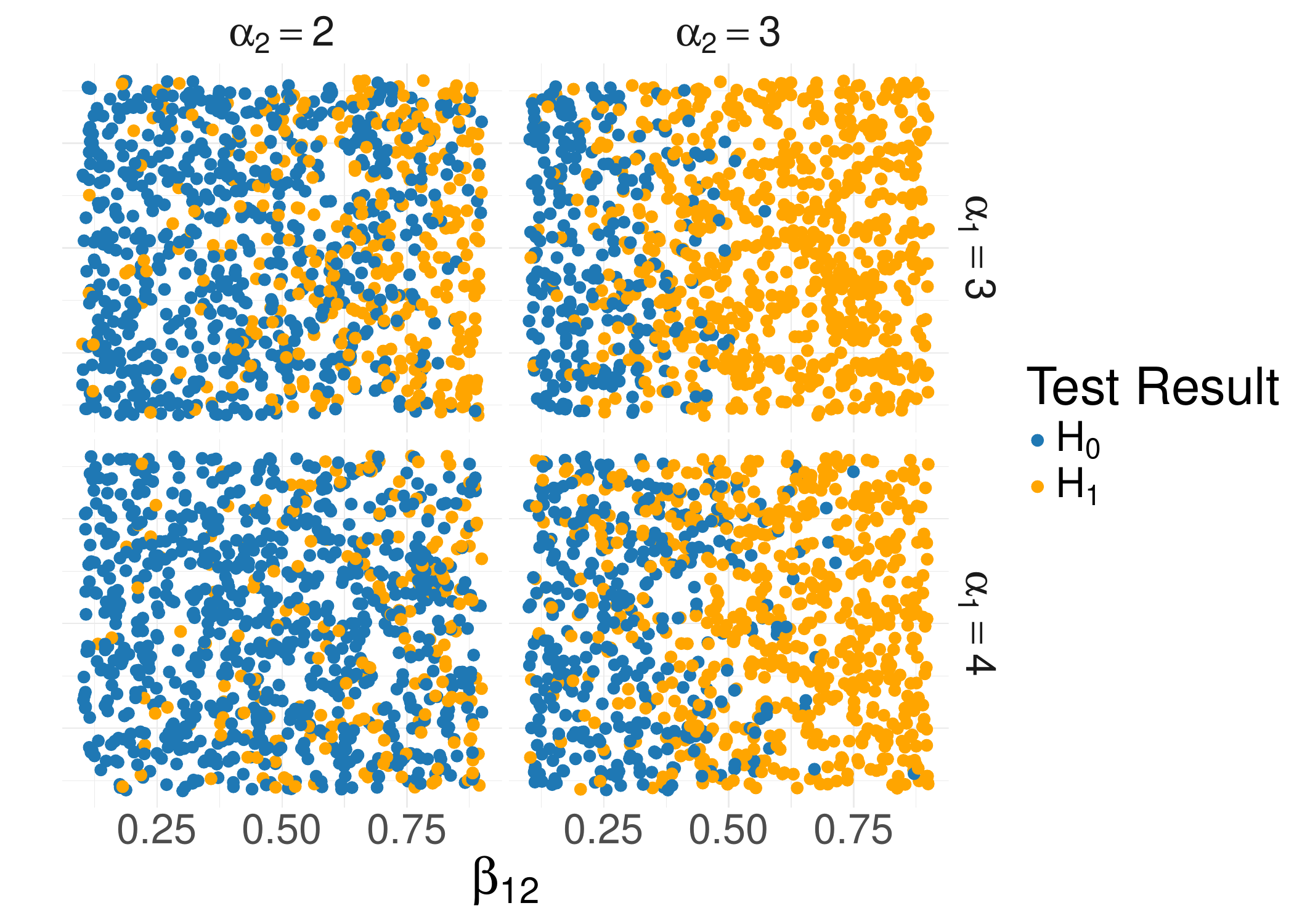}
        \caption{Confounder-Test decisions for configuration (B) and different tail indices $\alpha_1$ and $\alpha_2$ depending on $\beta_{12}$.}
        \label{fig:conftest_beta_plot1}
    \end{subfigure}
    \hspace{0.02\textwidth} 
    \begin{subfigure}[b]{0.48\textwidth}
        \centering
        \includegraphics[width=\textwidth]{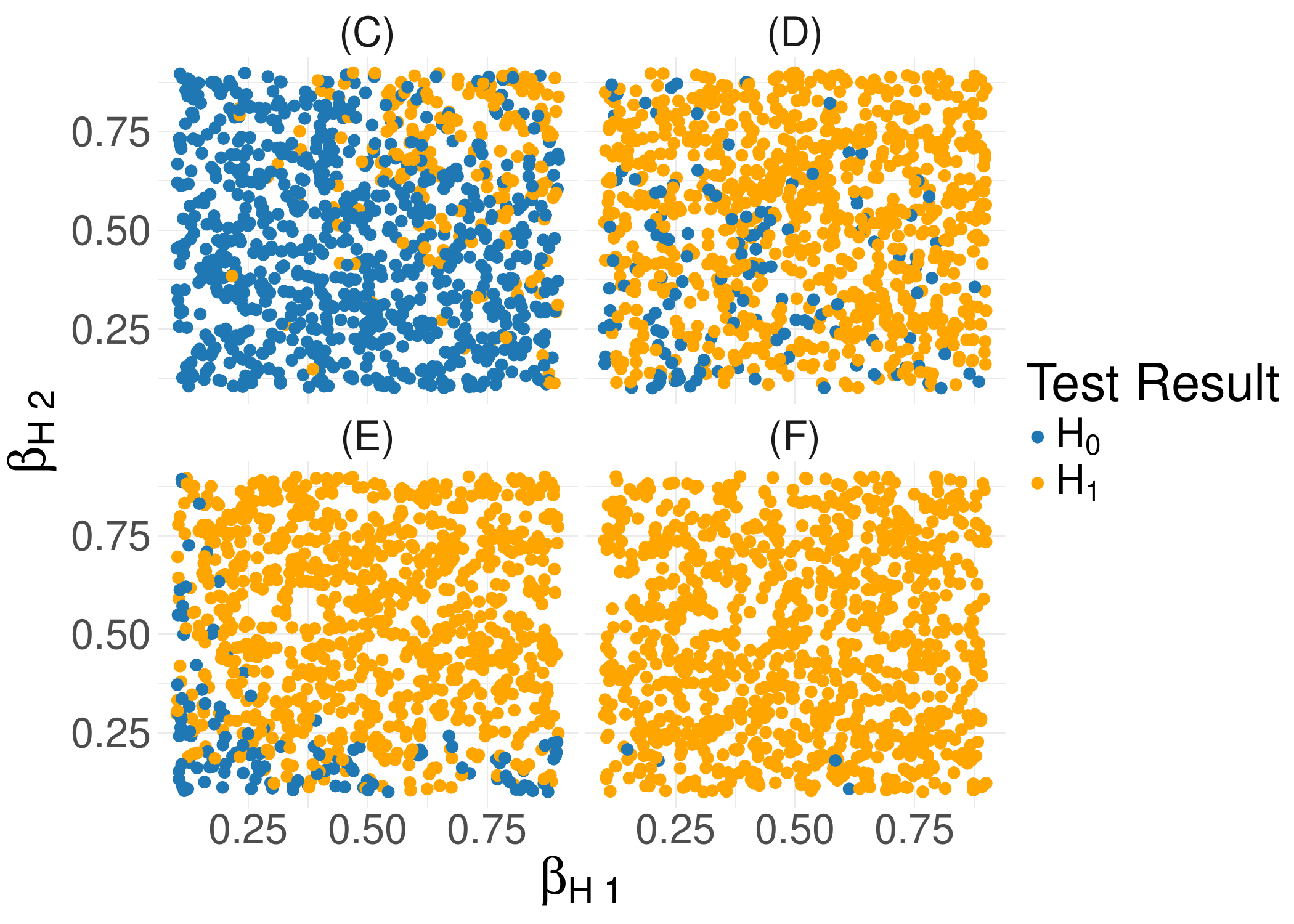}
        \caption{Confounder-Test decisions for configurations (C)-(F) and tail indices $\alpha_1=4$ and $\alpha_2=3$ depending on $\beta_{H1}$ and $\beta_{H2}$}
        \label{fig:conftest_beta_plot2}
    \end{subfigure}
    \caption{Confounder-Test decisions for different tail indices $\alpha_1\geq\alpha_2$ and confounding coefficients $\beta_{H1}$ and $\beta_{H2}$.}
    \label{fig:confounder_test_beta_plots}
\end{figure}

\section{Applications}
\label{sec:application}

In this section, we use the proposed identification schemes of Section~\ref{sec:device} in order to study three distinct application settings that cover diverse fields such as climate and finance. In two of the applications considered, the very nature of the setting suggests that there should exist a specific causal direction. We employ these as baseline scenarios in order to verify if the tail tests work even though standard assumptions such as, e.g., iid observations do not hold. However, the third financial application yields interesting novel insights in the detected causal direction for which neither previous empirical evidence nor theoretical considerations provide a clear causal direction. All of the applications share that we investigate extreme causal relations for quantities that we show to have heavy tails that do not necessarily coincide in their tail behavior (see Figure \ref{fig:histograms}).

We estimate the CTC, $\Delta_{X_1\rightarrow X_2}$ and the critical value $\Delta_{0.05}$ based on $k=\lfloor n^{0.4}\rfloor$ observations, following the tuning-parameter choice discussed in Section~\ref{subsec:best_k}. Following this initial estimation, we proceed with the tests for causality and confounder specified in Section~\ref{sec:device}. Moreover, we analyze the finite-sample behavior of the CTC estimates. To do so, we fix the value $k$ and vary the sample size through bootstrap resampling from the full dataset. For each sample size, we draw bootstrap samples, repeat the estimation $100$ times, and plot the resulting mean CTC estimates. This provides insight into the stability of the CTC estimates as the analysis moves further into the tails, which is particularly relevant in applications with limited sample sizes. Results are presented for $k\in \{\lfloor n^{0.4}\rfloor,\lfloor n^{0.5}\rfloor \}$, because smaller values of $k$ typically lead to higher estimator variance, whereas larger values may include too many non-extreme observations and therefore fail to capture tail behavior accurately. Additionally, we present density histograms of the Causality-Test bootstrap samples to quantify the uncertainty in the estimates and to show the location of the estimate $\hat{\Delta}_{X_1 \to X_2}$.

As a benchmark for the results we use the LiNGAM algorithm (see Section~\ref{sec:simulation}). The results are then interpreted with respect to the existing literature in the considered application cases.

\subsection{Climate Extremes}
\subsubsection*{Precipitation and Train Delays in Switzerland}

In this study, we analyze the influence of extreme precipitation on train delays in Switzerland. In this context, the direction of potential causality is unambiguous, as it is implausible that train delays affect precipitation. The influence of precipitation on train delays is a matter of public interest and has been previously assessed. \cite{palmqvist2023train} study the cause of train delays in Sweden, based on delay attribution codes provided by the railway service. They find that especially storms have a significant impact. Moreover, \cite{brazil2017weather} demonstrate that rain is the primary factor causing delays in the train system in Dublin. In their studies they employ simple regression models to analyze the effect of rainfall on train delays, primarily focusing on the central part of the distribution rather than on the extremes. Consequently, the hypothesis that extreme precipitation affects train delays in Switzerland appears plausible. 

We focus on the Swiss railway system that is internationally known for its punctuality. Thus investigating whether even extreme weather events can lead to significant delays does not suffer from hidden confounding effects (For other countries this might be an implausible assumption). To further avoid confounding effects from within system priorities, we study the primary rail route in Switzerland, namely the line from Zurich to Bern. The expected and actual times of departure and arrival of trains are retrieved from \textit{opentransportdata.swiss}. We extract the train delays by calculating the difference between the expected and the actual durations. Following a thorough data analysis, we restrict our focus to the summer months (May to August), as this period exhibits the highest levels of precipitation. Additionally, we exclude observations in which intermediate stops between Zurich and Bern are missing, or where the train does not complete the full route, likely due to construction work or similar disruptions. The analysis further reveals that Saturdays, Sundays and public holidays exhibit significantly lower delay times. Therefore, these days are also excluded from the dataset. The final dataset comprises 3,994 observations between May 2021 and July 2024. The precipitation data from the Zurich weather station is provided by \textit{MeteoSchweiz}. For each departure time, we extract the corresponding hour and match it to the precipitation in Zurich. In addition, we also compute the sum of precipitation for the last three and six hours. 

We apply the tail tests described in Section~\ref{sec:device}. The results of the Causality- and Confounder-Tests reported in \autoref{tab:ctc_train} show that, for both precipitation variables, the estimated asymmetry $\hat{\Delta}_{precipitation \to delay}$ exceeds $\Delta_{0.05}$ and that there is no evidence of a heavy-tailed confounder. Since the Causality-Test is significant in both cases, we conclude that there is evidence of a causal tail effect for $precipitation \to delay$ and $precipitation\,3h \to delay$.

In addition, \autoref{tab:tail_train} reports the results of the Tail Index-Test together with the corresponding Hill estimates. Although delay exhibits a heavier tail than precipitation, the difference is not large enough to reject the null hypothesis of equal tail indices.

\begin{table}[ht]
    \small
    \centering
    \begin{tabular}{lcccccc}
              & \multicolumn{4}{c}{\textbf{Causality-Test}}      & \multicolumn{2}{c}{\textbf{Confounder-Test}}\\
        \cmidrule(lr){2-5}  \cmidrule(lr){6-7}
     Cumulated Hours & $\hat{\Gamma}_{prec.\rightarrow delay}$   & $\hat{\Gamma}_{delay \rightarrow prec.}$ & $\hat{\Delta}_{prec.\rightarrow delay}$ & $\Delta_{0.05}$ &$t_{C}$  & Decision \\
         \midrule
      Precipitation 1h&  0.6686 (0.6670)& 0.5003 (0.4618)& 0.1683 (0.2053) & 0.1516 (0.1285) & 0.0013  &$H_0$\\
      Precipitation 3h&  0.7031 (0.7033)& 0.5004 (0.4734)& 0.2026 (0.2298) & 0.1458 (0.1256) &  0.0023  &$H_0$\\
      \bottomrule
    \end{tabular}
    \caption{Estimates for the Causality-Test and the Confounder-Test for train delays and precipitation, using current precipitation and precipitation accumulated over the previous three hours, with $k=\lfloor n^{0.4}\rfloor=27$. In the implementation of the CTC estimator, for ties in ranks we use the convention that from the interval of possible ranks each rank is assigned according to the order in which observations appear in the sample. Theoretically any value within the interval of ranks is possible and asymptotically there is no difference which convention is chosen. For comparison, we provide averages of tie-ranks in parentheses.}
    \label{tab:ctc_train}
\end{table}

\begin{table}[ht]
    \small
    \centering
    \begin{tabular}{lcccc}
        & \multicolumn{4}{c}{\textbf{Tail Index-Test}}  \\
        \cmidrule(lr){2-5}
        Cumulated Hours & $\hat{\gamma}_{precipitation}$ & $\hat{\gamma}_{delay}$ 
        & $t_{I}$  & Decision    \\
    \midrule
    Precipitation 1h & 0.3074 & 0.3812 & 1.3284 & $H_0$ \\
    Precipitation 3h & 0.1822 & 0.3926 & 7.1617 & $H_0$\\
    \bottomrule
    \end{tabular}
    \caption{Tail Index-Test results with corresponding Hill estimates.}
    \label{tab:tail_train}
\end{table}

The relationship between precipitation accumulated over the previous three hours and delay appears to be more significant than the relationship between current precipitation and train delays. In this case, the CTC estimate exceeds $0.7$, which aligns with the results of the simulation study on the finite-sample rate of convergence in Section~\ref{subsec:sim_sample_rate}. Consequently, the subsequent focus is on the relationship between delay and precipitation accumulated over the previous three hours, while the results for the current precipitation are reported in Appendix~\ref{app:add_figures}.

\autoref{fig:train} illustrates the behavior of the CTC estimates as the sample size increases, with fixed values of $k$ set to $\lfloor n^{0.4}\rfloor = 27$ and $\lfloor n^{0.5}\rfloor = 63$. We observe that $\hat{\Gamma}_{delay\rightarrow precipitation\,3h}$ approaches $0.5$ and $\hat{\Delta}_{precipitation\,3h \to delay}$ increases with the sample size. While there is evidence of causality in the tails of the data, this relationship does not appear to hold consistently across the entire distribution. The classical LiNGAM algorithm does not detect a causal link between precipitation and train delays, neither for current precipitation nor for precipitation aggregated over the past three hours. In contrast, pairwise LiNGAM identifies the correct causal direction from precipitation to delay.

\begin{figure}[htbp]
    \centering
    \begin{subfigure}[b]{0.49\textwidth}
        \centering
        \includegraphics[width=\textwidth]{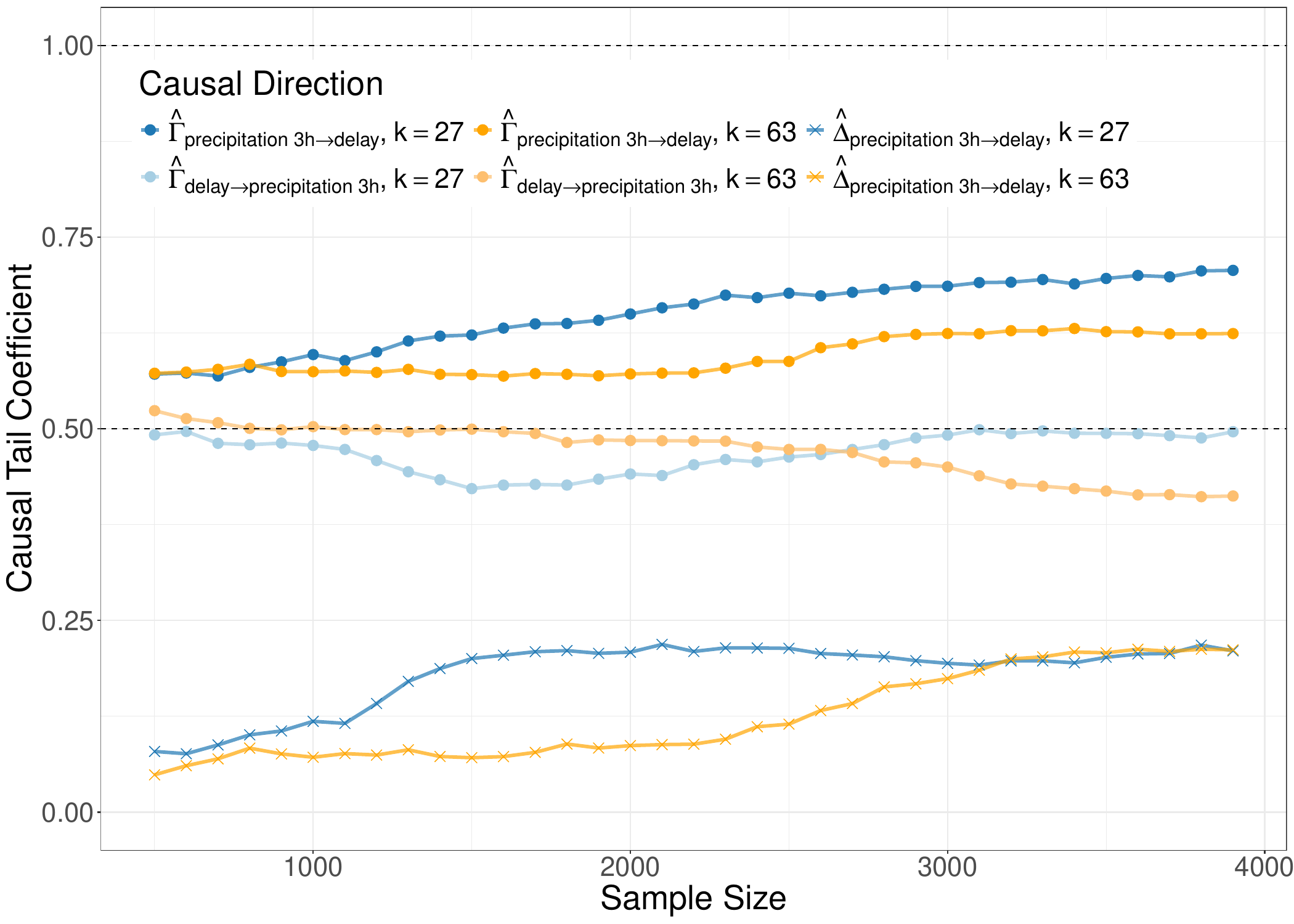}
    \end{subfigure}
    \begin{subfigure}[b]{0.49\textwidth}
        \centering
        \includegraphics[width=\textwidth]{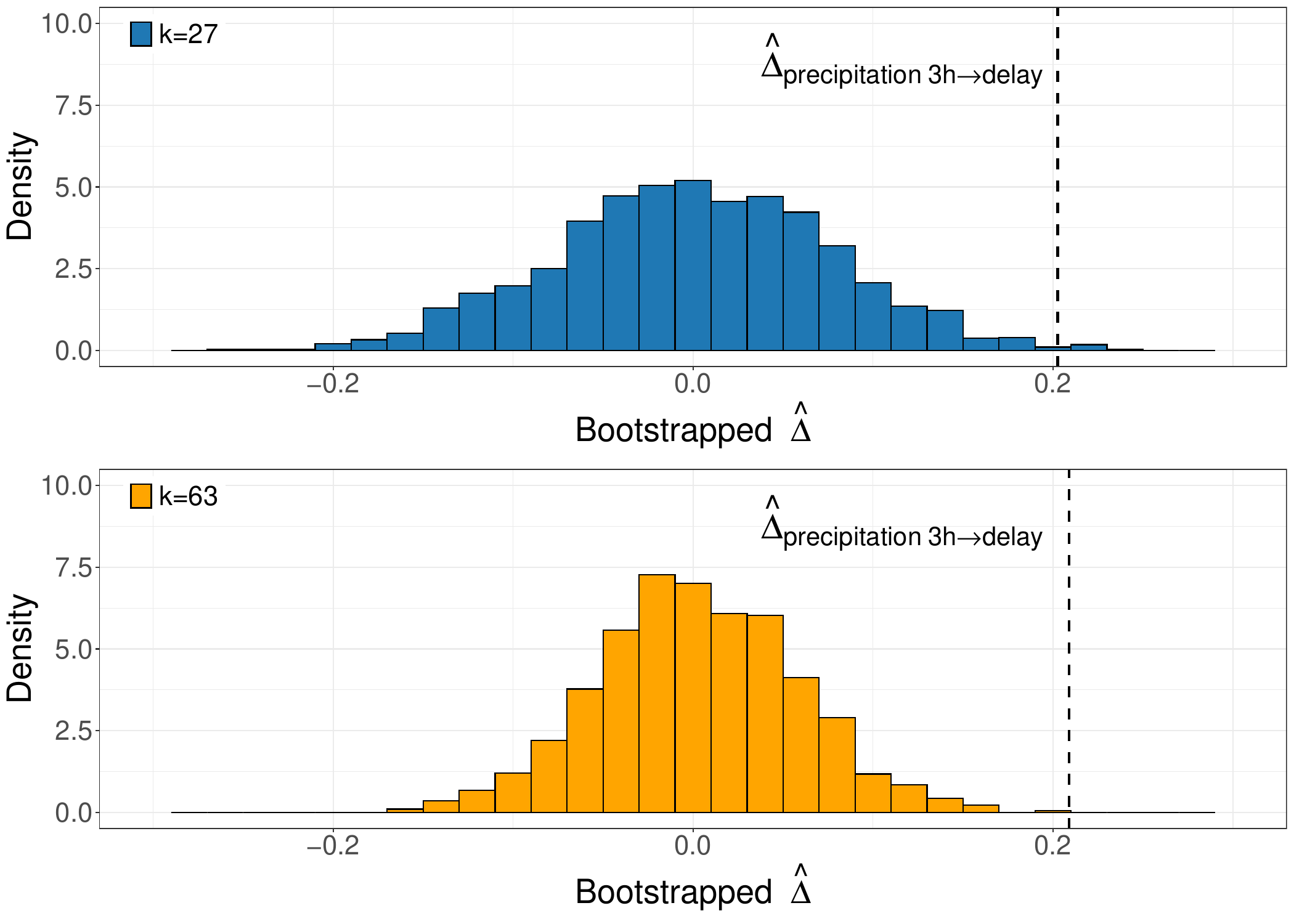}
    \end{subfigure}
    \caption{CTC estimates for the causal relation between precipitation accumulated over the previous three hours and train delays. The left panel shows the behavior of the CTC estimator for $k\in\{27,63\}$ as the sample size varies. The right panel shows the bootstrap distribution of the Causality-Test statistic together with the estimated value of $\Delta_{\mathrm{precipitation\,3h}\to\mathrm{delay}}$.
    }
    \label{fig:train}
\end{figure}

Based on the above findings, we conclude a causal relationship between extreme precipitation and extreme delays. This finding is consistent with the literature, supporting the hypothesis that extreme precipitation leads to significant delays in the train connection from Zurich to Bern. 

\subsubsection*{Precipitation and River Flows}

In the next application, the tail tests are used to analyze a scenario in which the causal direction is known a priori, while the timing is not. In this instance, the influence of daily precipitation (sum in mm) on daily mean river discharge (in $m^3/s$) is examined, as well as an investigation of the delays of flooding impact. The observed river discharges come from \textit{Bayerisches Landesamt für Umwelt, www.lfu.bayern.de}, while the second dataset is derived from \textit{Deutscher Wetterdienst} for the daily precipitation. We consider two of the largest rivers in Germany, both of which have experienced severe flood events causing substantial damage in the past \citep{bloschl2013june, merz2014extreme}. The first is the Danube, observed at the gauging station Passau/Ilzstadt with the corresponding weather station Passau-Oberhaus. The second is the Main, observed at the gauging station in Würzburg, where the associated weather station is also located.

Based on a preliminary data analysis, we identify May to August as the months with the highest precipitation levels, consistent with the first application scenario. For the Danube, the dataset comprises $3{,}072$ observations collected between May 1999 and August 2023. For the Main, $2{,}948$ observations are available for the period from May 2000 to August 2023.

For simplicity and to illustrate the use of the causal tail tests, we assume that a single nearby weather station captures the relevant extreme precipitation signal for river discharge. Previous studies by \citet{pasche2023causal}, \citet{mhalla2020causal}, and \citet{gnecco2021causal} have demonstrated that upstream river basins exert a significant influence on the downstream river discharge. We observe the same effect at a later stage of this analysis. Accordingly, we include the upper Main river catchment as a confounder in the estimation. Moreover, our analysis focuses on the effect of precipitation on the previous day on river discharge on the following day. A more detailed investigation of the time lag between precipitation and river discharge is presented later. Given the proximity of the weather station to the discharge measuring station in both cases, we assume that a 24-hour period is sufficient for precipitation-driven runoff to reach the river.

\begin{table}[ht]
    \small
    \centering
    \begin{tabular}{lcccccc}
          & \multicolumn{4}{c}{\textbf{Causality-Test}} & \multicolumn{2}{c}{\textbf{Confounder-Test}} \\
        \cmidrule(lr){2-5}  \cmidrule(lr){6-7}
     River & $\hat{\Gamma}_{precipitation\rightarrow river}$   & $\hat{\Gamma}_{river \rightarrow precipitation}$ & $\hat{\Delta}_{precipitation\rightarrow river}$ & $\Delta_{0.05}$ &$t_{C}$  & Decision\\
      \midrule
      Danube&  0.7359 (0.7355)& 0.4642 (0.5298)&0.2717 (0.2057) & 0.1971 (0.1585)& -0.1752  & $H_0$\\
      Main&  0.8559 (0.8543)& 0.6700\, (0.6807)& 0.1860 (0.1736) & 0.1373 (0.1272)& 0.8853   & $\boldsymbol{H_1}$\\
      \bottomrule
    \end{tabular}
    \caption{Estimates and test results for the Causality-Test and Confounder-Test for precipitation and river discharge, using $k=\lfloor n^{0.4}\rfloor=24$ observations for both the Danube and the Main. For tie ranks in the CTC estimate the same convention as in Table \ref{tab:ctc_train} applies with average ranks in parentheses.}
    \label{tab:ctcs_river}
\end{table}

\begin{table}[ht]
    \centering
    \small
    \begin{tabular}{lcccc}
        & \multicolumn{4}{c}{\textbf{Tail Index-Test}} 
         \\
        \cmidrule(lr){2-5}
        River & $\hat{\gamma}_{precipitation}$ & $\hat{\gamma}_{river}$ 
        & $t_{I}$  & Decision   \\
        \midrule
        Danube   & 0.0879  & 0.2195  & 316.4025  & $\boldsymbol{H_1}$   \\
        Main     & 0.1038  & 0.3548  & 205.7330  & $\boldsymbol{H_1}$   \\
        \bottomrule
    \end{tabular}
    \caption{Tail Index-Test results with corresponding Hill estimates.}
    \label{tab:tails_riverflow}
\end{table}

\autoref{tab:ctcs_river} shows that, for both rivers, $\hat{\Delta}_{precipitation \to river} > \Delta_{0.05}$, allowing us to reject the null hypothesis of the Causality-Test. However, for the Main, the Confounder-Test is also significant, indicating the possible presence of a heavy-tailed confounding factor. We assess this further using the Tail Index-Test reported in \autoref{tab:tails_riverflow}, which rejects the null hypothesis of equal tail indices, with $\hat{\gamma}_{river} > \hat{\gamma}_{precipitation}$. Since the Hill estimator estimates $1/\alpha$, this suggests that river discharge has a heavier estimated tail than precipitation. We therefore conclude that there is evidence of a causal tail effect from precipitation to the Danube. For the Main, however, potential confounding must be taken into account. To address this, we control for the upstream Main catchment at Schweinfurt Neuer Hafen and apply the adjusted CTC estimator $\hat{\Gamma}_{X_1\to X_2\mid H}$ from \eqref{eq:ctc_confounder}. The resulting test outcomes are reported in \autoref{tab:ctcs_river_conf}.

\begin{table}[ht]
    \small
    \centering
    \begin{tabular}{cccccc}
         \multicolumn{4}{c}{\textbf{Causality-Test with confounder adjustment}} & \multicolumn{2}{c}{\textbf{Confounder-Test}} \\
        \cmidrule(lr){1-4}  \cmidrule(lr){5-6}
    $\hat{\Gamma}_{prec.\rightarrow Main | Main\, up}$   & $\hat{\Gamma}_{Main \rightarrow prec. | Main\, up}$ & $\hat{\Delta}_{prec.\rightarrow Main | Main\, up}$ & $\Delta_{0.05}$ &$t_{C}$  & Decision\\
          \midrule
 0.7774 (0.7773) &0.4388 (0.4338) &0.3387 (0.3435)   &0.1719 (0.1722) &-0.3000 & $H_0$\\
 \bottomrule
    \end{tabular}
    \caption{Estimates and test results for the Causality-Test with confounder adjustment and the Confounder-Test for the Main. CTC estimates are computed using $k=\lfloor n^{0.4}\rfloor=24$ observations. For tie ranks in the CTC estimate the same convention as in Table \ref{tab:ctc_train} applies with average ranks in parentheses.}
    \label{tab:ctcs_river_conf}
\end{table}

Using the adjusted CTC, we no longer reject the null hypothesis in the Confounder-Test, while the Causality-Test remains significant. This leads us to conclude that precipitation has a causal influence on the Main catchment.

These findings are further supported by \autoref{fig:donau} and \autoref{fig:main}, which show a clear upward trend in $\hat{\Delta}_{precipitation \to river}$ and convergence of $\hat{\Gamma}_{river \to precipitation}$ toward $0.5$. For the Main, we use the adjusted CTC version given in \eqref{eq:ctc_confounder}.

\begin{figure}[ht]
    \centering
    \begin{subfigure}[b]{0.49\textwidth}
        \centering        \includegraphics[width=\linewidth]{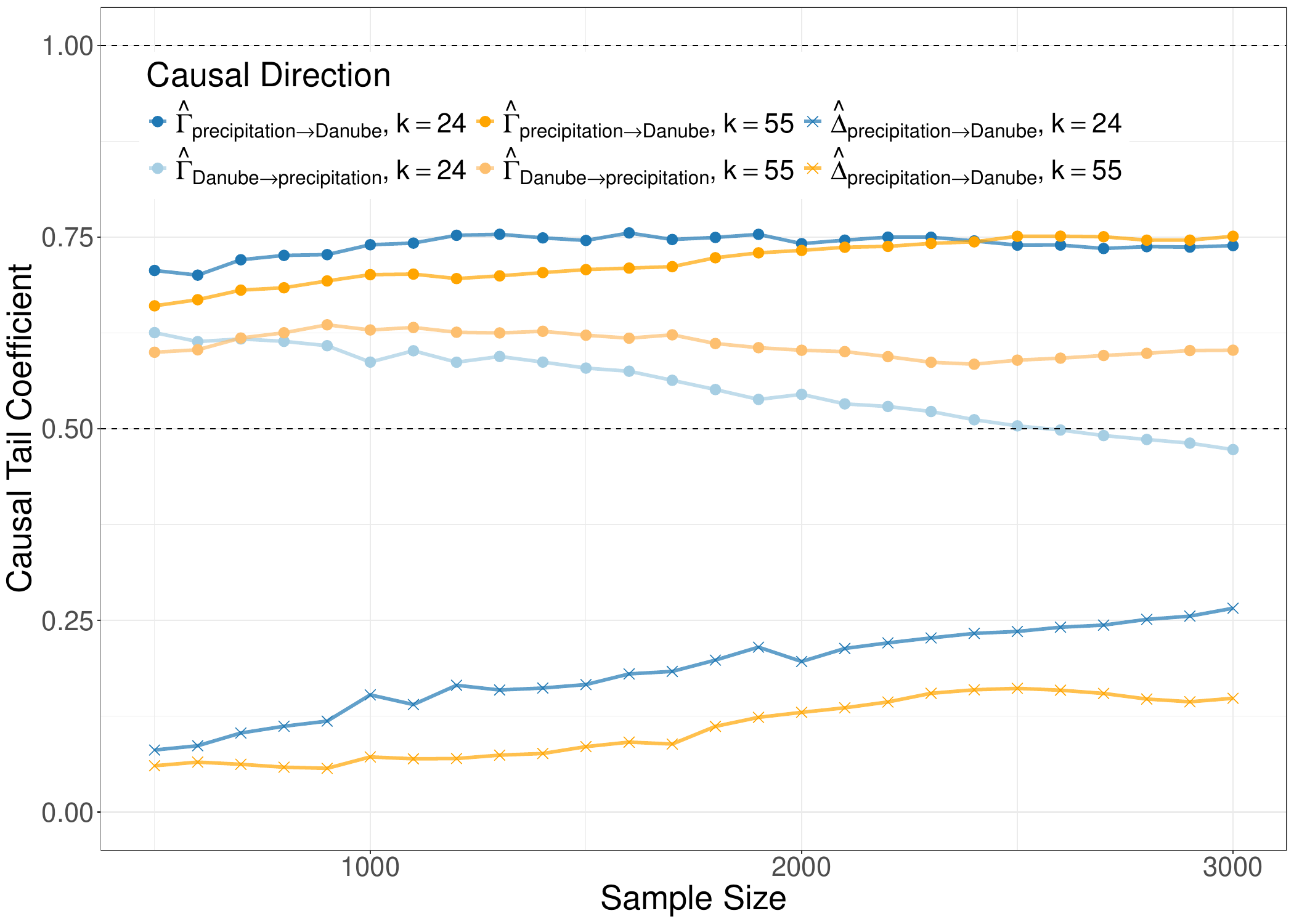} 
    \end{subfigure}
    \begin{subfigure}[b]{0.49\textwidth}
        \centering
        \includegraphics[width=\linewidth]{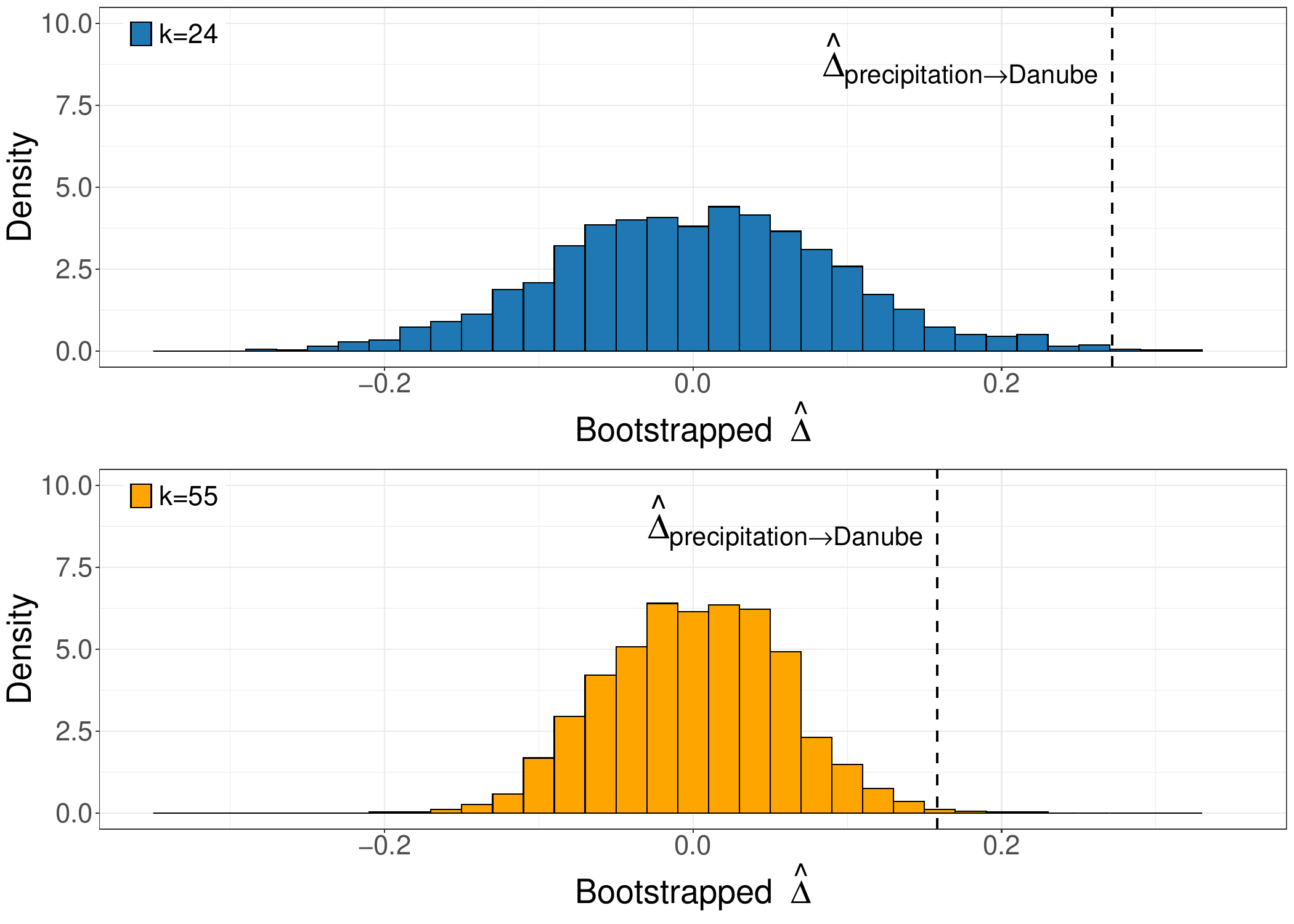}
    \end{subfigure}
    \vskip\baselineskip
    \caption{CTC estimates for the causal relation between precipitation at the Passau weather station and Danube discharge at Passau/Ilzstadt. The left panel shows the behavior of the estimator for $k\in\{24,55\}$ as the sample size varies. The right panel shows the bootstrap distribution of the Causality-Test statistic together with the estimated value of $\Delta_{precipitation\to Danube}$.
    }
    \label{fig:donau}
    \begin{subfigure}[b]{0.49\textwidth}
        \centering
        \includegraphics[width=\linewidth]{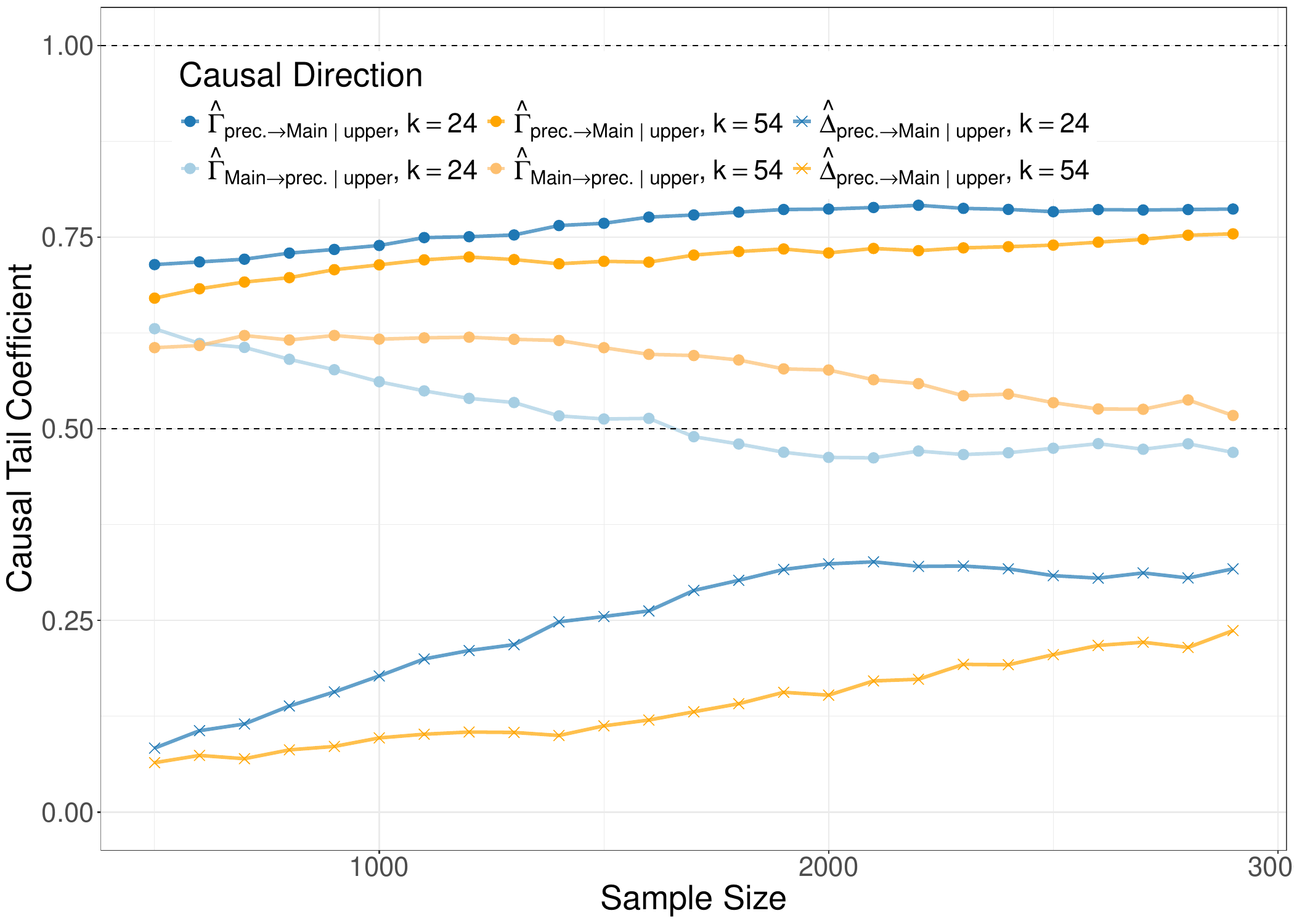} 
    \end{subfigure}
    \begin{subfigure}[b]{0.49\textwidth}
        \centering
      \includegraphics[width=\linewidth]{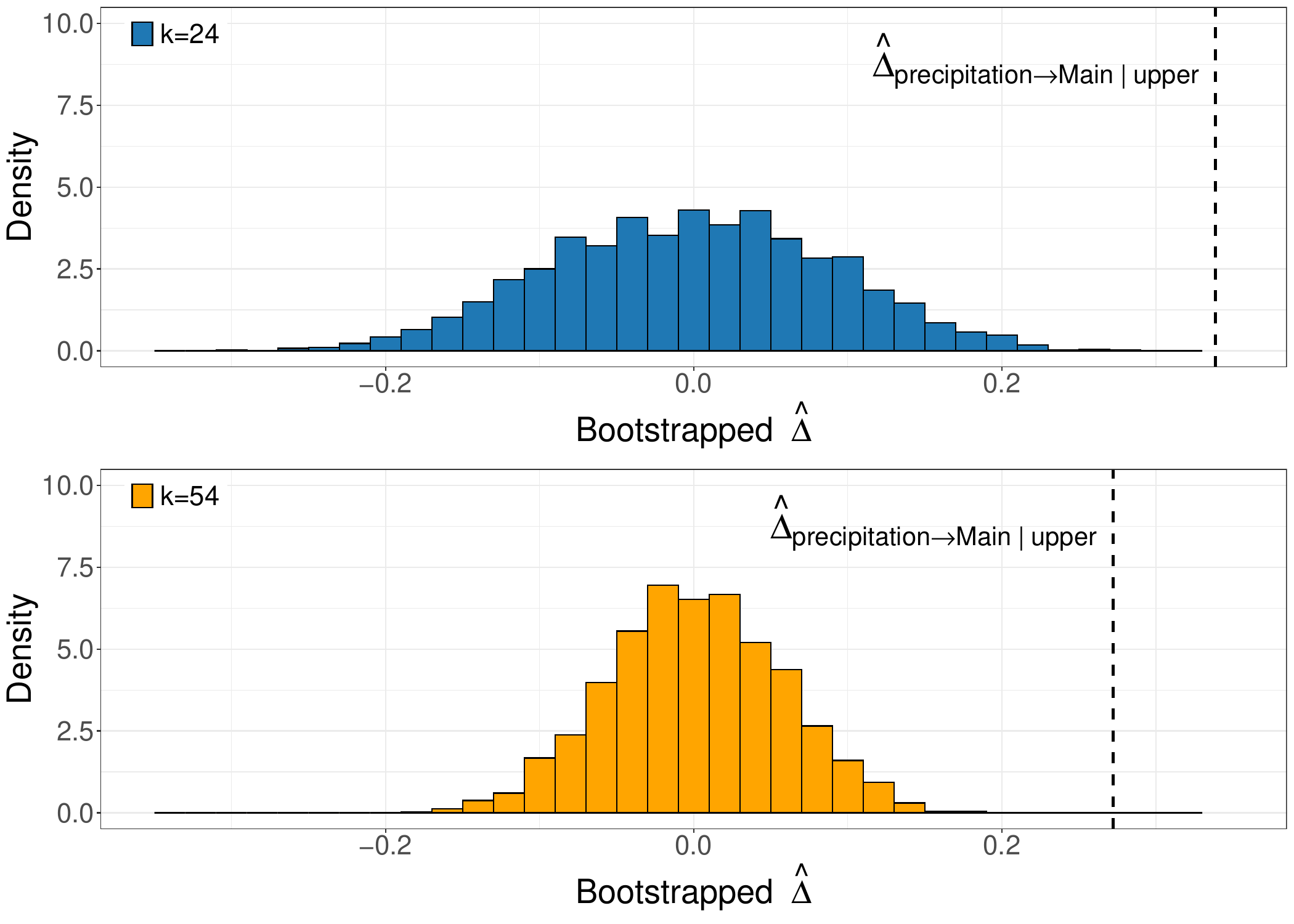} 
    \end{subfigure}
    \vskip\baselineskip
    \caption{Adjusted CTC estimates for the causal relation between precipitation at the Würzburg weather station and Main discharge, controlling for the upstream Main catchment. The left panel shows the behavior of the estimator for $k\in\{24,54\}$ as the sample size varies. The right panel shows the bootstrap distribution of the Causality-Test statistic together with the estimated value of $\Delta_{precipitation \to Main}$.}
    \label{fig:main}
\end{figure}

\begin{table}[ht]
    \small
    \centering
\begin{tabular}{cccccccc}
        & \multicolumn{4}{c}{\textbf{Causality-Test}} 
        & \multicolumn{2}{c}{\textbf{Confounder-Test}} & \\
        \cmidrule(lr){2-5} \cmidrule(lr){6-7}
Time Lag& $\hat{\Gamma}_{prec. \rightarrow river}$   & $\hat{\Gamma}_{river \rightarrow prec.}$ & $\hat{\Delta}_{prec. \rightarrow river}$ & $\Delta_{0.05}$ & $t_{C}$ & Decision & LiNGAM\\
\midrule
    \multicolumn{8}{c}{Danube} \\
\midrule
1& 0.7359 & 0.4642 & \textbf{0.2717} & 0.1971 & -0.1752& $H_0$ &25.5509\\
2& 0.7544 & 0.5763 & 0.1780 & 0.1796 & 0.3740  & $H_0$ &29.0021\\
3&0.7260 & 0.6298 & 0.0962 & 0.1819 & 0.6359 & $H_1$ &25.3989\\
4&0.6866 & 0.6841 & 0.0025 & 0.1732 & 0.9020 & $H_1$ &21.2748\\
5&0.6726 & 0.6531 & 0.0195 & 0.1809 & 0.7499 & $H_1$ &17.5819 \\
6&0.6390 & 0.6351 & 0.0039 & 0.1727 & 0.6617 &$H_1$ &13.9217\\
7&0.6575 & 0.5573 & 0.1002 & 0.1830 & 0.2806  &$H_0$ &12.6842\\
8&0.6484 & 0.5594 & 0.0890 & 0.1781 & 0.2912 &$H_0$ &12.4506\\
\midrule
    \multicolumn{8}{c}{Main with confounder adjustment} \\
\midrule
1 & 0.7774 & 0.4388 & \textbf{0.3387} & 0.1719 & -0.3000& $H_0$& 2.1099\\
2 & 0.7617 & 0.6013 & 0.1604 & 0.2136 &0.4963& $H_0$& 3.0330\\
3 & 0.7567 & 0.5100 & \textbf{0.2467} & 0.1819 &0.0490& $H_0$ &2.9874\\
4 & 0.7260 & 0.7888 & -0.0628 & 0.1805 &1.4148& $H_1$& 2.5189\\
5 & 0.7169 & 0.8474 & -0.1304 & 0.2041  &1.7019& $H_1$& 2.0399\\
6 & 0.6554 & 0.7905 & -0.1351 & 0.1922&1.4232& $H_1$& 1.6722\\
7 & 0.6286 & 0.7977 & -0.1690 & 0.2050 &1.4584& $H_1$& 0.0000\\
8 & 0.6076 & 0.6045 & 0.0031 & 0.2127 &0.5119& $H_1$& 0.0000\\
\bottomrule
    \end{tabular}
    \caption{Test results for the Causality-Test and Confounder-Test, together with estimated structural weights from LiNGAM, for different time lags between precipitation and river discharge. Bold values indicate cases in which $\hat{\Delta}_{prec. \to river}$ exceeds the critical value.}
    \label{tab:riverflows_lags}
\end{table}

After analyzing the causal relationship between precipitation on the previous day and river discharge on the following day, we now examine causal links across various time lags between precipitation and river discharge. We apply the CTC-based testing strategy, using the adjusted version in \eqref{eq:ctc_confounder} for the Main, and compare the results with those obtained from the LiNGAM algorithm, which captures causal structure in the full distribution. The results are given in \autoref{tab:riverflows_lags}, while results of the Tail Index-Test remain unchanged from \autoref{tab:tails_riverflow}. 

As the lag increases, the discrepancy between $\hat{\Gamma}_{precipitation\rightarrow river}$ and $\hat{\Gamma}_{river \rightarrow precipitation}$ diminishes. For the Main, the estimates of $\Delta_{precipitation \to river}$ even become negative at higher lags. While LiNGAM detects causal connections between one and eight days for the Danube and one and six days for the Main, we also observe that its causal effect decreases with increasing lag. We therefore conclude that, although an effect of precipitation on river discharge may still be observed several days later in the full distribution, this persistence is less evident in the extreme tail, where the effect appears to diminish more quickly.

Overall, both the existence and direction of causality are demonstrated in this example using the tail tests. However, the results obtained for lags exceeding one do not closely align with the theoretical findings presented in Section~\ref{sec:theory}. Potential reasons for this discrepancy may include one or more of the following. In contrast to other estimation frameworks, which use a gridded precipitation dataset to predict river discharge (\cite{mcmillan2010impacts}, \cite{muller2003estimation}), our study observes precipitation at a single weather station. Moreover, it is important to note that extreme river discharges do not typically result from a single high precipitation day. Instead, they are frequently the result of a series of days with elevated precipitation (\cite{tuel2021climatology}, \cite{doswell1996flash}), as evidenced in \autoref{tab:riverflows_lags}. This leads to dependent observations, thereby violating the assumption of independence. The different days of precipitation confound each other, where these confounders have the same tail indices as the cause but different causal coefficients (compare with \eqref{eq:ctc_val_conf}). Furthermore, the high uncertainty observed in \autoref{fig:donau} and \autoref{fig:main} for $\hat{\Delta}_{precipitation \to river}$ may be a contributing factor. 

Conversely, the estimates of the two directions for a lag of one day are significantly disparate, thereby underlining the assumption of a causal effect between the extreme precipitation measurements at the selected weather stations and the river discharges nearby. This effect is also observed when the lag exceeds one day, a finding that is corroborated by the results obtained using LiNGAM. This indicates that, in this example, causality is present not only in the tails of the data but throughout the entire dataset. Furthermore, the work of \citet{beven2012rainfall} provides a detailed explanation of how the catchment area influences the relationship between rainfall and river discharge, clarifying the different results observed for the Danube and the Main.

\subsection{Financial Extremes}
\subsubsection*{Financial Stock Markets and Cryptocurrencies}
Subsequent to the initial two applications that demonstrate the capacity of the causal tail tests to recover the causal structure of applications where tails are disparate and confounders are present, an application is now analyzed where the causal direction is not clear from the beginning. Therefore, the open and close prices of the S\&P 500 from \textit{Yahoo Finance}, as well as Bitcoin (BTC) prices from \textit{coinmetrics.io}, are observed. The influence of S\&P 500 on Bitcoin and vice versa has been previously examined in various studies. \cite{wang2020relationship} demonstrate in their work that the influence of S\&P 500 on Bitcoin is stronger than the influence of Bitcoin on S\&P 500, based on a Vector Autoregressive (VAR) model. Moreover, \cite{muglia2019bitcoin} utilize Dynamic Model Averaging to assert that Bitcoin has no predictive power over the S\&P 500 index, while \cite{kjaerland2018analysis} report a positive impact of S\&P 500 on the price of Bitcoin, employing Autoregressive Distributed Lag and Generalized Autoregressive Conditional Heteroscedasticity (GARCH) models. \cite{nguyen2022correlation} uses a VAR-GARCH model and focuses on periods of uncertainty, such as the time of the pandemic, and observed a significant influence of S\&P 500 returns on Bitcoin returns. The study also demonstrates that the stock market and Bitcoin are more correlated during periods of turmoil. In contrast to the extant literature, which has focused on the entire data set, we analyze the tails of the data to extract the causal relationship.

In the following, we apply the causal tail testing strategy to the return data, with the objective of comparing the outcomes with the extant literature. The data is analyzed for the period from 18 July 2010 to 4 December 2024, with Bitcoin prices recorded daily while S\&P 500 stocks are only traded on working days. From the raw data the log returns of S\&P 500 and Bitcoin are extracted for each time step $i$ with 
\begin{equation*}
    return(i)= \log\left(\frac{close\,price (i)}{open\, price(i)}\right),\, i=1,\ldots,n.
\end{equation*}
 To account for the temporal dependence, we fit a time-series model to the data and work with the residuals as proposed in \cite{mcneil2015quantitative}. For the S\&P 500, we use an ARMA-GARCH model, since it captures key stylized features of financial time series and is a standard model for index data. For Bitcoin returns, we use a model from the same ARMA-GARCH family to treat both series consistently within the same modeling framework\footnote{While other models may better capture the dynamics of cryptocurrencies than GARCH-type models (see for example \citet{chu2017garch} and \citet{catania2019forecasting}), we remain within the GARCH family to ensure methodological consistency across the analyses.}
 To analyze both variables jointly, we remove Bitcoin residuals corresponding to weekends and public holidays on which the S\&P 500 is not traded. The final dataset comprises $3{,}620$ observations.

 Since both S\&P 500 and Bitcoin returns may be influenced by general market volatility and other factors such as broader market movements, we consider these factors as observed confounders to work with the adjusted CTC from \eqref{eq:ctc_confounder}. Specifically, we consider as the confounding vector $H$ the residuals of the CBOE Volatility Index (VIX) from \textit{Yahoo Finance} and the residuals of the MSCI Europe Index (MSCI Europe) from \cite{msci2024}. The residuals are computed using the same procedure outlined earlier. Since the VIX represents the future stock market volatility over the next 30 days \citep{whaley2009understanding}, we lag the data by one trading day, to exclude any predictive power it might have to confound the relationship between S\&P 500 and Bitcoin. 

 We apply the tail tests to the residuals obtained from the fitted models. For comparison, we also apply the tests to the original data. However, this yields no substantial differences in the results.

\autoref{tab:ctc_finance} presents both the left tail and the right tail results of the Causality- and Confounder-Test. As also shown in \autoref{tab:tails_finance}, the left tail of Bitcoin is heavier than its right tail and generally the left tail is also rather the focus of financial risk assessment. According to Assumption \ref{assumption1}, our theoretical setup therefore focuses on the left tail and applies directly in this setting. For completeness, however, we conduct the analysis for both tails to investigate whether the inferred causal patterns differ between positive and negative extremes. 

For the right tail, the results in \autoref{tab:ctc_finance} indicate a significant causal tail effect from the S\&P 500 to Bitcoin for extreme positive returns (for more details also see Appendix~\ref{app:add_figures}.). For the left tail, however, the Confounder-Test rejects the null hypothesis, while the Causality-Test does not reveal a significant asymmetry, so the causal direction cannot be empirically identified from the plain Causality-Test. The Tail Index-Test reported in \autoref{tab:tails_finance} shows that Bitcoin has a heavier left tail than the S\&P 500, which is consistent with the fact that Bitcoin is generally more volatile and exhibits extreme returns more frequently than the S\&P 500. Taken together with the Confounder-Test results, this suggests the presence of a heavy-tailed confounder in the left tail.

The subsequent analysis with observed proxies of confounders focuses on the relationship between extreme negative returns of the S\&P 500 and Bitcoin. 
\begin{table}[ht]
    \small
    \centering
    \begin{tabular}{lcccccc}
      & \multicolumn{4}{c}{\textbf{Causality-Test}}   & \multicolumn{2}{c}{\textbf{Confounder-Test}}\\
        \cmidrule(lr){2-5} \cmidrule(lr){6-7}
Tail & $\hat{\Gamma}_{S\&P\,500\rightarrow BTC}$   & $\hat{\Gamma}_{BTC \rightarrow S\&P\,500}$ & $\hat{\Delta}_{S\&P\,500\rightarrow BTC}$ & $\Delta_{0.05}$&$t_{C}$  & Decision \\
       \midrule
Right & 0.7689 &  0.5088 & 0.2601& 0.1550& 0.0449 & $H_0$ \\
Left& 0.6898 & 0.6036 &  0.0862 & 0.1856&  0.5283& $H_1$ \\
  \bottomrule
    \end{tabular}
    \caption{Estimates and test results for the Causality-Test and Confounder-Test for S\&P 500 and Bitcoin returns, using $k=\lfloor n^{0.4}\rfloor=26$ observations.}
    \label{tab:ctc_finance}
\end{table}
\begin{table}[ht]
    \small
    \centering
    \begin{tabular}{lcccccc}
      & \multicolumn{4}{c}{\textbf{Tail Index-Test}} 
       \\
        \cmidrule(lr){2-5} 
        Tail & $\hat{\gamma}_{S\&P\, 500}$ & $\hat{\gamma}_{BTC}$ 
        & $t_{I}$  & Decision    \\
        \midrule
    Right &   0.0937   &     0.1957       &     29.9054 & $H_0$     \\
   Left   &    0.0804   &    0.2301       &   90.8498 &  $H_1$     \\
   \bottomrule
     \end{tabular}
    \caption{Test results for the Tail Index-Test with corresponding Hill estimates.}
    \label{tab:tails_finance}
\end{table}
In particular, we consider the MSCI Europe and the VIX as potential confounders for the left tail and report the corresponding results in \autoref{tab:ctc_finance_conf}. For the MSCI Europe, $\hat{\Delta}_{S\&P\,500 \to BTC \mid H}$ exceeds $\Delta_{0.05}$, while the Confounder-Test no longer rejects the null hypothesis, indicating that MSCI Europe is a suitable choice for the confounding variable. Additional insights into the finite-sample behavior of the adjusted CTC estimates with MSCI Europe as the confounder are provided in \autoref{fig:fiance_CTCs_left_tail}. The figure shows that $\hat{\Delta}_{S\&P\,500 \to BTC \mid MSCI\,EU}$ increases with sample size, further supporting the existence of a causal relationship from extreme negative S\&P 500 returns to Bitcoin returns after adjustment for confounding. 
\begin{table}[ht]
    \small
    \centering
    \begin{tabular}{lcccccc}
         & \multicolumn{4}{c}{\textbf{Causality-Test with confounder adjustment}} 
        & \multicolumn{2}{c}{\textbf{Confounder-Test}} \\
        \cmidrule(lr){2-5} \cmidrule(lr){6-7}
    Confounder $H$&$\hat{\Gamma}_{S\&P\,500\rightarrow BTC|H}$   & $\hat{\Gamma}_{BTC \rightarrow S\&P\,500|H}$ & $\hat{\Delta}_{S\&P\,500\rightarrow BTC|H}$ & $\Delta_{0.05}$ &$t_{C}$  & Decision\\
    \midrule
    MSCI EU & 0.7349 & 0.5224 & 0.2125 &0.1678 & 0.1142 & $H_0$\\
    VIX & 0.7048&0.5500 &0.1549 & 0.1834 & 0.2548 & $H_0$\\
    \bottomrule
    \end{tabular}
    \caption{Test results for the Causality-Test with confounder adjustment and the corresponding Confounder-Test for the left tail of S\&P 500 and Bitcoin returns.}
    \label{tab:ctc_finance_conf}
\end{table}
In this setting, the classical LiNGAM algorithm does not detect any causal relationship. However, the pairwise LiNGAM approach suggests that, if a causal link is present, it runs from the S\&P 500 to Bitcoin, which is consistent with our findings.

Overall, the results of our testing strategy determine a causal relationship between extreme negative returns of the S\&P 500 and negative Bitcoin returns once confounding effects are taken into account. This suggests that adverse shocks in the S\&P 500, which may reflect broader economic downturns or market-wide stress, spill over into the cryptocurrency market. This interpretation is consistent with the existing literature. As shown in Appendix~\ref{app:add_figures}, analogous results are obtained for extreme positive returns, although in that case confounding effects appear to be less pronounced.
\begin{figure}[ht]
    \centering
    \begin{subfigure}[b]{0.49\textwidth}
        \centering
        \includegraphics[width=\textwidth]{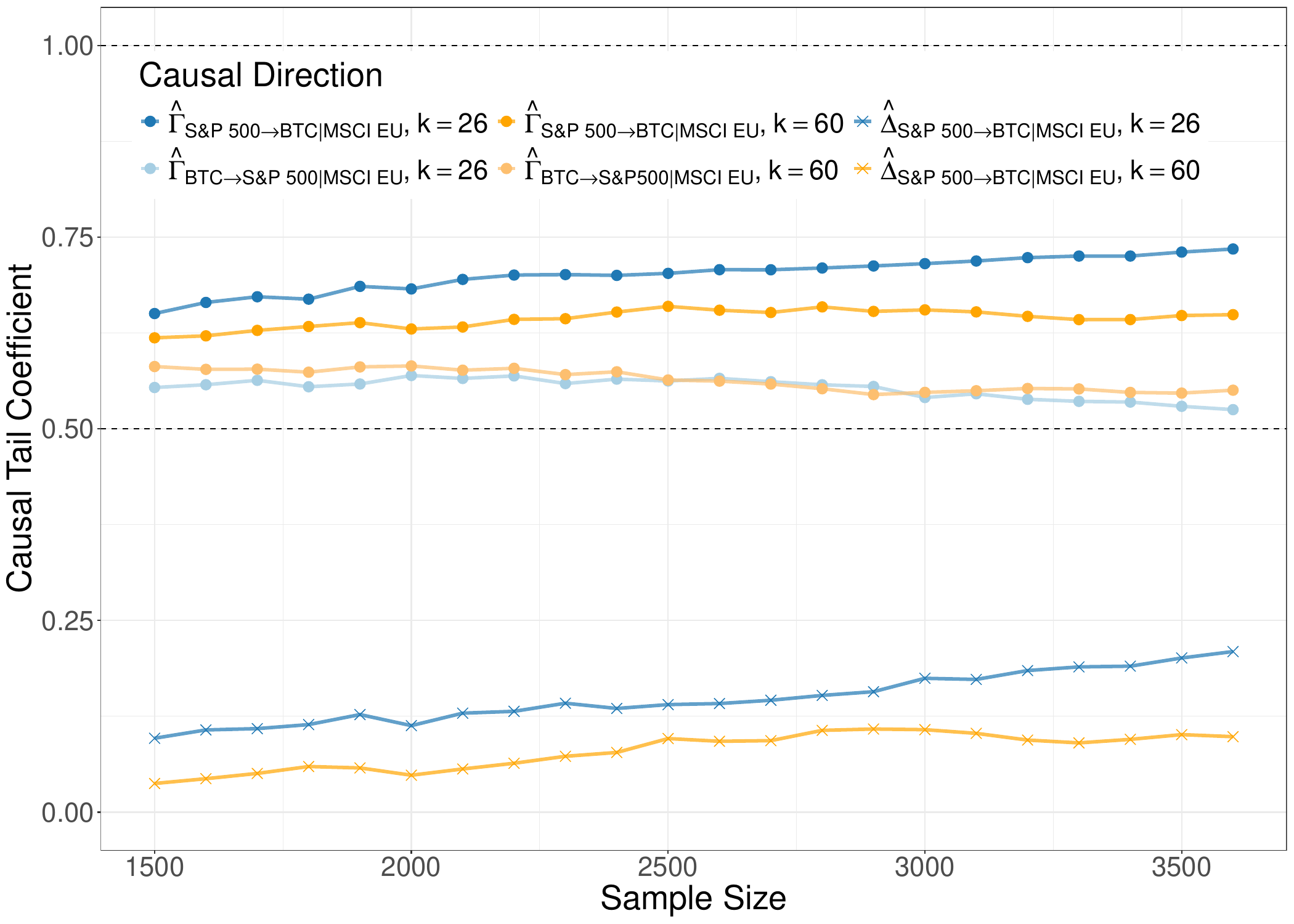}
    \end{subfigure}
    \begin{subfigure}[b]{0.49\textwidth}
        \centering
        \includegraphics[width=\textwidth]{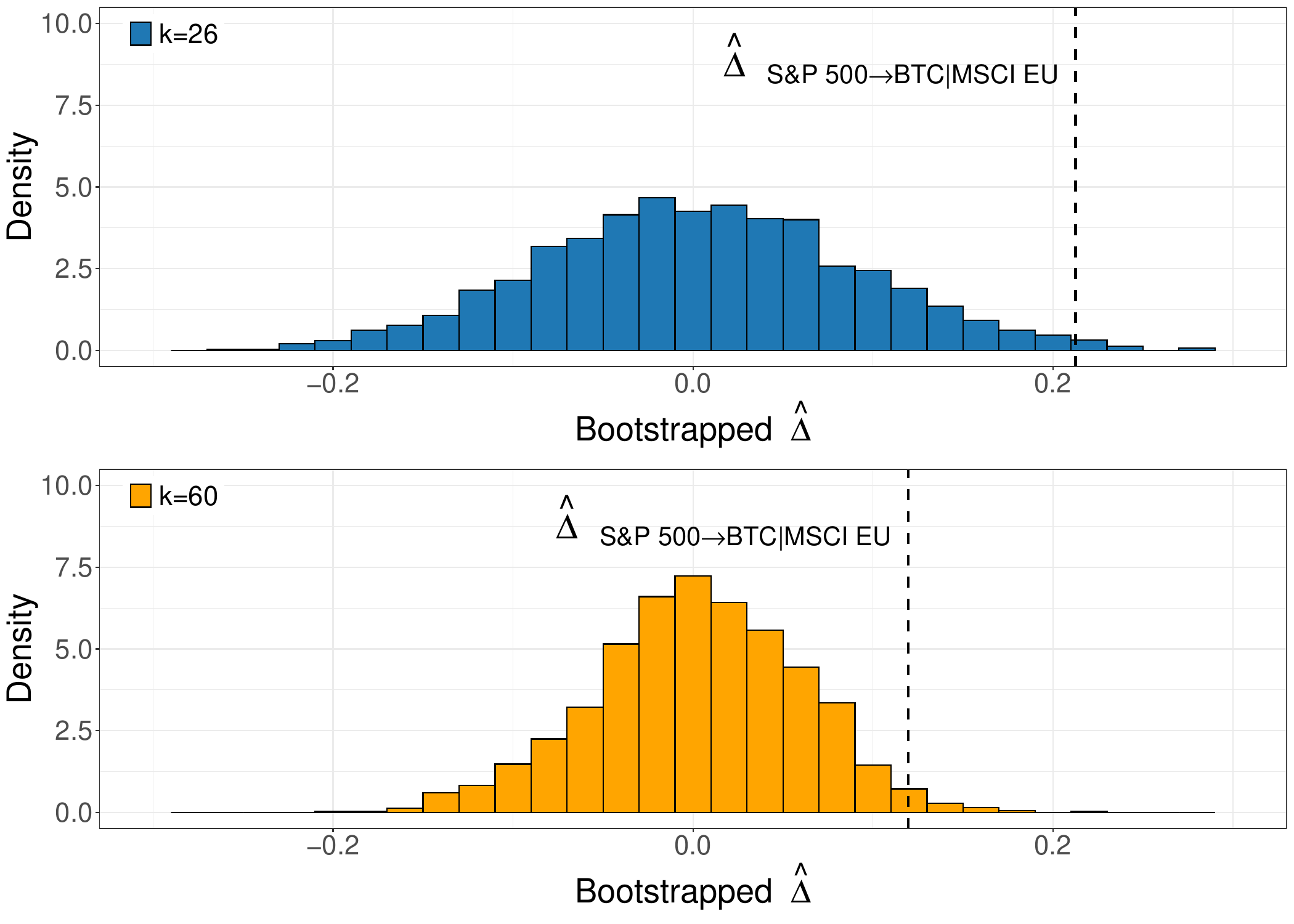}
    \end{subfigure}
    \vskip\baselineskip
    \caption{Adjusted CTC estimates for the left-tail causal relation between S\&P 500 and Bitcoin returns, using MSCI Europe as the confounder. The left panel shows the behavior of the estimator for $k\in\{26,60\}$ as the sample size varies. The right panel shows the bootstrap distribution of the Causality-Test statistic together with the estimated value of $\Delta_{S\&P\,500  \to BTC|MSCI\,EU}$.}
    \label{fig:fiance_CTCs_left_tail}
\end{figure}

\section{Conclusions}
\label{sec:conclusion}

In this paper, we provide identification conditions for causality in extremes under general conditions. In fact, heterogeneity in the tails helps to disentangle different causal tail scenarios. We then provide tools in the form of estimators and tests and derive their asymptotic properties that can be used to empirically detect causal relationships. The obtained results hold even in the presence of light-tailed confounders and can be adapted to situations where a proxy for a heavy-tailed confounder exists. The comprehensive simulation study provides extensive finite-sample evidence for the derived asymptotic results. We determine the causal relationship in three different applications and demonstrate the use of the tail tests in situations where tail indices of relevant variables generally differ. The application of the causal tail tests across three distinct empirical settings, ranging from climate-related variables to financial returns, demonstrates their usefulness for identifying causal relationships in the tails of distributions. In particular, the tests remain informative even in challenging situations in which the response variable has a heavier tail than the explanatory variable or where confounding effects are present. Moreover, the proposed approach uncovers tail-specific causal relationships in settings where classical methods such as LiNGAM do not provide conclusive results. A central advantage of the framework is the Confounder-Test, which allows us to detect the presence of heavy-tailed confounding and, when necessary, to apply the confounder-adjusted version of the CTC. This makes it possible to recover meaningful causal relationships even in situations where the existence of confounding is not evident a priori. The validity of the approach is supported by applications in settings with a known causal direction, while its broader potential is illustrated by an additional application in which the direction is less clear and confounding effects are substantial.

Beyond identifying causal links, in the applications the proposed testing strategy also provides insight into temporal dynamics. In the river flow application, for example, it captures the decline in causal influence as the time lag between precipitation and discharge increases. This highlights that the method can be used not only to infer causal direction but also to study how causal effects vary across lags. At the same time, the analysis points to important limitations. The causal tail tests may be less reliable in small samples or in situations with strongly asymmetric confounding effects. Our results also show that graphical diagnostics can be highly informative alongside the formal test statistics. In particular, plots of the finite-sample behavior of the CTC estimates provide additional insight into convergence patterns and help to interpret the inferred causal structure more clearly.

For the sake of simplicity, we restrict attention to Linear Structural Causal Models. However, the validity of the assumption of an underlying LSCM in our chosen applications might only serve as an approximation. Nevertheless, even when the relationship between the relevant variables is nonlinear in the full distribution, it may still be approximately linear in the tails. Future work could aim to explore the theoretical properties of the CTC under nonlinear structural models, thereby broadening its applicability to more complex causal scenarios.

In this paper, we have abstracted from the time-dependence of involved variables following \citet{drees2008some} that suggests that direct nonparametric methods developed for the classical i.i.d. setting can also be useful for the tail analysis of serially dependent data. However, a next logical step would be to explicitly account for temporal dependencies and apply the testing theory to the adjusted version of CTC for time-series data from the work of \citet{bodik2024causality} on data with different tails. We leave this for future work.

While our tail tests are effective in uncovering causal directions in several empirically relevant settings, they do not cover all possible scenarios. In particular, the current framework does not fully account for asymmetric confounding effects. We attempt to detect and address such effects through the Confounder-Test and the Tail Index-Test. In light of the observed performance of the CTC under heterogeneous tails, it would be valuable to develop a test that jointly exploits the CTC estimates, their asymmetry, and the estimated tail indices.

Finally, we note that, while recent research has made progress in disentangling underlying causal structures in extreme settings, there remains a significant gap in methods for quantifying the magnitude of causal effects. Developing reliable approaches to estimate effect sizes in the tails is an important direction for future research, particularly for applications in risk management, climate science, and financial stability.

\section*{Acknowledgments}
We thank Sebastian Lerch, Uwe Ehret, Nicola Gnecco and Sam Allen, and participants of the Advances in Risk Modeling Workshop, the Workshop on Dependence Modelling, the 13th Annual Conference of the IAAE, 7th Joint Statistical Meeting of the DAG Stat and Essex University for valuable comments. Moreover, we are grateful to Jonas Bhend for supporting us with the data from MeteoSchweiz.

\section*{Funding}
Melanie Schienle gratefully acknowledges funding by the Klaus Tschira Foundation and the DFG Research Unit 5583.

\section*{Declaration of generative AI and AI-assisted technologies in the manuscript preparation process}
During the preparation of this work, the authors used ChatGPT for language editing, proofreading, improving clarity and readability, refining the presentation of mathematical arguments, and improving the presentation and design of figures and plots. The authors reviewed and edited the output as needed and take full responsibility for the content of the published article.
\clearpage

\bibliographystyle{ref.bst}
\bibliography{ref.bib}

\clearpage
\appendix
\section{Proofs}\label{app:proof}
\renewcommand{\theequation}{A.\arabic{equation}} \setcounter{equation}{0}

First, we establish the following technical Lemmas.
\begin{lemma}\label{lem:twotails} Under Assumption \ref{assumption1} it holds that
\begin{equation}
\label{eq:sum_equi}
\p(S_p>x)\sim \p(\epsilon_1>x)+\p(\epsilon_2>x)+\ldots+\p(\epsilon_p>x).
\end{equation}
\end{lemma}
\begin{proof} \cite{jessen2006regularly}, \cite{embrechts1997modelling} , Lemma 1.3.1 and \cite{feller1971introduction} p. 278.
\end{proof}

\begin{lemma}\label{lem:tail_diff}
Under Assumptions~\ref{assumption1} and~\ref{ass2}, suppose that $\alpha_1>\alpha_2$. Then, as $u\to\infty$ it holds that
\begin{equation}
\label{eq:conv_rate_rho}
    \p(\epsilon_1+\epsilon_2>u)  = \p(\epsilon_1>u)+\p(\epsilon_2>u)+ \bar F_{\epsilon_2}(u)O(u^{-\rho+\delta}),
\end{equation}
with $\rho:=\min\{1,\alpha_1-\alpha_2\}$ and for every $0<\delta<\rho$.
\end{lemma}

\begin{proof}

Define $A_u:=\{\epsilon_1+\epsilon_2>u\}$, $B_u:=\{\epsilon_1>u\}\cup\{\epsilon_2>u\}$ and let $\epsilon_1^+:=\max(0,\epsilon_1)$, $\epsilon_1^-:=\max(0,-\epsilon_1)$.
Then
\[
    |\p(A_u)-\p(B_u)| \leq \p(A_u\setminus B_u) + \p(B_u\setminus A_u).
\]
The von Mises condition in Assumption \ref{ass2} implies that there exist
constants $C>0$ and $u_0>0$ such that, for all $u\geq u_0$,
\begin{equation}
\label{eq:vonMises_result}
    \left|\frac{\bar F_{\epsilon_i}(u-h)}{\bar F_{\epsilon_i}(u)} -1 \right|
    \leq C\frac{|h|}{u},
    \qquad|h|\leq u/2,\quad i\in\{1,2\}.
\end{equation}
We first consider $A_u\setminus B_u$. On this event, $\epsilon_1\leq u$, $\epsilon_2\leq u$ and $\epsilon_1+\epsilon_2>u$. In particular, $\epsilon_1>0$. Splitting according to whether
$\epsilon_1\leq u/2$ gives
\begin{align}
\p(A_u\setminus B_u)
&\leq
\e\left(
    \left\{
        \bar F_{\epsilon_2}(u-\epsilon_1)-\bar F_{\epsilon_2}(u)
    \right\}
    \textbf{1}\{0<\epsilon_1\leq u/2\}
\right)
+
\p(\epsilon_1>u/2).
\label{eq:AminusB}
\end{align}
With \eqref{eq:vonMises_result} it follows
\[
    0\leq
    \frac{\bar F_{\epsilon_2}(u-y)-\bar F_{\epsilon_2}(u)}
         {\bar F_{\epsilon_2}(u)}
    \leq
    C\frac{y}{u},
    \qquad 0\leq y\leq u/2.
\]
Hence
\[
    \frac{\p(A_u\setminus B_u)}
         {\bar F_{\epsilon_2}(u)}
    \leq
    \frac{C}{u}
    \e\left(\epsilon_1^+
        \textbf{1}\{\epsilon_1^+\leq u/2\}
    \right)
    +
    \frac{\p(\epsilon_1>u/2)}
         {\bar F_{\epsilon_2}(u)}.
\]

We now consider the different cases for $\alpha_1$ separately.\\

1. $\alpha_1>1$: Since the right tail of $\epsilon_1$ is regularly varying with index $\alpha_1>1$, we have $\e(\epsilon_1^+)<\infty$. Hence
\[
    \frac1u
    \e\left(\epsilon_1^+\textbf{1}\{\epsilon_1^+\le u/2\}\right)
    = O(u^{-1}).
\]

2. $0<\alpha_1<1$: Since $\epsilon_1^+\ge0$
\[
    \e\left(\epsilon_1^+\textbf{1}\{\epsilon_1^+\le t\}\right)
    \le \int_0^t \p(\epsilon_1>y)\,dy.
\]
By Karamata's theorem (\cite{resnick2007heavy}, Theorem 2.1) we obtain
\[
   \e\left(\epsilon_1^+\textbf{1}\{\epsilon_1^+\le u/2\} \right)
    = O\left(u^{1-\alpha_1}\ell_1(u)\right),
\]
and by Potter’s bounds (\cite{resnick2007heavy},  Proposition 2.6) for every $\delta>0$
\[
    \frac1u\e\left(
        \epsilon_1^+\textbf{1}\{\epsilon_1^+\le u/2\}\right)
    =O(u^{-\alpha_1+\delta}).
\]
3. $\alpha_1=1$: Potter’s bounds imply, for every $\delta>0$,
\[
\e\left(\epsilon_1^+\textbf{1}\{\epsilon_1^+\le u/2\}\right) \leq C_0+ C\int_{y_0}^{u/2}y^{-1+\delta}dy=O(u^\delta)\]
and therefore
\[
    \frac1u
    \e\left(
        \epsilon_1^+\textbf{1}\{\epsilon_1^+ \le u/2\}
    \right)
    =
    O(u^{-1+\delta}).
\]

Overall we obtain for every $\alpha_1>0$, $\delta>0$,
\[
    \frac{1}{u}
    \e\left(\epsilon_1^+
        \textbf{1}\{\epsilon_1^+\leq u/2\}
    \right)
    =
    O\left(
        u^{-\min\{1,\alpha_1\}+\delta}
    \right).
\]
Moreover,
\begin{align*}
    \frac{\p(\epsilon_1>u/2)}
         {\bar F_{\epsilon_2}(u)}=
    \frac{\ell_1(u/2)(u/2)^{-\alpha_1}}
         {\ell_2(u)u^{-\alpha_2}}=2^{\alpha_1}
    u^{-(\alpha_1-\alpha_2)}
    \frac{\ell_1(u/2)}
         {\ell_2(u)}=
    O\left(
        u^{-(\alpha_1-\alpha_2)+\delta}
    \right),
\end{align*}
where Potter’s bounds and slow variation have been used. Therefore,
since
\[
    \rho=\min\{1,\alpha_1-\alpha_2\}
    \leq \min\{1,\alpha_1\},
\]
we obtain
\begin{equation}
\label{eq:AminusB_rate}
    \p(A_u\setminus B_u)
    =
    \bar F_{\epsilon_2}(u)O(u^{-\rho+\delta}).
\end{equation}

We next consider $B_u\setminus A_u$. We have
\[
    B_u\setminus A_u
    \subseteq
    \{\epsilon_2>u,\,
      \epsilon_1+\epsilon_2\leq u\}
    \cup
    \{\epsilon_1>u,\,
      \epsilon_1+\epsilon_2\leq u\}.
\]
The second event is trivially bounded by $\p(\epsilon_1>u)=\bar F_{\epsilon_1}(u)$. Furthermore, by Potter’s bounds,
\begin{equation}
\label{eq:F1F2ratio}
    \frac{\bar F_{\epsilon_1}(u)}{\bar F_{\epsilon_2}(u)}
    =
    \frac{\ell_1(u)}{\ell_2(u)}
    u^{-(\alpha_1-\alpha_2)}
    =
    O\left(
        u^{-(\alpha_1-\alpha_2)+\delta}
    \right)
    =
    O(u^{-\rho+\delta}).
\end{equation}

For the first event,
\[
    \p(\epsilon_2>u,\, \epsilon_1+\epsilon_2\leq u)
    \leq
    \e(\bar F_{\epsilon_2}(u)-\bar F_{\epsilon_2}(u+\epsilon_1^-)).
\]
Splitting according to $\epsilon_1^-\leq u/2$ and using \eqref{eq:vonMises_result} gives
\[
\begin{aligned}
&
\frac{
    \p(\epsilon_2>u,\, \epsilon_1+\epsilon_2\leq u)
}{\bar F_{\epsilon_2}(u)}\leq  \frac{C}{u}
    \e( \epsilon_1^-\textbf{1}\{\epsilon_1^-\leq u/2\} ) +  \p(\epsilon_1^->u/2).
\end{aligned}
\]
By Assumption \ref{assumption1}, the left tail of $\epsilon_1$ satisfies
\[
    \p(\epsilon_1^->y) = \p(\epsilon_1<-y) =    O\!\left(\p(\epsilon_1>y)\right)= O(y^{-\alpha_1}\ell_1(y)).
\]
Hence, analogously to the argument above, we obtain
\[
    \frac{1}{u} \e(\epsilon_1^-\textbf{1}\{\epsilon_1^-\leq u/2\}) = O\left(u^{-\min\{1,\alpha_1\}+\delta} \right) \quad 
    \textrm{ and }    \quad \p(\epsilon_1^->u/2)  =  O(u^{-\alpha_1+\delta}).
\]
Consequently,
\[
    \p(\epsilon_2>u,\epsilon_1+\epsilon_2\leq u)  = \bar F_{\epsilon_2}(u) O(u^{-\rho+\delta}).
\]
Together with \eqref{eq:F1F2ratio}, this gives
\begin{equation}
\label{eq:BminusA_rate}
    \p(B_u\setminus A_u)
    =
    \bar F_{\epsilon_2}(u)O(u^{-\rho+\delta}).
\end{equation}

Combining \eqref{eq:AminusB_rate} and
\eqref{eq:BminusA_rate}, we obtain
\[
    |\p(A_u)-\p(B_u)|
    =
    \bar F_{\epsilon_2}(u)O(u^{-\rho+\delta}).
\]

Finally, by independence,
\begin{align*}
\p(B_u)
&=
\p(\epsilon_1>u)
+
\p(\epsilon_2>u)
-
\p(\epsilon_1>u)\p(\epsilon_2>u)
\\
&=\bar F_{\epsilon_1}(u)+\bar F_{\epsilon_2}(u)-\bar F_{\epsilon_1}(u)\bar F_{\epsilon_2}(u).
\end{align*}
Since $\bar F_{\epsilon_1}(u) = O(u^{-\alpha_1+\delta})$, and $\rho\leq\alpha_1$, we have
\[
    \bar F_{\epsilon_1}(u)\bar F_{\epsilon_2}(u) = \bar F_{\epsilon_2}(u)O(u^{-\rho+\delta}).
\]
Therefore,
\begin{equation*}
   \p(\epsilon_1+\epsilon_2>u) = \bar F_{\epsilon_1}(u)+\bar F_{\epsilon_2}(u)
    + \bar F_{\epsilon_2}(u)O(u^{-\rho+\delta}),
    \qquad
    \rho:=\min\{1,\alpha_1-\alpha_2\}.
\end{equation*}
    
\end{proof}

\begin{lemma}\label{lem:independece}

Let $\epsilon_1, \epsilon_2$ be independent regularly varying random variables such that $\epsilon_2$ is heavier-tailed than $\epsilon_1$.

Then, for every fixed $x\in\mathbb R$ with $\p(\epsilon_1\leq x)>0$,
\begin{equation}
\label{eq:conditional_rate_u}
    \p(\epsilon_1\leq x \mid \epsilon_1+\epsilon_2>u ) -  \p(\epsilon_1\leq x)  = O\!\left(u^{-\rho+\delta}\right)
\end{equation}
for every $0<\delta<\rho$, where $\rho:=\min\left\{1,\,\alpha_1-\alpha_2\right\}.$

Consequently, if $u_n$ satisfies $\p(\epsilon_2>u_n)\sim\frac{k}{n}$ and $k=n^\nu$, then
\begin{equation}
\label{eq:conditional_rate_n}
    \p(
        \epsilon_1\leq x \mid \epsilon_1+\epsilon_2>u_n) - \p(\epsilon_1\leq x) =O\!\left(n^{-(1-\nu)\rho/\alpha_2+\delta}\right)
\end{equation}
for every $\delta>0$.
In particular,
\[
    \sqrt{k}\,
    \left|\p(\epsilon_1\leq x  \mid \epsilon_1+\epsilon_2>u_n)-\p(\epsilon_1\leq x) \right| \to 0
\]
whenever
\begin{equation}
\label{eq:nu_general_condition2}
        \nu
        <
        \frac{2\rho}
             {\alpha_2+2\rho},
        \qquad
        \rho
        =
        \min\{1,\alpha_1-\alpha_2\}.
\end{equation}
\end{lemma}

\begin{proof}
Again we use that the von Mises condition from Assumption~\ref{ass2} implies \eqref{eq:vonMises_result}. We first notice that
    \begin{align}
    \label{eq:A3_rate_split}
&\left|    \p(\epsilon_1\leq x,\epsilon_1+\epsilon_2>u) -\p(\epsilon_1\leq x,\epsilon_2>u) \right| \notag\\
\leq  &\underbrace{\p(   \epsilon_1\leq x,\, \epsilon_2\leq u,\,\epsilon_1+\epsilon_2>u )}_{(I)}
        + \underbrace{\p( \epsilon_1\leq x,\,\epsilon_2>u,\, \epsilon_1+\epsilon_2\leq u}_{(II)})
\end{align}
We first consider $(I)$. If $x\leq0$, then $(I)=0$. If $x>0$, the event in $(I)$ implies $u-x<\epsilon_2\leq u.$
Hence
\[
    (I)\leq
    \bar F_{\epsilon_2}(u-x)-\bar F_{\epsilon_2}(u).
\]
By \eqref{eq:vonMises_result} and since $x$ is fixed,
\[
    \frac{(I)}{\bar F_{\epsilon_2}(u)}=O\left(\frac{x}{u}\right)
    =
    O(u^{-1}).
\]
Thus,
\begin{equation*}
    (I)
    =
    O\left(
        u^{-1}\bar F_{\epsilon_2}(u)
    \right).
\end{equation*}
We next consider $(II)$. Let $\epsilon_1^-:=\max(0,-\epsilon_1)$. On the event
\[
    \{\epsilon_1\le x,\,
      \epsilon_2>u,\,
      \epsilon_1+\epsilon_2\le u\},
\]
we necessarily have $\epsilon_1<0$, so that $u<\epsilon_2\le u+\epsilon_1^-$. Using independence of $\epsilon_1$ and $\epsilon_2$, it follows that
\[
    (II)
    \le
    \mathbb{E}\left(\bar F_{\epsilon_2}(u)-\bar F_{\epsilon_2}(u+\epsilon_1^-) \right).
\]
Splitting according to whether $\epsilon_1^-\le u/2$ gives
\begin{align*}
\frac{(II)}{\bar F_{\epsilon_2}(u)}
&\le
 \mathbb{E}\left(\frac{\bar F_{\epsilon_2}(u)-\bar F_{\epsilon_2}(u+\epsilon_1^-)}{\bar F_{\epsilon_2}(u)}\textbf{1}\{\epsilon_1^-\le u/2\}\right)
+ \mathbb{E}\left(\frac{\bar F_{\epsilon_2}(u)-\bar F_{\epsilon_2}(u+\epsilon_1^-)}{\bar F_{\epsilon_2}(u)}\textbf{1}\{\epsilon_1^->u/2\}\right)\\
&\leq
 \mathbb{E}\left(\frac{\bar F_{\epsilon_2}(u)-\bar F_{\epsilon_2}(u+\epsilon_1^-)}{\bar F_{\epsilon_2}(u)}\textbf{1}\{\epsilon_1^-\le u/2\}\right)
+ \p(\epsilon_1^->u/2).
\end{align*}
With \eqref{eq:vonMises_result} it follows 
\[
    0\le
    \frac{\bar F_{\epsilon_2}(u)-\bar F_{\epsilon_2}(u+y)}
         {\bar F_{\epsilon_2}(u)}
    \le
    C\frac{y}{u}.
\]
Therefore,
\begin{equation*}
    \frac{(II)}{\bar F_{\epsilon_2}(u)}  \le \frac{C}{u}\e \left(\epsilon_1^-\textbf{1}\{\epsilon_1^-\le u/2\}\right) +\p(\epsilon_1^->u/2).
\end{equation*}
By Assumption \ref{assumption1} and Potter’s bounds (\cite{resnick2007heavy}, Proposition 2.6), for every $\delta>0$, $\ell_1(u)=O(u^\delta)$ as $u \to \infty$, and therefore the left tail of $\epsilon_1$ satisfies
\begin{equation}\label{eq:II_1a}
       \p\left(\epsilon_1^->u/2\right) =  O\left((u/2)^{-\alpha_1}\ell_1(u/2) \right)=O(u^{-\alpha_1}u^\delta), \textrm{ as } u\to \infty. 
\end{equation}

Analogously to the three cases considered in Lemma~\ref{lem:tail_diff} we obtain for all $\alpha_1>0$

\begin{equation*}
    \frac1u\e\left(\epsilon_1^-\textbf{1}\{\epsilon_1^-\le u/2\}\right) = O\left(u^{-\min\{1,\alpha_1\}+\delta}\right)
\end{equation*}
for every $\delta>0$. Together with \eqref{eq:II_1a}
\[ 
\frac{(II)}{\bar{F}_{\epsilon_2}(u)}=O(u^{-\alpha_1}u^\delta)+ O\left(u^{-\min\{1,\alpha_1\}+\delta}\right)=O\left(u^{-\min\{1,\alpha_1\}+\delta}\right).
\]

This yields,
\begin{equation*}
    (I)+(II)= \bar{F}_{\epsilon_2}(u)O\left(u^{-\min\{1,\alpha_1\}+\delta}\right).
\end{equation*}
Therefore it follows from \eqref{eq:A3_rate_split}
\begin{align*}
       &\p(\epsilon_1\leq x,\epsilon_1+\epsilon_2>u) =\p(\epsilon_1\leq x,\epsilon_2>u) +O\left(\bar{F}_{\epsilon_2}(u) u^{-\min\{1,\alpha_1\}+\delta}\right)\\
    =
    &\bar F_{\epsilon_2}(u) \left( \p(\epsilon_1\leq x)+O\left(u^{-\min\{1,\alpha_1\}+\delta}\right) \right)=\bar F_{\epsilon_2}(u) \left( \p(\epsilon_1\leq x)+O\left(u^{-\rho+\delta}\right) \right)
\end{align*}
with $\rho:=\min(1,\alpha_1-\alpha_2)$, and since $\rho=\min\{1,\alpha_1-\alpha_2\}\leq \min\{1,\alpha_1\}$.
With Lemma~\ref{lem:tail_diff} it follows for $\delta>0$
\begin{equation*}
\p(\epsilon_1+\epsilon_2>u) = \p(\epsilon_1>u)+\p(\epsilon_2>u)+\bar F_{\epsilon_2}(u)O(u^{-\rho+\delta}).
\end{equation*}
Since
\[
\frac{\bar F_{\epsilon_1}(u)}{\bar F_{\epsilon_2}(u)}=O\left(u^{-(\alpha_1-\alpha_2)+\delta}\right)=O\left(u^{-\rho+\delta}\right)
\]
we get
\[
\p(\epsilon_1+\epsilon_2>u) = \bar F_{\epsilon_2}(u)\left(1+O(u^{-\rho+\delta})\right).
\]
Therefore, for $0<\delta<\rho$,
\begin{align*}
\p(\epsilon_1\leq x\mid\epsilon_1+\epsilon_2>u) 
=\frac{\bar F_{\epsilon_2}(u)( \p(\epsilon_1\leq x)+O(u^{-\rho+\delta}))}{\bar F_{\epsilon_2}(u)(1+O(u^{-\rho+\delta}))}
= \frac{\p(\epsilon_1\leq x)+O(u^{-\rho+\delta})}{1+O(u^{-\rho+\delta})}
=\p(\epsilon_1\leq x)+O(u^{-\rho+\delta}).
\end{align*}

For the second part of the proof we set $k=n^\nu$ and 
\[
    u_n = \left(\frac{n}{k}\right)^{1/\alpha_2} L\left(\frac{n}{k}\right) =  n^{(1-\nu)/\alpha_2}L(n^{1-\nu})
\]
where $L$ is a slowly varying function such that 
\[
    \bar F_{\epsilon_2}(u_n)\sim \frac{k}{n}.
\]
With the preceding result,
\[
\p(\epsilon_1\leq x\mid\epsilon_1+\epsilon_2>u_n)-\p(\epsilon_1\leq x)
= O(u_n^{-\rho+\eta})
\]
for every sufficiently small $\eta>0$ and,
\[
u_n^{-\rho+\eta} =n^{-(1-\nu)(\rho-\eta)/\alpha_2}
L(n^{1-\nu})^{-\rho+\eta}.
\]
and with Potter’s bounds we obtain,
\[
\p(\epsilon_1\leq x\mid\epsilon_1+\epsilon_2>u_n)-\p(\epsilon_1\leq x)
= O\left(n^{-(1-\nu)\rho/\alpha_2+\delta}\right).
\]

Therefore, since $\sqrt{k}=n^{\nu/2}$,
\[
\sqrt{k}\,
\left|\p(\epsilon_1\leq x\mid\epsilon_1+\epsilon_2>u_n)-\p(\epsilon_1\leq x)\right|
=O\left(n^{\nu/2-(1-\nu)\rho/\alpha_2+\delta}\right).
\]
Thus the latter converges to zero whenever
\[
    \frac{\nu}{2}
    <
    \frac{(1-\nu)\rho}{\alpha_2},
\]
or equivalently,
\[
    \nu
    <
    \frac{2\rho}{\alpha_2+2\rho},\quad
    \textrm{ with }
    \rho=\min\{1,\alpha_1-\alpha_2\}.
\]

\end{proof}

\begin{corollary}\label{cor:bounded_conditional_rate}
Let the conditions of Lemma~\ref{lem:tail_diff} hold with $\alpha_1>\alpha_2$, and define $\rho:=\min\{1,\alpha_1-\alpha_2\}$. Let $Z$ be a random variable such that $(Z,\epsilon_1)$ is independent
of $\epsilon_2$. Then, for every bounded measurable function $g=g(Z,\epsilon_1)$ and every $0<\delta<\rho$,
\begin{equation}
\label{eq:bounded_conditional_rate}
    \e\left( g(Z,\epsilon_1)
        \mid \epsilon_1+\epsilon_2>u
    \right)  =   \e(g(Z,\epsilon_1))
    + O(u^{-\rho+\delta}),
    \qquad u\to\infty.
\end{equation}

Moreover, if $u_n$ satisfies
\[
    \p(\epsilon_2>u_n)\sim\frac{k}{n}, \qquad k=n^\nu,
\]
then, for every $\delta>0$,
\begin{equation*}
    \e\left(g(Z,\epsilon_1)\mid \epsilon_1+\epsilon_2>u_n\right)-\e(g(Z,\epsilon_1))
    = O\left(n^{-(1-\nu)\rho/\alpha_2+\delta}\right).
\end{equation*}
\end{corollary}

\begin{proof}
Let $A_u:=\{\epsilon_1+\epsilon_2>u\}$ and $B_u:=\{\epsilon_2>u\}$. Then similar to \eqref{eq:AminusB_rate} and \eqref{eq:BminusA_rate} from the proof of Lemma~\ref{lem:tail_diff} we get,
\[
    \p(A_u\setminus B_u) + \p(B_u\setminus A_u) = \bar F_{\epsilon_2}(u) O(u^{-\rho+\delta}).
\]
Hence, since $g$ is bounded,
\[
\begin{aligned}
\left|
\e(g(Z,\epsilon_1)\textbf{1}\{A_u\})
-
\e(g(Z,\epsilon_1)\textbf{1}\{B_u\})
\right| \leq
\|g\|_\infty \left( \p(A_u\setminus B_u) + \p(B_u\setminus A_u)\right)
=
\bar F_{\epsilon_2}(u)
O(u^{-\rho+\delta}).
\end{aligned}
\]
By independence of $(Z,\epsilon_1)$ and $\epsilon_2$,
\[
    \e(g(Z,\epsilon_1)\textbf{1}\{B_u\}) = \e(g(Z,\epsilon_1))\bar F_{\epsilon_2}(u).
\]
Furthermore, Lemma~\ref{lem:tail_diff} gives
\[
    \p(A_u) =\bar F_{\epsilon_2}(u)(1+O(u^{-\rho+\delta})).
\]
Therefore,
\[
\begin{aligned}
\e(g(Z,\epsilon_1)\mid A_u)
=
\frac{
    \bar F_{\epsilon_2}(u)
    \left(
        \e(g(Z,\epsilon_1))
        +O(u^{-\rho+\delta})
    \right)
}{
    \bar F_{\epsilon_2}(u)
   (1+O(u^{-\rho+\delta}))
} =
\e(g(Z,\epsilon_1))
+
O(u^{-\rho+\delta}),
\end{aligned}
\]
which proves \eqref{eq:bounded_conditional_rate}.

Finally, regular variation of $\bar F_{\epsilon_2}$ implies
\[
    u_n  = n^{(1-\nu)/\alpha_2}L(n)
\]
for some slowly varying $L$. Applying Potter’s bounds yields
\[
    O(u_n^{-\rho+\delta})
    =
    O\!\left(
        n^{-(1-\nu)\rho/\alpha_2+\delta}
    \right).
\]
\end{proof}

\subsection{Proof of Proposition \ref{proposition}}
\begin{proof}
For $\alpha_1=\alpha_2$, \cite{gnecco2021causal} show that the CTC can be rearranged as 
    \begin{equation}
    \label{eq:ctc_rearranged}
        \Gamma_{X_i\rightarrow X_j}=\frac{1}{2}+\frac{1}{2}\underset{x\rightarrow \infty}{\lim}\frac{\sum_{h\in A_{ij}}\beta_{h\rightarrow i}^{\alpha_h} \p(\epsilon_h>x)}{\sum_{h\in An(i,G)}\beta_{h\rightarrow i}^{\alpha_h} \p(\epsilon_h>x)}.
    \end{equation}
where  $A_{ij}$ denotes the common ancestors of $X_i$ and $X_j$. Case 3 follows directly from \eqref{eq:ctc_rearranged} and the results in \cite{gnecco2021causal}. 

Using Lemma~\ref{lem:twotails} and the relation \eqref{eq:sum_equi}, all steps in the derivation of \eqref{eq:ctc_rearranged} also go through in the unequal-tail case. Thus we get in the unequal tail case where no common confounder is present: 
\begin{equation}
    \begin{split}
      \Gamma_{X_2\rightarrow X_1}&=\frac{1}{2}+\frac{1}{2}\underset{x\rightarrow \infty}{\lim}\frac{\beta_{1\rightarrow 2}^{\alpha_1} \p(\epsilon_1>x)}{\beta_{1\rightarrow 2}^{\alpha_1} \p(\epsilon_1>x)+\beta_{2\rightarrow 2}^{\alpha_2} \p(\epsilon_2>x)}\\
      &\sim \frac{1}{2}+\frac{1}{2}\underset{x\rightarrow \infty}{\lim}\frac{\beta_{1\rightarrow 2}^{\alpha_1} \ell_1(x)x^{-\alpha_1}}{\beta_{1\rightarrow 2}^{\alpha_1} \ell_1(x)x^{-\alpha_1}+\beta_{2\rightarrow 2}^{\alpha_2} \ell_2(x)x^{-\alpha_2}},
    \end{split} \label{eq:different}
\end{equation}
where $\ell_1,\ell_2$ are slowly varying functions and where we use the property of regularly varying variables  \eqref{eq:reg_var}. Interpreting \eqref{eq:ctc_rearranged} for different tail indices $\alpha_1$ and $\alpha_2$ of $\epsilon_1$ and $\epsilon_2$ yields the results since as $x$ approaches infinity, either $x^{-\alpha_1}$ or $x^{-\alpha_2}$ dominate the result.

Hence for independent $X_1$ and $X_2$, both $\Gamma_{X_1\rightarrow X_2}$ and $\Gamma_{X_2\rightarrow X_1}$ equal 0.5, since the second term in \eqref{eq:different} is zero. In the case of $X_1\rightarrow X_2$, the value of $\Gamma_{X_1\rightarrow X_2}$ is 1, irrespective of the values of $\alpha_1$ and $\alpha_2$ since $An(1,G)=A_{12}$.

\end{proof}

\subsection{Proof of Proposition \ref{proposition_confounder} }
\begin{proof}
For $\alpha_H>\max(\alpha_1,\alpha_2)$ the CTC results of Proposition \ref{proposition}  continue to hold due to $\epsilon_1$ and $\epsilon_2$ dominating equation \eqref{eq:ctc_rearranged}, so that all terms with $\epsilon_H$ vanish in the limit.

The CTC results for the heavy-tailed confounder follow from the following consideration. For no connection between $X_1, X_2$, we get:
\begin{equation}\label{eq:ctc_val_conf_ind}
    \Gamma_{X_2\rightarrow X_1}=\frac{1}{2}+\frac{1}{2}\underset{x\rightarrow \infty}{\lim}\frac{\beta_{H\rightarrow 2}^{\alpha_H} \p(\epsilon_H>x)}{\beta_{2\rightarrow 2}^{\alpha_2} \p(\epsilon_2>x)+\beta_{H\rightarrow 2}^{\alpha_H} \p(\epsilon_H>x)}
\end{equation}
and for a causal connection $X_1\to X_2$:
\begin{equation}\label{eq:ctc_val_conf}
    \Gamma_{X_2\rightarrow X_1}=\frac{1}{2}+\frac{1}{2}\underset{x\rightarrow \infty}{\lim}\frac{\beta_{1\rightarrow 2}^{\alpha_1} \p(\epsilon_1>x)+\beta_{H\rightarrow 2}^{\alpha_H} \p(\epsilon_H>x)}{\beta_{1\rightarrow 2}^{\alpha_1} \p(\epsilon_1>x)+\beta_{2\rightarrow 2}^{\alpha_2} \p(\epsilon_2>x)+\beta_{H\rightarrow 2}^{\alpha_H} \p(\epsilon_H>x)}
\end{equation}
which in both cases equals 1 for $\alpha_H<\min(\alpha_1,\alpha_2)$. Because $\Gamma_{X_1\rightarrow X_2}=1$ in this case as well, the unconditional CTC cannot identify the causal direction. 

For the results on $\Gamma_{X_1\rightarrow X_2|H}$ we note that it is constructed in such a way that it removes the effect of the confounder.

It remains to consider the intermediate case $\min(\alpha_1,\alpha_2)\leq \alpha_H \leq \max(\alpha_1,\alpha_2)$. For no directed causal connection between
$X_1$ and $X_2$ we obtain, analogously to \eqref{eq:ctc_val_conf_ind}, and $j\neq i$ 
\[
    \Gamma_{X_i\to X_j}
    =
    \frac{1}{2}
    +
    \frac{1}{2}
    \lim_{x\to\infty}
    \frac{
        \beta_{H\to i}^{\alpha_H}
        \mathbb{P}(\epsilon_H>x)
    }{
        \beta_{i\to i}^{\alpha_i}\mathbb{P}(\epsilon_i>x)
        +
        \beta_{H\to i}^{\alpha_H}
        \mathbb{P}(\epsilon_H>x)
    }.
\]
If $\alpha_H<\alpha_i$, then $\epsilon_H$ is heavier-tailed than
$\epsilon_i$ and therefore $\Gamma_{X_i\to X_j}=1$.
If $\alpha_H>\alpha_i$, then $\epsilon_i$ is heavier-tailed than
$\epsilon_H$ and hence $\Gamma_{X_i\to X_j}=0.5$. Finally, if $\alpha_H=\alpha_i$, then
\[
    \Gamma_{X_i\to X_j}
    =
    \frac{1}{2}
    +
    \frac{1}{2}
    \frac{
        \beta_{H\to i}^{\alpha_H}
    }{
         \beta_{i\to i}^{\alpha_i}+\beta_{H\to i}^{\alpha_H}
    }
    \in
    \left(\frac{1}{2},1\right).
\]

For a causal connection $X_1\to X_2$, \eqref{eq:ctc_val_conf} gives us for $\alpha_2<\min(\alpha_1,\alpha_H)$, $\Gamma_{X_2\to X_1}=0.5$. If instead $\min(\alpha_1,\alpha_H)<\alpha_2$, then at least one of $\epsilon_1$ and $\epsilon_H$ has a strictly heavier tail than $\epsilon_2$. The same dominant term then appears
in both the numerator and denominator, and therefore $\Gamma_{X_2\to X_1}=1$. For $\alpha_2=\alpha_H\leq \alpha_1$ it follows $\Gamma_{X_2\to X_1}\in \left(0.5,1\right)$. This leads us to $c_{H1},c_{H2}, c_H'\in [0.5,1]$.

\end{proof}

\subsection{Proof of Theorem \ref{theorem:independence}}
\begin{proof}
Since $\beta_{j\to i}>0$ are fixed, scaling an innovation by
$\beta_{j\to i}$ does not change its tail index and only affects
multiplicative constants. Hence, for notational simplicity, we set
all causal coefficients equal to one.

We first consider the case of a lighter-tailed confounder and set $X_1=\epsilon_H+\epsilon_1,\, X_2=\epsilon_H+\epsilon_2$ with $\alpha_H>\max(\alpha_1,\alpha_2)$. All steps for $X_1$ apply analogously to $X_2$ and are therefore omitted. Define $\rho_2:=\min\{1,\alpha_H-\alpha_2\}$.

Further, for $t>0$, let
\[
    \mu(t):=\e[F_1(X_1)\mid X_2>t],
    \qquad
    \sigma^2(t):=\mathbb V(F_1(X_1)\mid X_2>t).
\]

By Lemma~\ref{lem:tail_diff}, Lemma~\ref{lem:independece}, and
Corollary~\ref{cor:bounded_conditional_rate}, for every bounded
measurable $g$ and sufficiently small $\delta>0$,
\[
    \e[g(X_1)\mid X_2>u]-\e[g(X_1)]  =O(u^{-\rho_2+\delta}),  \qquad u\to\infty.
\]

Let $u_n:=F_2^{\leftarrow}(1-k/n)$ be the $1-k/n$-quantile of $X_2$, so that $\p(X_2>u_n)=k/n$. Since $\alpha_H>\alpha_2$,
\[
    \p(X_2>x)\sim\p(\epsilon_2>x),
\]
and hence $\p(\epsilon_2>u_n)\sim k/n$. For $k=n^\nu$, regular variation therefore yields
\[
    \e[g(X_1)\mid X_2>u_n]-\e[g(X_1)]  = O\left(n^{-(1-\nu)\rho_2/\alpha_2+\delta}\right).
\]
Applying the preceding bound first with $g(x)=F_1(x)$ and then with $g(x)=F_1(x)^2$, and using
$F_1(X_1)\sim \mathrm{Unif}(0,1)$, gives $\mu(u_n)  = \frac12+ O\left(n^{-(1-\nu)\rho_2/\alpha_2+\delta}\right)$ and $\sigma^2(u_n)\to\frac1{12}$.

Now let $T_n:=X_{(n-k),2}$. Since $\bar F_2\in RV(-\alpha_2)$ and $k\to\infty$, $k/n\to0$, gives
\[
    \frac{T_n}{u_n}\xrightarrow{P}1.
\]
Consequently, the preceding threshold bounds also give
\[
    \mu(T_n)-\frac12 =   O_P\left( n^{-(1-\nu)\rho_2/\alpha_2+\delta}\right), \qquad \sigma^2(T_n)\to\frac1{12}.
\]
Since $k=n^\nu$ and
\[
    \nu<\frac{2\rho_2}{\alpha_2+2\rho_2},
\]
$\delta>0$ can be chosen sufficiently small such that
\[
    \sqrt{k}\left|\mu(T_n)-\frac12\right|=o_P(1).
\]

Define the oracle estimator
\[
    \tilde\Gamma_{X_2\to X_1}
    :=
    \frac1k\sum_{i=1}^n
    F_1(X_{i,1})
    \mathbf 1\{X_{i,2}>T_n\}.
\]
Let $J_n:=\{i:X_{i,2}>T_n\}$. By continuity of $F_2$, $|J_n|=k$ almost surely. Write $J_n=\{j_{n,1},\ldots,j_{n,k}\}$ and define $U_{n,r}:=F_1(X_{j_{n,r},1}),\, r=1,\ldots,k$. Then
\[
    \tilde\Gamma_{X_2\to X_1}
    =
    \frac1k\sum_{r=1}^k U_{n,r}.
\]

Let $\mathcal A_n:=\sigma(T_n,J_n)$. Conditional on $\mathcal A_n$, the selected observations are i.i.d. with conditional
mean $\mu(T_n)$ and variance $\sigma^2(T_n)$.

With $s_n^2:=k\sigma^2(T_n)$, the $\mathcal{A}_n$-Lindeberg condition is, for every $\varepsilon>0$,
\[
    \frac1{s_n^2}
    \sum_{r=1}^k
    \e\left[
        (U_{n,r}-\mu(T_n))^2
        \mathbf 1\{
            |U_{n,r}-\mu(T_n)|>\varepsilon s_n
        \}
        \,\middle|\,\mathcal A_n
    \right]
    \xrightarrow{P}0,
\]
fulfilled, since $0\le U_{n,r}\le1$ and $\sigma^2(T_n)\xrightarrow{P}1/12$ and therefore $s_n=\sqrt{k}\sigma(T_n)\xrightarrow{P}\infty$.

Hence, by the Conditional Central Limit Theorem \cite[Theorem~1 and Corollary~3]{bulinski2017conditional},
\[
    \frac{
        \sqrt{k}
        \bigl(
            \tilde\Gamma_{X_2\to X_1}-\mu(T_n)
        \bigr)
    }{\sigma(T_n)}
    \xrightarrow{d}\mathcal{N}(0,1).
\]

Using $\sqrt{k}\bigl(\mu(T_n)-1/2\bigr)=o_P(1)$ and $\sigma^2(T_n)\xrightarrow{P}\frac1{12}$ Slutsky's theorem yields
\[
    \sqrt{k}\left(
    \tilde\Gamma_{X_2\to X_1}-\frac12
    \right)
    \xrightarrow{d}
    \mathcal{N}\left(0,\frac1{12}\right).
\]
It remains to replace $F_1$ by $\hat F_1$. Since exactly $k$ observations enter the estimator,
\[
\sqrt{k}
\left|
    \hat\Gamma_{X_2\to X_1}
    -
    \widetilde\Gamma_{X_2\to X_1}
\right|
\leq
\sqrt{k}\,
\|\hat F_1-F_1\|_\infty.
\]
By the Dvoretzky-Kiefer-Wolfowitz \citep{massart1990tight} inequality we get 
\[
\p\left(\sqrt{n}\underset{x}{\sup}|\hat{F}_1(x)-F_1(x)|>\lambda\right)\leq 2\exp(-2\lambda^2),
\]
which yields with $\lambda=\sqrt{0.5\log(2/\varepsilon)}$
\[
\p\left(\underset{x}{\sup}|\hat{F}_1(x)-F_1(x)|>\frac{\lambda}{\sqrt{n}}\right)\leq \varepsilon
\]
for every $\varepsilon>0$ and therefore
\[
\sqrt{k}\,\|\hat F_1-F_1\|_\infty =O_{\p}\left(\sqrt{\frac{k}{n}}\right)=o_{\p}(1).
\]
Hence, by Slutsky's theorem,

\[
\sqrt{k}\left(\hat{\Gamma}_{X_2\to X_1}-\frac12\right)\xrightarrow{d}\mathcal N\left(0,\frac1{12}\right).
\]

Exchanging the roles of $X_1$ and $X_2$ gives the corresponding result for $\hat\Gamma_{X_1\to X_2}$ with $\rho_1:=\min\{1,\alpha_H-\alpha_1\}$. Thus both results hold provided
\[
    \nu<
    \min_{j\in\{1,2\}}
    \frac{2\rho_j}{\alpha_j+2\rho_j}.
\]

If no confounder is present, $X_1$ and $X_2$ are independent under $H_0$. Hence the conditional mean and variance  equal $1/2$ and $1/12$ exactly. The same argument therefore applies for $k/n\to 0$ as $k,n\to \infty$.
\end{proof}

\subsection{Proof of Lemma \ref{lemma:empcopula}}

\begin{proof}
We prove the two statements separately.

\medskip
\noindent
(1) 
Since $U_2$ takes values in $[0,1]$, we may use the identity
\[
\e(U_2\mid U_1>1-t)
=
\int_0^1 \p(U_2>s\mid U_1>1-t)\,ds.
\]
By the definition of conditional probability,
\[
\p(U_2>s\mid U_1>1-t)
=
\frac{\p(U_2>s,\;U_1>1-t)}{\p(U_1>1-t)}.
\]
Because $U_1\sim\mathrm{Unif}(0,1)$, we have
\[
\p(U_1>1-t)=t.
\]
Moreover, by definition of the survival copula,
\[
\p(U_2>s,\;U_1>1-t)
=
\p(U_1>1-t,\;U_2>1-(1-s))
=
\bar C(t,1-s).
\]
Hence
\[
\e(U_2\mid U_1>1-t)
=
\frac{1}{t}\int_0^1 \bar C(t,1-s)\,ds.
\]
By substitution with $r=1-s$, this becomes
\[
\e(U_2\mid U_1>1-t)
=
\frac{1}{t}\int_0^1 \bar C(t,r)\,dr.
\]

Taking the limit as $t\downarrow 0$ proves the first claim.

\medskip
\noindent
(2)
Let
\[
R_{m,1}:=\sum_{\ell=1}^n \mathbf{1}\{X_{\ell,1}\le X_{m,1}\},
\qquad
R_{m,2}:=\sum_{\ell=1}^n \mathbf{1}\{X_{\ell,2}\le X_{m,2}\}
\]
denote the ranks of $X_{m,1}$ and $X_{m,2}$, respectively. Since the margins are
continuous, there are no ties almost surely, and therefore
\[
\hat F_{2}(X_{m,2})=\frac{R_{m,2}}{n},
\qquad
\mathbf{1}\{X_{m,1}>X_{(n-k),1}\}
=
\mathbf{1}\{R_{m,1}>n-k\}.
\]
Thus,
\[
\hat{\Gamma}_{X_1 \to X_2}
=
\frac{1}{kn}\sum_{m=1}^n R_{m,2}\mathbf{1}\{R_{m,1}>n-k\}.
\]
Using the identity
\[
R_{m,2}=\sum_{r=1}^n \mathbf{1}\{R_{m,2}\ge r\},
\]
we obtain
\[
\hat{\Gamma}_{X_1 \to X_2}
=
\frac{1}{kn}\sum_{m=1}^n\sum_{r=1}^n
\mathbf{1}\{R_{m,2}\ge r,\;R_{m,1}>n-k\}.
\]
Interchanging the order of summation yields
\[
\hat{\Gamma}_{X_1 \to X_2}
=
\frac{1}{k}\sum_{r=1}^n
\frac{1}{n}\sum_{m=1}^n
\mathbf{1}\{R_{m,2}\ge r,\;R_{m,1}>n-k\}.
\]

On the other hand, the empirical survival copula can be written as
\[
\hat{\bar C}_n(u_1,u_2)
=
\frac{1}{n}\sum_{m=1}^n
\mathbf{1}\{R_{m,1}>n(1-u_1),\;R_{m,2}>n(1-u_2)\},
\qquad u\in[0,1]^2.
\]
Now set $s=n-r+1$. Then
\[
R_{m,2}\ge r
\quad\Longleftrightarrow\quad
R_{m,2}>n-s,
\]
and hence
\[
\mathbf{1}\{R_{m,2}\ge r,\;R_{m,1}>n-k\}
=
\mathbf{1}\{R_{m,2}>n-s,\;R_{m,1}>n-k\}.
\]
Therefore,
\[
\frac{1}{n}\sum_{m=1}^n
\mathbf{1}\{R_{m,2}\ge r,\;R_{m,1}>n-k\}
=
\hat{\bar C}_n\left(\frac{k}{n},\frac{s}{n}\right).
\]
Since $r=1,\dots,n$ corresponds exactly to $s=1,\dots,n$, we conclude that
\[
\hat{\Gamma}_{X_1 \to X_2}
=
\frac{1}{k}\sum_{s=1}^n
\hat{\bar C}_n\left(\frac{k}{n},\frac{s}{n}\right).
\]
This proves the second claim.
\end{proof}

\subsection{Proof of Theorem \ref{theorem:causal}}

\begin{proof}
By Lemma~\ref{lemma:empcopula},
\[
\Gamma_{X_1\to X_2}(t)
=
\e\left(F_2(X_2)\mid F_1(X_1)>1-t\right)
=
\frac{1}{t}\int_0^1 \bar C(t,s)\,ds,
\]
and
\[
\hat{\Gamma}_{X_1\to X_2}
=
\frac{1}{k}\sum_{s=1}^n
\hat{\bar C}_n\left(\frac{k}{n},\frac{s}{n}\right),
\]
where $\hat{\bar C}_n$ denotes the empirical survival copula.

We decompose
\[
\sqrt{n}\Bigl(
\hat{\Gamma}_{X_1\to X_2}
-
\Gamma_{X_1\to X_2}(t)
\Bigr)
=
A_n+B_n+D_n,
\]
where
\[
A_n
:=
\sqrt{n}\left[
\frac{1}{t}\cdot\frac{1}{n}\sum_{s=1}^n
\hat{\bar C}_n\left(t,\frac{s}{n}\right)
-
\frac{1}{t}\cdot\frac{1}{n}\sum_{s=1}^n
\bar C\left(t,\frac{s}{n}\right)
\right],
\]
\[
B_n
:=
\sqrt{n}\left[
\frac{1}{t}\cdot\frac{1}{n}\sum_{s=1}^n
\bar C\left(t,\frac{s}{n}\right)
-
\frac{1}{t}\int_0^1 \bar C(t,s)\,ds
\right],
\]
and
\[
D_n
:=
\sqrt{n}\left[
\hat{\Gamma}_{X_1\to X_2}
-
\frac{1}{t}\cdot\frac{1}{n}\sum_{s=1}^n
\hat{\bar C}_n\left(t,\frac{s}{n}\right)
\right].
\]

We first consider $A_n$. Define the empirical survival copula process
\[
\mathbb G_n(u,v):= \sqrt{n}\bigl(\hat{\bar C}_n(u,v)-\bar C(u,v)\bigr), \qquad (u,v)\in[0,1]^2.
\]
Then
\[
A_n
=
\frac{1}{t}\cdot\frac{1}{n}\sum_{s=1}^n
\mathbb G_n\left(t,\frac{s}{n}\right).
\]

By weak convergence results for the empirical copula process, see
\citet{segers2012asymptotics}, the process $\mathbb G_n$ converges weakly in
$\ell^\infty([0,1]^2)$ to a centered Gaussian process $\mathbb G$ with almost
surely continuous sample paths. 
Define
\[
T_n(f):=\frac{1}{n}\sum_{s=1}^n f\!\left(\frac{s}{n}\right),
\qquad T(f):=\int_0^1 f(s)\,ds.
\]
Since $T_n(f)\to T(f)$ for every continuous $f$ by the convergence of Riemann
sums, the extended continuous mapping theorem \citep{van1996weak} yields
\[
A_n \xrightarrow{d} \frac{1}{t}\int_0^1 \mathbb G(t,s)\,ds,
\]
where the limit is centered Gaussian.

Since the survival copula $ \bar{C}(t,\cdot)$ is Lipschitz continuous with Lipschitz constant at most
$1$ \citep{nelsen2006introduction}, we obtain
\[
\begin{aligned}
\left|
\frac{1}{n}\sum_{s=1}^n \bar C\left(t,\frac{s}{n}\right)
-
\int_0^1 \bar C(t,u)\,du
\right|
&\le
\sum_{s=1}^n
\int_{(s-1)/n}^{s/n}
\left|
\bar C\left(t,\frac{s}{n}\right)-\bar C(t,u)
\right|\,du \\
&\le
\sum_{s=1}^n
\int_{(s-1)/n}^{s/n}
\left|
\frac{s}{n}-u
\right|\,du \\
&\le
\sum_{s=1}^n \int_{(s-1)/n}^{s/n} \frac{1}{n}\,du
=
\frac{1}{n}.
\end{aligned}
\]
This implies that the corresponding discretization error is of order
$O(n^{-1})$ and hence
\[
B_n
=
\sqrt{n}\,
O\!\left(\frac{1}{n}\right)
=
O\!\left(\frac{1}{\sqrt{n}}\right)
\to 0.
\]

It remains to control $D_n$. We can split $D_n$ as
\begin{align*}
D_n
&=
\sqrt{n}\left[\frac{1}{k}\sum_{s=1}^n \widehat{\bar C}_n \left( \frac{k}{n}, \frac{s}{n}\right)
    - \frac{1}{t}\frac{1}{n}\sum_{s=1}^n\widehat{\bar C}_n\left(t,\frac{s}{n}\right)\right]\\
&=
\sqrt{n}\left( \frac{1}{k}-\frac{1}{tn}\right)\sum_{s=1}^n\widehat{\bar C}_n\left(\frac{k}{n}, \frac{s}{n}\right)
+
\frac{\sqrt{n}}{tn}
\sum_{s=1}^n
\left[ \widehat{\bar C}_n\left(\frac{k}{n},\frac{s}{n} \right)
    -\widehat{\bar C}_n \left(t, \frac{s}{n}\right)\right]\\
&=:D_{n,1}+D_{n,2}.
\end{align*}

We show that both terms converge to zero. Since $k=\lfloor tn \rfloor$, write $\delta_n:=tn-k\in[0,1)$.
Then
\[
    \left| \frac{1}{k}-\frac{1}{tn} \right|
    = \frac{tn-k}{ktn}
    =\frac{\delta_n}{ktn} \leq \frac{1}{ktn}.
\]
Since $k/n\to t>0$, it follows that
\[
    \left| \frac{1}{k}-\frac{1}{tn} \right|
    =O(k^{-2})=O(n^{-2}).
\]
Moreover, $0\leq\widehat{\bar C}_n(u,v)\leq1$, and therefore
\begin{align*}
    |D_{n,1}|
    &\leq
    \sqrt{n}\,
    \left|
        \frac{1}{k}-\frac{1}{tn}
    \right|
    \sum_{s=1}^n
    \widehat{\bar C}_n
    \left(
        \frac{k}{n},
        \frac{s}{n}
    \right)
\\
    &\leq
    \sqrt{n}\,
    n\,
    O(n^{-2})
    =
    O(n^{-1/2})
    \to 0.
\end{align*}

For $D_{n,2}$, note that $k=\lfloor tn\rfloor$ implies
\[
    \left|t-\frac{k}{n}\right|<\frac{1}{n}.
\]
Since the margins are continuous, the corresponding empirical
thresholds differ by less than one rank. Hence, uniformly in
$v\in[0,1]$,
\[
    \left|
        \widehat{\bar C}_n\left(\frac{k}{n},v\right)
        -
        \widehat{\bar C}_n(t,v)
    \right|
    \leq \frac{1}{n}.
\]
Therefore,
\begin{align*}
    |D_{n,2}|
    &\leq
    \frac{\sqrt{n}}{tn}
    \sum_{s=1}^n
    \left|
        \widehat{\bar C}_n
        \left(\frac{k}{n},\frac{s}{n}\right)
        -
        \widehat{\bar C}_n
        \left(t,\frac{s}{n}\right)
    \right| \\
    &\leq
    \frac{\sqrt{n}}{tn}
    \sum_{s=1}^n \frac{1}{n}
    =
    \frac{1}{t\sqrt{n}}
    \to 0.
\end{align*}

Thus $D_n=D_{n,1}+D_{n,2}\to 0$.


Hence Slutsky's theorem implies
\[
\sqrt{n}\Bigl(
\hat{\Gamma}_{X_1\to X_2}
-
\Gamma_{X_1\to X_2}(t)
\Bigr)
\xrightarrow{d}
\mathcal N(0,\sigma_{\Gamma_1}^2)
\]
for some $\sigma_{\Gamma_1}^2\ge 0$.

The corresponding result for $\hat{\Gamma}_{X_2\to X_1}$ follows by the
same argument, exchanging the roles of $X_1$ and $X_2$. For the difference $
\hat{\Delta}_{X_1\to X_2}$ note that, by the representations above,
\[
\sqrt{n}\bigl(\hat{\Gamma}_{X_1\to X_2}-\Gamma_{X_1\to X_2}(t)\bigr)
=
\frac{1}{t}\int_0^1 \mathbb G(t,s)\,ds + o_p(1),
\]
\[
\sqrt{n}\bigl(\hat{\Gamma}_{X_2\to X_1}-\Gamma_{X_2\to X_1}(t)\bigr)
=
\frac{1}{t}\int_0^1 \mathbb G(s,t)\,ds + o_p(1),
\]
where $\mathbb G$ is again a centered Gaussian process. Hence
\[
\sqrt{n}\bigl(
\hat{\Delta}_{X_1\to X_2}-\Delta_{X_1\to X_2}(t)
\bigr)
=
\frac{1}{t}\int_0^1\{\mathbb G(t,s)-\mathbb G(s,t)\}\,ds + o_p(1)
\xrightarrow{d}
\mathcal N(0,\sigma_\Delta^2),
\]
for some $\sigma_\Delta^2\ge 0$.
\end{proof}

\subsection{Proof of Theorem \ref{theorem}}
\begin{proof}
Analogously to the proof of Theorem~\ref{theorem:independence}, we omit the structural coefficients. Thus, without loss of generality, we may write  $X_1=\epsilon_1,\, X_2=X_1+\epsilon_2$.
By Lemma~\ref{lem:independece},
\[
    \p(X_1\le x\mid X_2>u)  \sim  \p(X_1\le x),  \qquad u\to\infty.
\]

The remainder of the proof follows analogously to the proof of Theorem~\ref{theorem:independence}. In particular, the same rate arguments with $\rho=\min\{1,\alpha_1-\alpha_2\}$ yield the asymptotic normality of the oracle estimator, and the replacement of $F_1$ by $\hat F_1$ is again $o_{\p}(k^{-1/2})$ by the Dvoretzky-Kiefer-Wolfowitz inequality. Hence
\[
    \sqrt{k}\left(
        \hat\Gamma_{X_2\to X_1}-\frac12
    \right)
    \xrightarrow{d}    \mathcal{N}\left(0,\frac1{12}\right).
\]
\end{proof}

\section{Additional Causal Configurations}\label{app:examples}
\renewcommand{\thefigure}{B.\arabic{figure}} \setcounter{figure}{0}
\begin{figure}[H]
\centering

\begin{minipage}[t]{0.49\textwidth}
\centering
\caption*{\textbf{(G)} Common confounder $H$}
\begin{tikzpicture}[>=stealth, scale=0.75, transform shape]
    \node[draw, circle] (h) {$H$};
    \node[draw, circle, below left=1cm and 1.2cm of h] (x1) {$X_1$};
    \node[draw, circle, below right=1cm and 1.2cm of h] (x2) {$X_2$};

    \draw[->] (h) -- node[midway, above left] {$1$} (x1);
    \draw[->] (h) -- node[midway, above right] {$1$} (x2);
\end{tikzpicture}
\end{minipage}
\hfill
\begin{minipage}[t]{0.49\textwidth}
\centering
\caption*{\textbf{(H)} Causal connection and confounder $H$}
\begin{tikzpicture}[>=stealth, scale=0.75, transform shape]
    \node[draw, circle] (h) {$H$};
    \node[draw, circle, below left=1cm and 1.2cm of h] (x1) {$X_1$};
    \node[draw, circle, below right=1cm and 1.2cm of h] (x2) {$X_2$};

    \draw[->] (h) -- node[midway, above left] {$1$} (x1);
    \draw[->] (h) -- node[midway, above right] {$1$} (x2);
    \draw[->] (x1) -- node[midway, above] {$1$} (x2);
\end{tikzpicture}
\end{minipage}

\vspace{0.2em}

\begin{minipage}[t]{0.49\textwidth}
\centering
\renewcommand{\arraystretch}{1.2}
\begin{tabular}{c|cc}
Index Case & $\Gamma_{X_1\to X_2}$ & $\Gamma_{X_2\to X_1}$ \\
\hline
$\alpha_1 > \alpha_2$ & 0.75&  0.5 \\
$\alpha_1 = \alpha_2$ & 0.75&  0.75 \\
$\alpha_1 < \alpha_2$ & 0.75 & 1
\end{tabular}
\end{minipage}
\hfill
\begin{minipage}[t]{0.49\textwidth}
\centering
\renewcommand{\arraystretch}{1.2}
\begin{tabular}{c|cc}
Index Case & $\Gamma_{X_1\to X_2}$ & $\Gamma_{X_2\to X_1}$ \\
\hline
$\alpha_1 > \alpha_2$ & 1 & 0.5 \\
$\alpha_1 = \alpha_2$ & 1&  $\boldsymbol{\frac{19}{20}}$\\
$\alpha_1 < \alpha_2$ & 1 & 1
\end{tabular}
\end{minipage}

\vspace{0.3cm}
\begin{minipage}[t]{0.49\textwidth}
\centering
\caption*{\textbf{(J)} Common confounder $H$}
\begin{tikzpicture}[>=stealth, scale=0.75, transform shape]
    \node[draw, circle] (h) {$H$};
    \node[draw, circle, below left=1cm and 1.2cm of h] (x1) {$X_1$};
    \node[draw, circle, below right=1cm and 1.2cm of h] (x2) {$X_2$};

    \draw[->] (h) -- node[midway, above left] {$1$} (x1);
    \draw[->] (h) -- node[midway, above right] {$1$} (x2);
\end{tikzpicture}
\end{minipage}
\hfill
\begin{minipage}[t]{0.49\textwidth}
\centering
\caption*{\textbf{(I)} Causal connection and confounder $H$}
\begin{tikzpicture}[>=stealth, scale=0.75, transform shape]
    \node[draw, circle] (h) {$H$};
    \node[draw, circle, below left=1cm and 1.2cm of h] (x1) {$X_1$};
    \node[draw, circle, below right=1cm and 1.2cm of h] (x2) {$X_2$};

    \draw[->] (h) -- node[midway, above left] {$1$} (x1);
    \draw[->] (h) -- node[midway, above right] {$1$} (x2);
    \draw[->] (x1) -- node[midway, above] {$1$} (x2);
\end{tikzpicture}
\end{minipage}

\vspace{0.2em}

\begin{minipage}[t]{0.49\textwidth}
\centering
\renewcommand{\arraystretch}{1.2}
\begin{tabular}{c|cc}
Index Case & $\Gamma_{X_1\to X_2}$ & $\Gamma_{X_2\to X_1}$ \\
\hline
$\alpha_1 > \alpha_2$ & 1 & 0.75 \\
$\alpha_1 = \alpha_2$ & 0.75&  0.75 \\
$\alpha_1 < \alpha_2$ & 0.5 & 0.75
\end{tabular}
\end{minipage}
\hfill
\begin{minipage}[t]{0.49\textwidth}
\centering
\renewcommand{\arraystretch}{1.2}
\begin{tabular}{c|cc}
Index Case & $\Gamma_{X_1\to X_2}$ & $\Gamma_{X_2\to X_1}$ \\
\hline
$\alpha_1 > \alpha_2$ & 1 & $\frac{17}{18}$ \\
$\alpha_1 = \alpha_2$ & 1&  $\boldsymbol{\frac{19}{20}}$\\
$\alpha_1 < \alpha_2$ & 1 & 1
\end{tabular}
\end{minipage}

\caption{Additional causal configurations between $X_1$ and $X_2$ with a confounder $H$. (G) and (H) correspond to the case $\alpha_H=\alpha_1$, whereas panels (I) and (J) correspond to $\alpha_H=\alpha_2$. For illustration all causal weights are set to 1 and $\alpha_H=3$.}

\label{fig:cases_2}
\end{figure}

\section{Test for Different Tail Indices}\label{app:tail_test}
\renewcommand{\thetable}{C.\arabic{table}}
\setcounter{table}{0}
\renewcommand{\theequation}{C.\arabic{equation}} \setcounter{equation}{0}

Let $X$ and $Y$ be two random variables with corresponding tail indices $\alpha_1$ and $\alpha_2$. To assess whether the tail behavior of the two variables differs, we test the null hypothesis
\begin{equation*}
    H_0: \alpha_1=\alpha_2 \qquad H_1:\alpha_1\neq \alpha_2,
\end{equation*}
using the test statistic proposed by \cite{hoga2018detecting}. Since the Hill estimator, which is central in this procedure, does not deliver unbiased results in finite samples if the location parameter does not vanish, we proceed as in \cite{hoga2018detecting} and \cite{drees2003extreme} by shifting the data with constants $u_W$, $W\in\{X,Y\}$ such that the transformed observations approximately follow a Pareto distribution, chosen as zero, the 95\% quantile or the maximal value. This transformation is validated using a Pareto quantile plot.

\begin{algorithm}
\caption{Test for different tail indices}
For observations $X_1,\ldots,X_n$ and $Y_1,\ldots,Y_n$ of Variables $X,Y$ where an underlying GPD \eqref{eq:GPD} in the tails is assumed. Let $X_{\lfloor nt \rfloor-\lfloor kt \rfloor:\lfloor nt \rfloor}$ denote the $(\lfloor kt \rfloor+1)$-largest value of $X_1,\ldots,X_{\lfloor nt \rfloor}$. Define 
\[
    \hat{\gamma}^{(k)}_{H,X}(t) = \frac{1}{\lfloor kt \rfloor} \sum_{i=1}^{\lfloor kt \rfloor} \log\left( \frac{X_{\lfloor nt \rfloor - i : \lfloor nt \rfloor}}{X_{\lfloor nt \rfloor - \lfloor kt \rfloor : \lfloor nt \rfloor}} \right)
\]
with $t\in (0,1]$.

\begin{enumerate}
    \item \textbf{Shift}: Set shift $u_X$ and $u_Y$, so that $X+u_X$ and $Y+u_Y$ are approximately Pareto distributed
    \item \textbf{Set}: 
    \[
k^*_W = \arg\min_{k = k_{\min}, \dots, k_{\max}} \left[
\sup_{j=1, \dots, k_{\max}} \left| W_{n-j:n} - \left(k/j\right)^{\hat{\gamma}_{H,W}} \, W_{n-k:n}\right| \right],
\] 
with $k_{min}=\max(\lfloor 0.05 n\rfloor,50)$ and $k_{max}=\lfloor n^{0.8} \rfloor$, $W\in\{X,Y\}$ and set $k^*=\min(k^*_X, k^*_Y)$.
    \item \textbf{Test}:
    \begin{enumerate}
        \item Calculate Hill estimator for $X$ and $Y$ with $\hat{\gamma}^{(k^*)}_{H,X}(1)$ and $\hat{\gamma}^{(k^*)}_{H,Y}(1)$
        \item Calculate test statistic
        \[
        T_{I} = \frac{ \left[ \hat{\gamma}_{H,X}(1) - \hat{\gamma}_{H,Y}(1) \right]^2}
        {\displaystyle\int_{t_0}^{1} t^2 \left\{ \left[ \hat{\gamma}^{(k^*)}_{H,X}(t) - \hat{\gamma}^{(k^*)}_{H,Y}(t) \right] - \left[ \hat{\gamma}^{(k^*)}_{H,X}(1) - \hat{\gamma}^{(k^*)}_{H,Y}(1) \right] \right\}^2 \, dt}
        \]
\item If $T_{I}>55.44$ reject $H_0$ at the $5\%-$level
    \end{enumerate}
\end{enumerate}
\end{algorithm}

The value 55.44 is set according to \cite{hoga2018detecting}, who gives the values in \autoref{tab:teststatistic_quantiles} for different $\alpha-$levels and $t_0=0.2$.
\begin{table}[ht]
    \centering
    \begin{tabular}{llllll}
    \toprule
       $\alpha$-level  & $0.5\%$ & $1\%$ & $2.5\%$ & $5\%$ & $10\%$\\
        Quantile & $166.1$ & $126.6$ & $82.26$ & $55.44$ & $34.01$\\
        \bottomrule
    \end{tabular}
    \caption{Quantiles of test statistic $T_{I}$ for different $\alpha-$levels under $H_0$.}
    \label{tab:teststatistic_quantiles}
\end{table}

\newpage
We also consider a test for equality of tail indices based on iid observations and the asymptotic normality of the Hill estimator. Under suitable regularity conditions (see \cite{de1998asymptotic}),
\begin{equation*}
    \sqrt{k}(\hat{\gamma}-\gamma)\overset{D}{\rightarrow}\mathcal{N}(0,\gamma^2).
\end{equation*}
This yields the test statistic
\begin{equation}\label{eq:hill_test_stat}
    T_{\gamma}=\sqrt{k}\frac{\hat{\gamma}_1-\hat{\gamma}_2}{\sqrt{\hat{\gamma}_1^2+\hat{\gamma}_2^2}}.
\end{equation}
Under $H_0:\gamma_1=\gamma_2$ we have $T_\gamma\to N(0,1)$ and reject $H_0$ if $|T_{\gamma}|>z_{1-\alpha/2}$, where $z_{1-\alpha/2}$ denotes the $1-\alpha/2$-quantile of a standard normal distribution.

\autoref{tab:alternative_index_test} reports the outcomes of the simulation study from Section~\ref{subsec:tail_index_sim} but instead of using the test proposed by \cite{hoga2018detecting}, it is based on the test statistic $T_\gamma$.

\section{Additional Simulations and Figures}
\label{app:add_figures}
\renewcommand{\thefigure}{D.\arabic{figure}} \setcounter{figure}{0}
\renewcommand{\thetable}{D.\arabic{table}}
\setcounter{table}{0}

\subsection{Additional Simulation Studies}
\autoref{fig:add_simu_conv_pareto} shows the results of the convergence study from Section
\ref{sec:simulation} for Pareto distributed variables instead of
Student's $t$. \autoref{fig:add_simu_conv_mod10} presents the corresponding
results for Student's $t$ distributed variables under model
\eqref{eq:model2}. While the Pareto distribution yields a faster convergence
of the CTC estimator to the true value, the results for
model \eqref{eq:model2} do not differ noticeably from those for model
\eqref{eq:model1}.

\autoref{fig:add_simu_conv_t_with_conf_mod9} shows the change in behavior of the CTC if a heavy-tailed confounder (here $\alpha_H=2$) is present based on model \eqref{eq:model1} with $\beta_{H1}=\beta_{H2}=0.5$.

\autoref{tab:alternative_index_test} reports the results for the Tail Index-Test when
using the alternative test statistic from \eqref{eq:hill_test_stat}. We
observe results similar to \autoref{tab:index_test_hoga}, that the test has higher power when the tail indices differ more strongly.

In \autoref{fig:simulation_causality_test_mod9} and \autoref{fig:simulation_lingam_mod9}, we visualize the convergence behavior of the Causality-Test and LiNGAM under model~\eqref{eq:model1} with $\beta_{H1}=\beta_{H2}=0$. The results show that, when the causal effect is not confined to the tails of the distribution, LiNGAM performs better and achieves more accurate classification
across sample sizes than the Causality-Test, which relies only on extreme observations.

\label{app:add_simulation}
\begin{figure}[H]
    \centering
    \includegraphics[width=\textwidth]{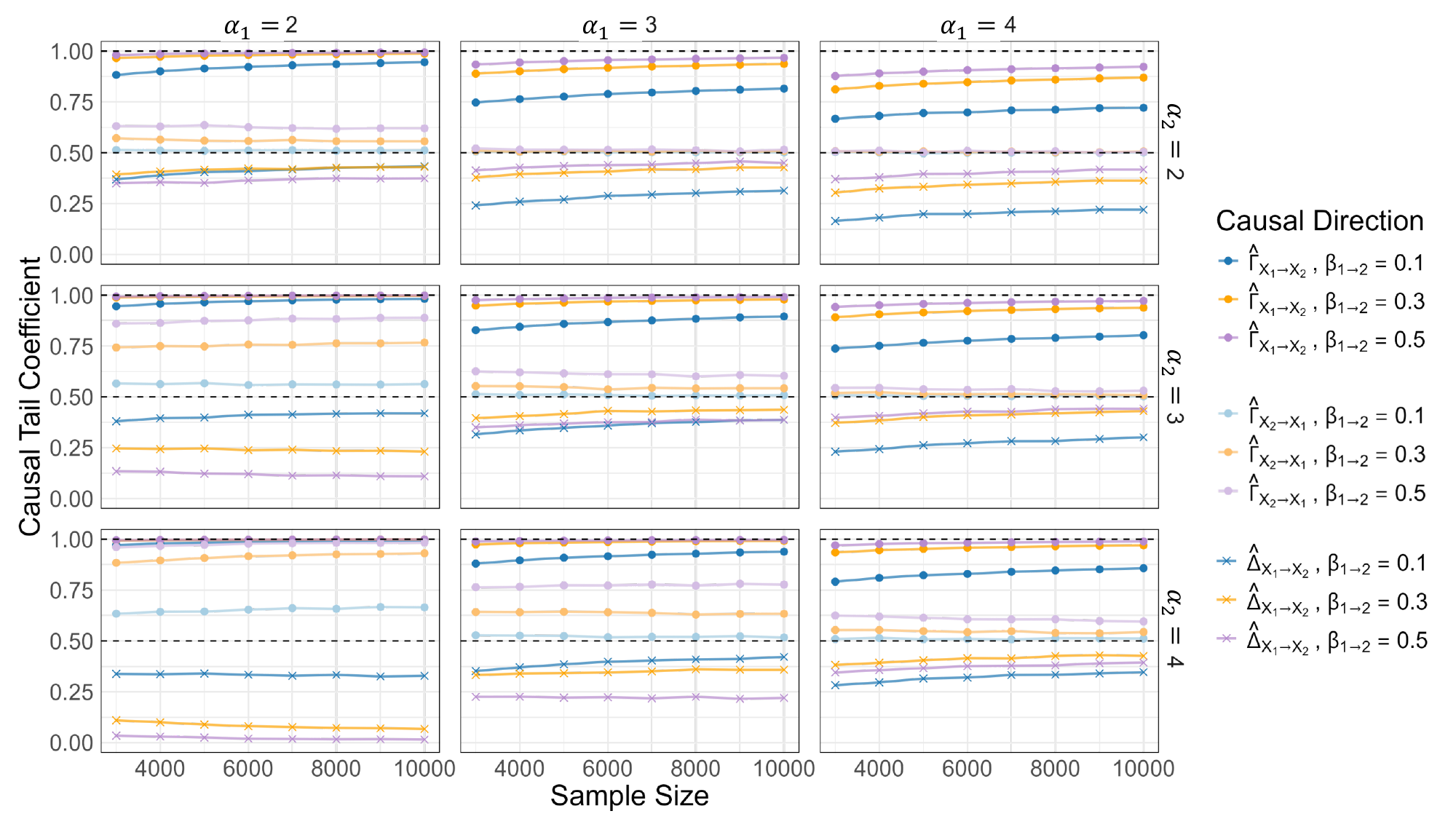}
    \caption{Estimations of the CTC under model \eqref{eq:model1} for different sample sizes, values of $\beta_{1 2}$ and tail indices $\alpha_1$ and $\alpha_2$, where $\epsilon_1$ and $\epsilon_2$ follow a Pareto distribution with tail indices $\alpha_1$ and $\alpha_2$.}
    \label{fig:add_simu_conv_pareto}
\end{figure}

\begin{figure}[H]
    \centering
    \includegraphics[width=\textwidth]{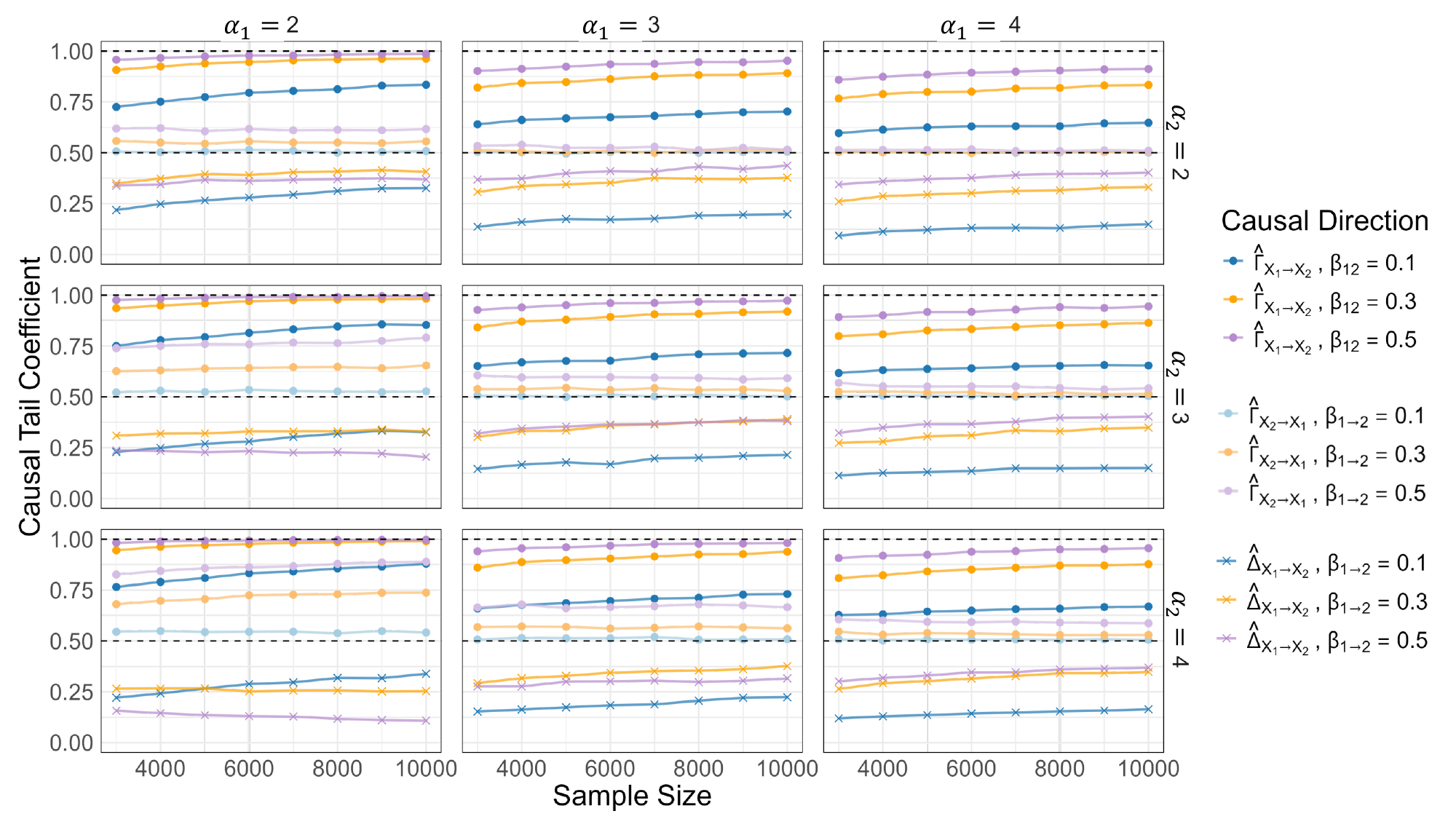}
    \caption{Estimations of the CTC under model \eqref{eq:model2} for different sample sizes, values of $\beta_{1 2}$ and tail indices $\alpha_1$ and $\alpha_2$, where $\epsilon_1$ and $\epsilon_2$ follow a Student's $t$ distribution with $\alpha_1$ and $\alpha_2$ degrees of freedom.}
    \label{fig:add_simu_conv_mod10}
\end{figure}

\begin{figure}[ht]
    \centering
    \includegraphics[width=\textwidth]{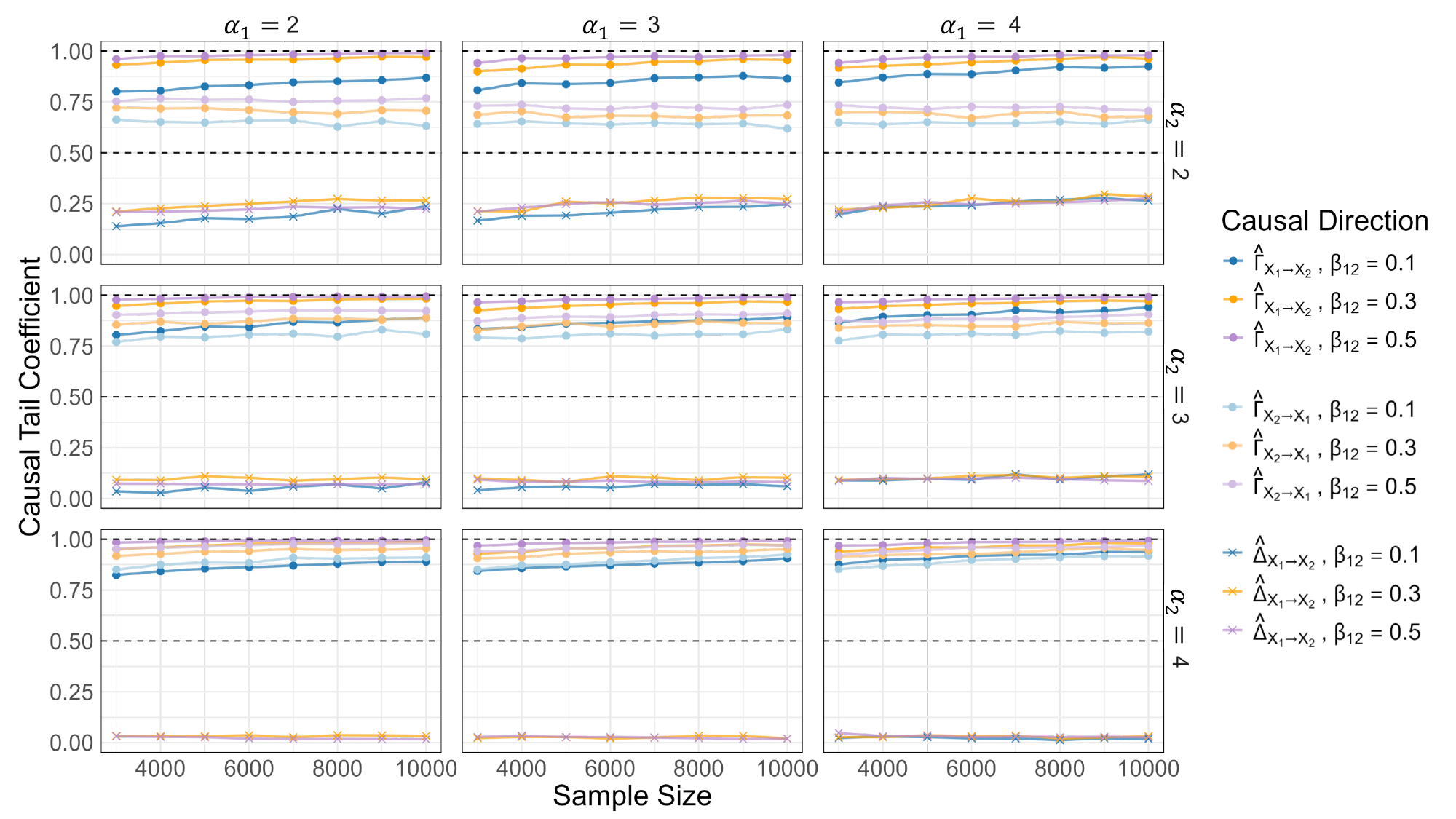}
    \caption{Estimations of the CTC in presence of a heavy-tailed confounder under model \eqref{eq:model1} for different sample sizes, values of $\beta_{1 2}$ and tail indices $\alpha_1$ and $\alpha_2$, where $\epsilon_1$ and $\epsilon_2$ follow a Student's $t$ distribution with $\alpha_1$ and $\alpha_2$ degrees of freedom. }
\label{fig:add_simu_conv_t_with_conf_mod9}
\end{figure}

\begin{figure}[H]
    \centering
    \includegraphics[width=\textwidth]{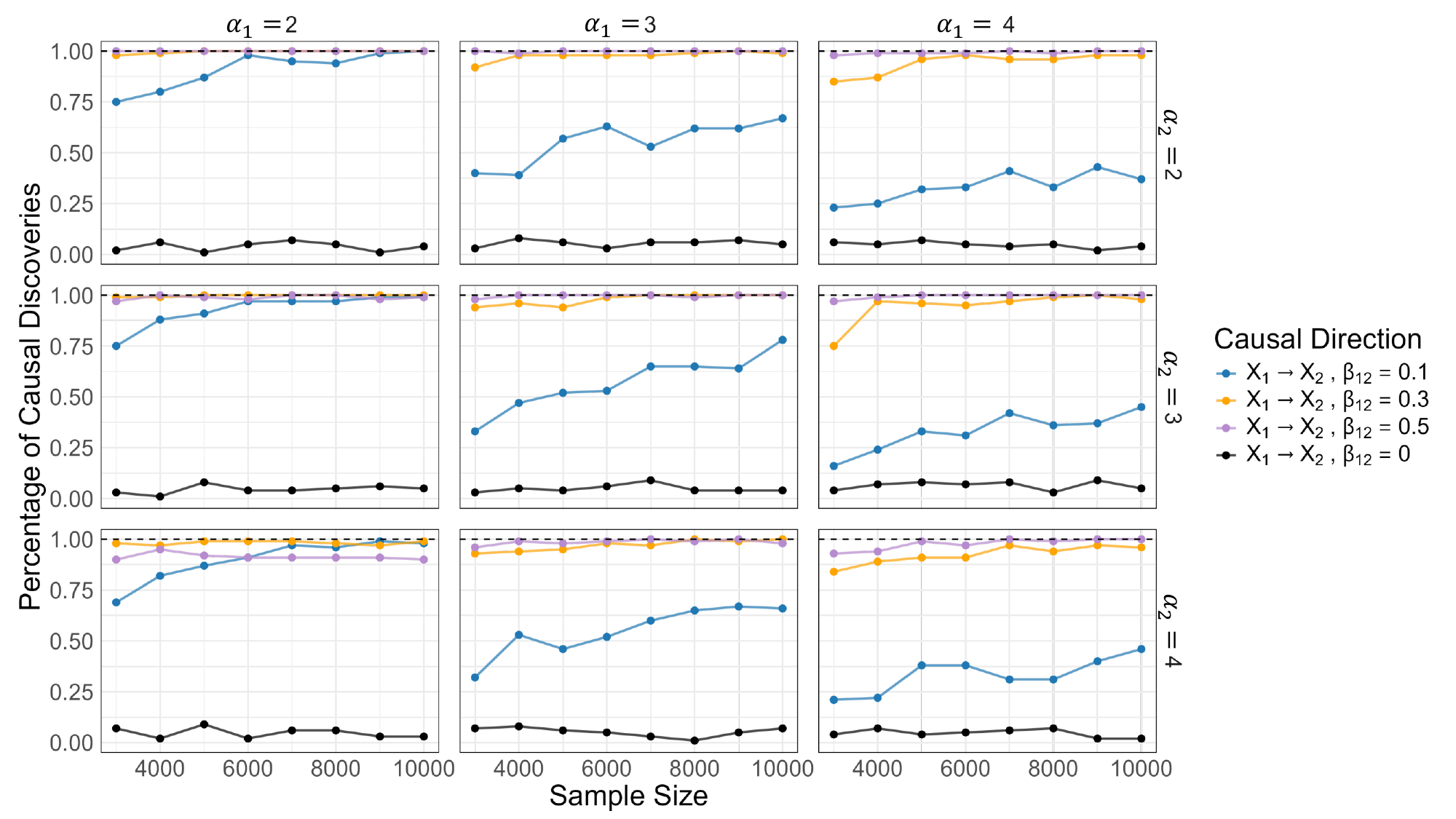}
    \caption{Percentage of causal discoveries for the Causality-Test under model \eqref{eq:model1}, based on the rejection rule $|\Delta_{X_1\to X_2}|> \Delta_{0.05}$, for different sample sizes, causal coefficients $\beta_{12}$, and tail indices $\alpha_1$ and $\alpha_2$, where $\epsilon_1$ and $\epsilon_2$ follow Student’s $t$ distributions with $\alpha_1$ and $\alpha_2$ degrees of freedom.}
    \label{fig:simulation_causality_test_mod9}
\end{figure}

\begin{figure}[H]
    \centering
    \includegraphics[width=\textwidth]{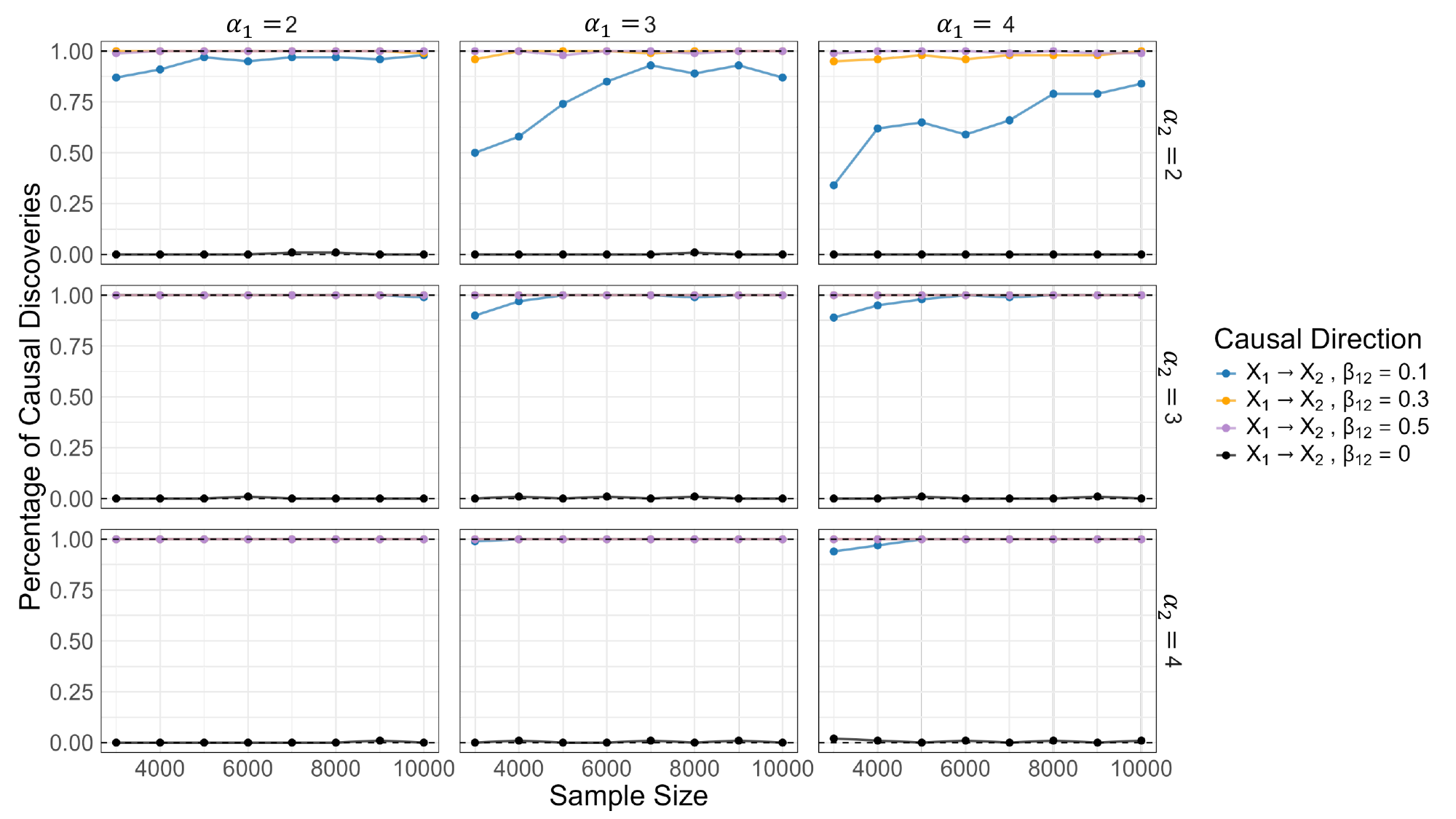}
    \caption{ Percentage of causal discoveries for LiNGAM under model \eqref{eq:model1} for different sample sizes, causal coefficients $\beta_{12}$, and tail indices $\alpha_1$ and $\alpha_2$, where $\epsilon_1$ and $\epsilon_2$  follow Student’s $t$ distributions with $\alpha_1$ and $\alpha_2$ degrees of freedom.}
    
    \label{fig:simulation_lingam_mod9}
\end{figure}

\begin{table}[ht]
    \centering
\begin{tabular}{cc|cc|cc|cc}
\toprule
\multicolumn{2}{c}{Tail indices} &
\multicolumn{2}{c}{no conf.} &
\multicolumn{2}{c}{light-tailed conf.} &
\multicolumn{2}{c}{heavy-tailed conf.} \\
$\alpha_1$ & $\alpha_2$ &
(A) & (B) & (C) & (D) & (E) & (F) \\
\midrule
\cellcolor{gray!20}2 & \cellcolor{gray!20}2 &  3.80 &  8.50 & 18.00 & 19.00 &   --   &   --   \\
2 & 3 & 99.60 & 99.70 & 96.90 & 98.50 &   --   &   --   \\
2 & 4 &100.00 &100.00 &100.00 &100.00 &   --   &   --   \\
\midrule
3 & 2 & 99.30 & 85.40 & 96.70 & 75.20 &   --   &   --   \\
\cellcolor{gray!20}3 & \cellcolor{gray!20}3 &  4.00 & 16.90 & 16.60 & 32.60 &  9.60 &  5.30 \\
3 & 4 & 70.30 & 91.00 & 62.30 & 85.10 & 34.80 & 22.50 \\
\midrule
4 & 2 &100.00 & 99.80 &100.00 & 98.90 &   --   &   --   \\
4 & 3 & 70.20 & 27.80 & 62.20 & 19.70 & 35.20 & 24.10 \\
\cellcolor{gray!20}4 & \cellcolor{gray!20}4 &  3.00 & 20.30 & 14.50 & 27.10 &  2.30 &  6.70 \\
\bottomrule
\end{tabular}
\caption{Empirical rejection rates (in \%) of the Tail Index-Test based on \eqref{eq:hill_test_stat} at the $5\%$ level for different tail index combinations ($\alpha_1,\alpha_2$) and configurations (A)-(F), where gray cells mark the equal index cases ($H_0$ true).}
\label{tab:alternative_index_test}
\end{table}

\clearpage
\subsection{Precipitation and Train Delays in Switzerland}

\begin{figure}[H]
    \centering
    \begin{subfigure}[b]{0.48\textwidth}
        \centering
        \includegraphics[width=\textwidth]{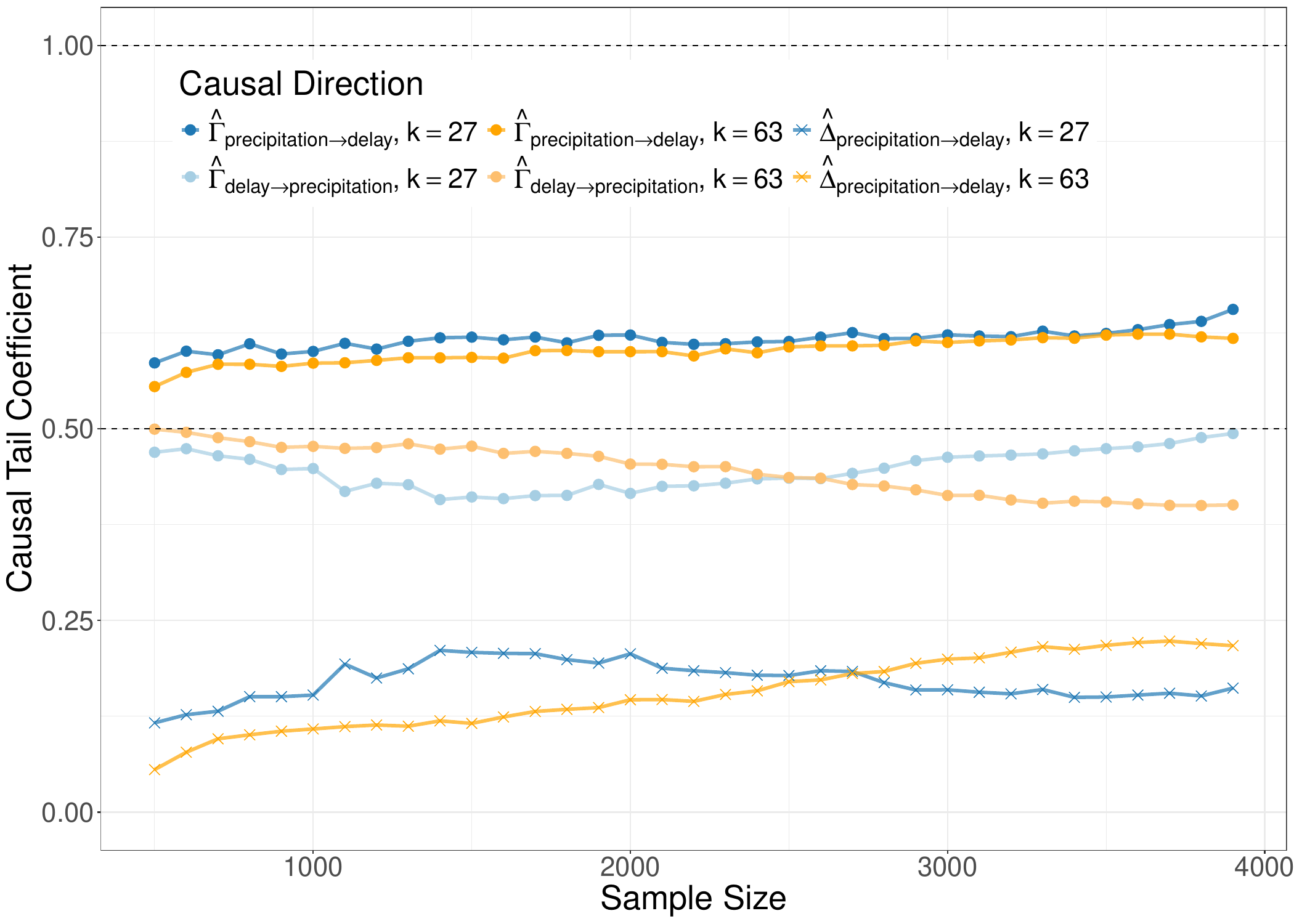}
    \end{subfigure}
    \begin{subfigure}[b]{0.48\textwidth}
        \centering
        \includegraphics[width=\textwidth]{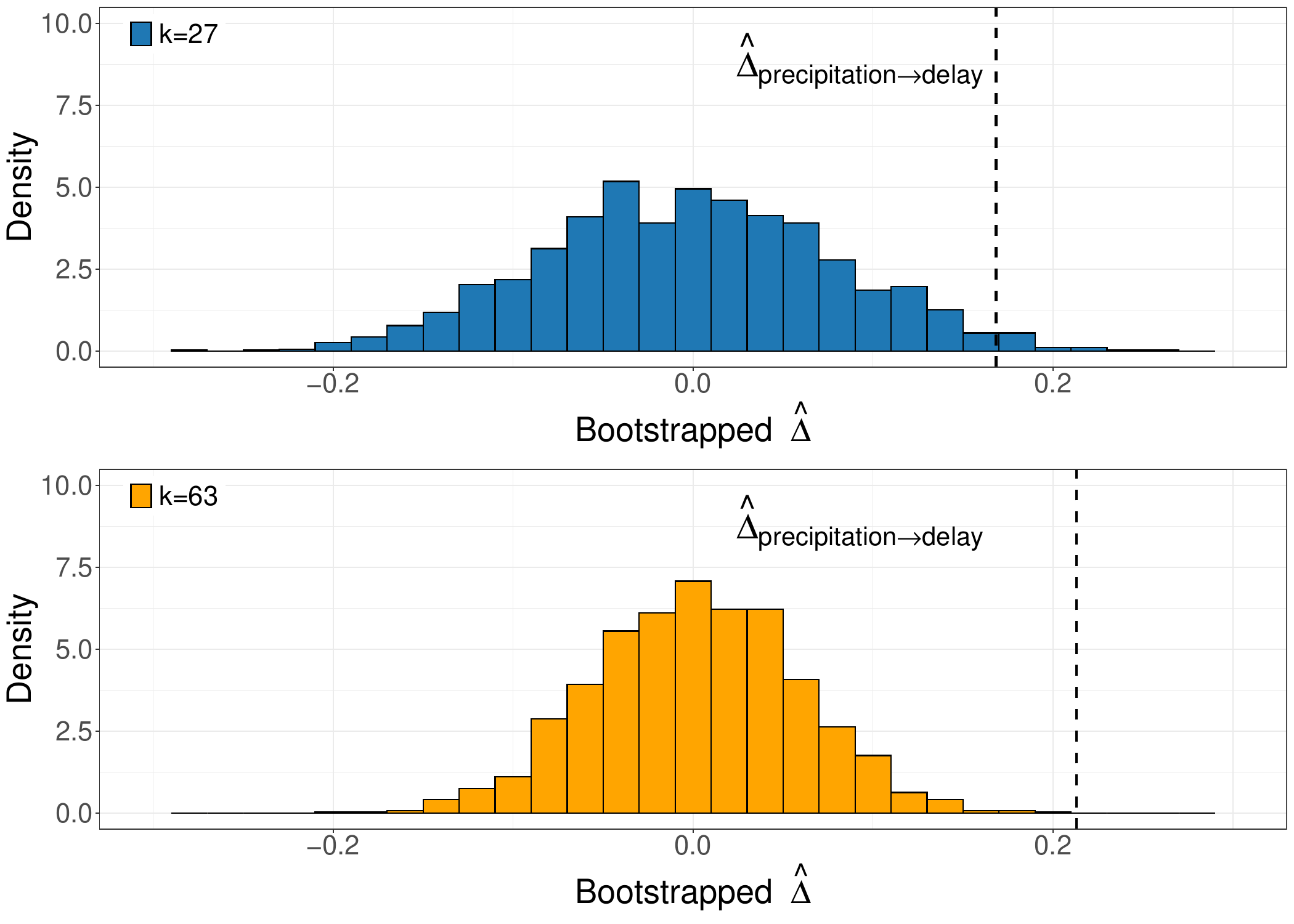}
     \end{subfigure}
    \caption{Estimates of the CTC for causality between current precipitation and train delays. The left panel shows
the convergence of the CTC estimator for $k\in \{27, 63\}$ and the right panel shows the bootstrap distributions.}
    \label{fig:add_delays}
\end{figure}

\subsection{Financial Stock Markets and Cryptocurrencies}
\begin{figure}[H]
    \centering
    \begin{subfigure}[b]{0.49\textwidth}
        \centering
        \includegraphics[width=\textwidth]{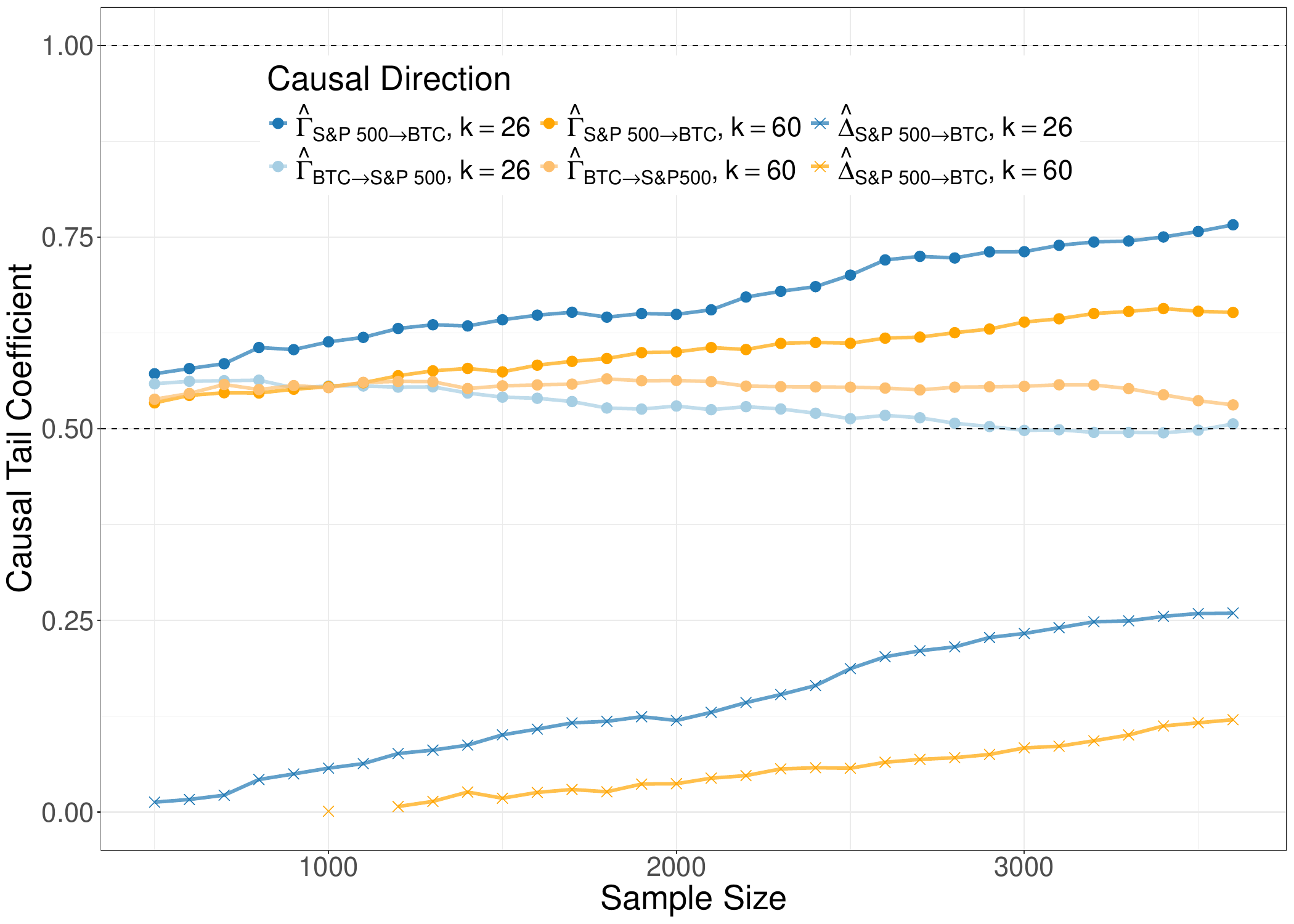}
    \end{subfigure}
    \begin{subfigure}[b]{0.49\textwidth}
        \centering
        \includegraphics[width=\textwidth]{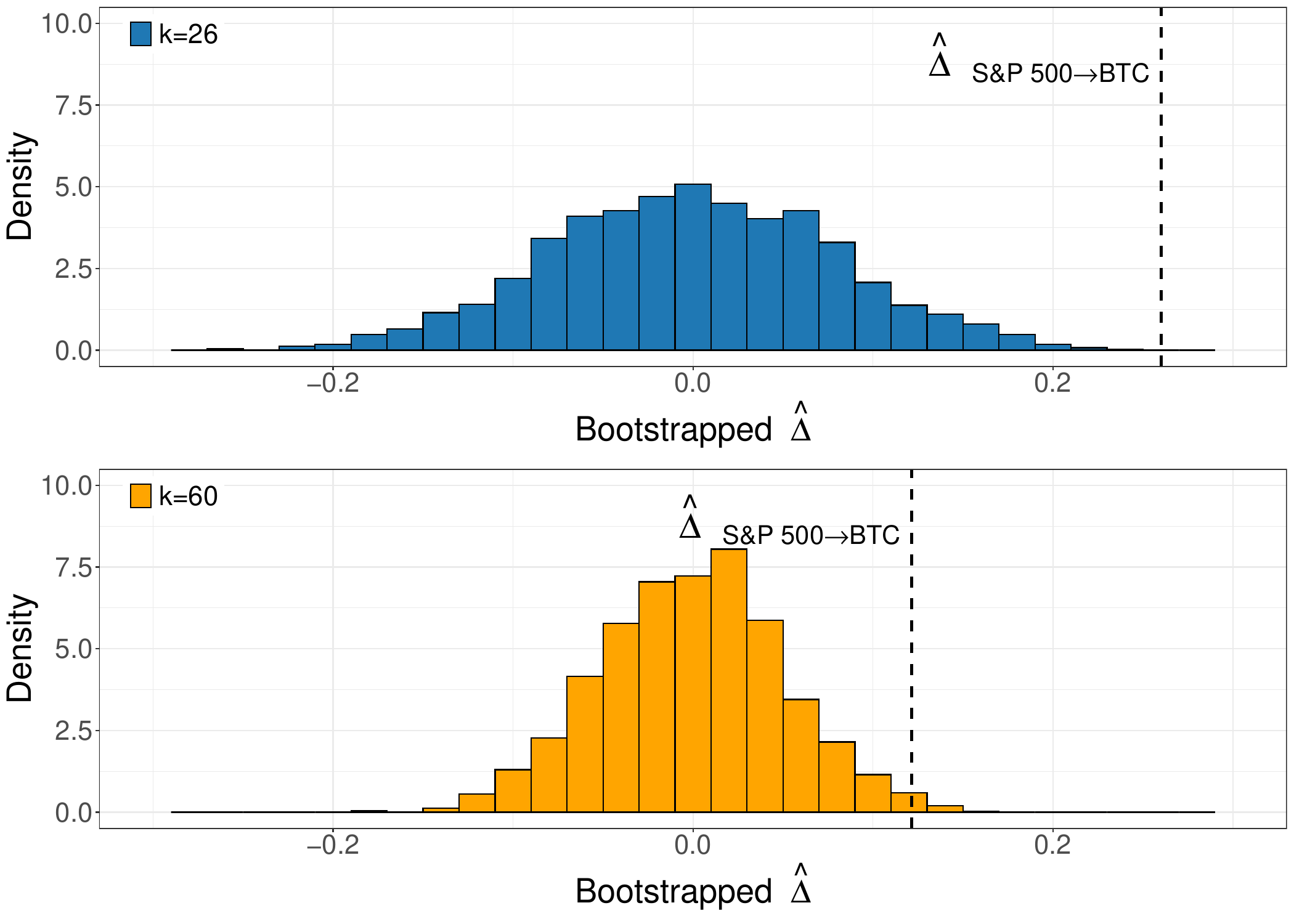}
    \end{subfigure}
    \vskip\baselineskip
    \caption{Estimates of the CTC for causality between S\&P 500 and Bitcoin for the right tails. The left panel shows the convergence of the estimator for $k\in\{26,60\}$ and the right panel the bootstrap samples of the Causality-Test and the estimated value of $\Delta_{S\&P\,500 \to BTC}$.}
    \label{fig:fiance_CTCs_right_tail}
\end{figure}

\end{document}